\documentclass[aps
,nofootinbib]{revtex4}
\usepackage{amsmath,amsthm,amsfonts,amscd,mathrsfs,pifont,bm}
\usepackage{eucal,color}
\usepackage{graphicx}

\newtheorem{thm}{Theorem}[section]
\newtheorem{lem}{Lemma}[section]

\numberwithin{equation}{section}
\numberwithin{thm}{section}
\newtheorem{definition}{Definition}[section]

\begin{document}
\title{Relaxation of a Simple Quantum Random Matrix Model}
\author{P. Vidal
}
\affiliation{Department of Applied Mathematics, H.I.T.- Holon Institute of Technology, Holon 58102, Israel}
\author{G. Mahler}
\affiliation{ Institut f\"ur Theoretische Physik 1 , Universit\"at Stuttgart
Pfaffenwaldring 57, D-70550 Stuttgart, Germany}
\begin{abstract}
We will derive here the relaxation behavior of a simple quantum random matrix model. The aim is to derive the effective equations which rise when a random matrix interaction is taken in the weak coupling limit. 
The physical situation this model represents is that a quantum particle restricted to move on two sites, where every site has $N$ possible energy states. The hopping from one site to another is then modeled by a random matrix. The techniques used here can be 
applied to many variations of the model.
\end{abstract}
\maketitle
\section{INTRODUCTION}
\label{sec:intro}
We will derive from the Schr\"odinger equation  an effective equation which will turn out to be a rate equation. 
The Hamiltonian is taken to have a deterministic part plus a weak random part. 
The statistics of these types of models have been investigated in \cite{BHZ95}, \cite{BZ95} and the dynamics have been numerically investigated in \cite{SGM06}, \cite{MMGM05} and \cite{BSG08} in the context of the emergence of Fourier's law and  statistical relaxation in closed quantum systems.
Random matrices are used in many situation either to model a complex system or mimic quantum chaos.
On the other hand rate equations are widely used in order to model some complicated non-equilibrium situation by a simple set of differential equations.
The essence of the results here is thus the emergency of these simple equations from complex quantum or quantum chaotic systems.
The random matrix here represents somehow the "complexity".
We take a simple random matrix model to illustrate how to treat fully the random interaction but 
this type of analysis can be forwarded to more complicated models with structures of various kinds. We will comment on this latter on.
The model here is that of a quantum particle that can only move between two sites which we denote site $1$ and site $2$.
Each site has $N$ energy levels and for simplicity we take them to be equidistant. These energies are taken to be bounded between $0$ and $1$ and as $N$ increase they grow nearer to each other. The particle can then hop from one energy level of one site to the next with a random amplitude. 
Pictorially it can be represented as in figure  \ref{fig1}.
\begin{figure}[h]
\label{fig1}
\begin{center}
\resizebox{0.5\columnwidth}{!}{  \includegraphics[width=.2\textwidth]{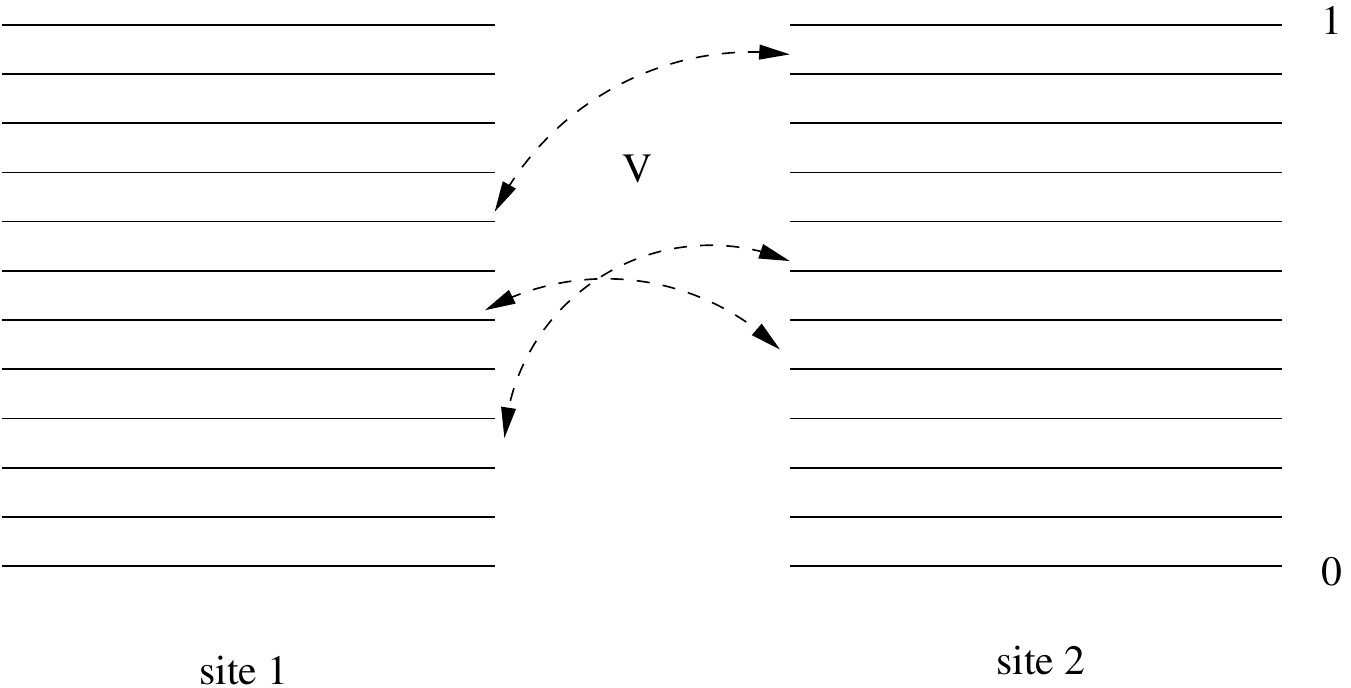}}
\caption{Hopping particle}
\end{center}
\end{figure}
If $P^{1}_{t}$ is the probability to be on site $1$ and $P^{2}_{t}$ that of being on site $2$ we will prove that they satisfy on average 
\begin{align}
 \frac{d}{dt}P^{1}_{t}&=-4\pi\left(P^{1}_{t}-P^{2}_{t}\right) \notag \\
 \frac{d}{dt}P^{2}_{t}&=-4\pi\left(P^{2}_{t}-P^{1}_{t}\right) \notag
\end{align}
in certain limits.
Similar models were introduced and studied in  \cite{Weg78}, \cite{SGM06} and \cite{Leb} where the emergency of diffusion and relaxation behavior was discusses.
The methods we use are those used in \cite{Spo77}, \cite{EY00}, \cite{ES}.
We will expand the formal solution of the Schr\"odinger equation in powers of the  random interaction and then average over the product of random matrices (sections \ref{sec:duhamel}, \ref{sect:AveragingAndGraphs}).
This average will be equal to a sum graph dependent functions and, in the limits considered, we will show that some graphs yield a vanishing contribution  (sections \ref{sec:crossing}, \ref{sec:nested} and \ref{sec:simple}). The remaining graphs can then be summed over again and a solution to a rate equation is found (section \ref{sec:effective}).
Section \ref{sec:error1} is devoted to showing that the error in the limits considers tends to zero.
\subsection{THE MODEL}
\label{sec:model}
The model we analyze here is a two site tight binding model. At every site the particle has $N$ possible states to be in, each one with different energy, $E$. Our Hilbert space is then spanned by the vectors $|x,E\rangle$ where $x$ refers to the site and can take on the values $1$ or $2$. For simplicity we take the spectrum to be equidistant and we also take it to be bounded between $0$ and $1$. Thus the spectrum consists of the points $\{\frac{1}{N}, \frac{2}{N}, \dots \frac{N-1}{N}, 1\}$
Our unperturbed Hamiltonian is the following:
\begin{align}
\hat{H}_{0} & =  \sum_{x=1}^{2} \hat{H}_{0}^{x} \\
\hat{H}_{0}^{x} & =  \sum_{n=1}^{N} E_{n} |x,E_{n}\rangle \langle x,E_{n}| 
\end{align}
with $E_{n}=\frac{n}{N}$.
Our density of states is thus constant.
The perturbation is given by a type of GUE matrix. Each matrix entry is a complex gaussian distributed random variable. We restrict the interaction to be between energy states of different sites.
\begin{align}
\hat{V} & = \sum_{n,m=1}^{N} V_{1}(n,m) |1,E_{n}\rangle \langle 2,E_{m}| + V_{2} (n,m) |2,E_{n}\rangle \langle 1,E_{m}| \label{eq:interactionM}
\end{align}
$V$ has thus two off diagonal blocks while the rest is zero.
The distribution over this type of random matrix is then
\begin{align}
 P(V)= & \frac{1}{Z}e^{-\frac{N}{2} \text{Tr}\left[V^{2}\right]} \notag\\
Z= & \left(\frac{2N}{\pi}\right)^{N^{2}} \notag
\end{align}
We have then the following for the average on a matrix element:
\begin{align}
 \langle z \bar{z}\rangle = \frac{1}{N}
\end{align}
Our total Hamiltonian is:
\begin{align}
\hat{H} & =\hat{H}_{0}+\lambda \hat{V} \label{eq:Hamiltonian}
\end{align}
We will be interested in calculating the time evolution of the probability of the particle to be on site $1$ or $2$. These are 
\begin{align}
 \hat{P}^{1}=&\sum_{E_{0}}|1,E_{0}\rangle\langle 1,E_{0}| \label{p1}\\
 \hat{P}^{2}=&\sum_{E_{0}}|2,E_{0}\rangle\langle 2,E_{0}| \label{p2}
\end{align}
The theorem is then as follows:
\begin{thm}
 \label{thm:Main}
Say $|\psi_{t}^{N}\rangle$ is a solution to the Schr\"odinger equation with the Hamiltonian of Eq. (\ref{eq:Hamiltonian}) with initial data $|\psi_{0}^{N}\rangle$.
The initial data is taken such that the population around the edges of the spectrum of $H_{0}$ is zero in a small neighborhood of distance $\epsilon$ of the edges. That is  $\langle x_{0},E |\psi_{0}^{N}\rangle=0$ if $E\leq\epsilon$ or $1-\epsilon\leq E$.
Then in the limit $N\rightarrow \infty$ and $t\rightarrow \infty$ (taken in this order), and with the following scaling
\begin{align}
 \lambda^2t=&T <\infty \label{sca1}
\end{align}
the average over the random matrix of the time evolution of the probabilities (\ref{p1}) and (\ref{p2})
will follow the next differential equations:
\begin{align}
  \frac{d}{d T}P^{1}_{T}& =-4\pi \left(P^{1}_{T}-P^{2}_{T}\right) \\
  \frac{d}{d T}P^{2}_{T}& =-4\pi \left(P^{2}_{T}-P^{1}_{T}\right) 
 \end{align}
The initial data is given by 
\begin{align}
P^{1}_{0}=&\lim_{N\rightarrow \infty} \langle \psi_{0}^{N}|\hat{P}^{1}|\psi_{0}^{N}\rangle \notag \\
P^{2}_{0}=&\lim_{N\rightarrow \infty} \langle \psi_{0}^{N}|\hat{P}^{2}|\psi_{0}^{N}\rangle \notag
\end{align}
Eq. (\ref{sca1}) is called the Van Hove limit.
\end{thm}
%%%%%%%%%%%%%%%%%%%%%%%%%%%%%%%%%%%%%%%%%%%%%%%%%%%%%%%%%%%%%%%%%%%%%
%%%%%%%%%%%%%%%%%%%%%%%%%%%%%%%%%%%%%%%%%%%%%%%%%%%%%%%%%%%%%%%%%%%%
%%%%%%%%%%%%%%%%%%%%%%%%%%%%%%%%%%%%%%%%%%%%%%%%%%%%%%%%%%%%%%%%%
%%%%%%%%%%%%%%%%%%%%%%%%%%%%%%%%%%%%%%%%%%%%%%%%%%%%%%%%%%%%%%%%%
\section{EXPANSION AND IDENTITIES} 
\label{sec:duhamel}
According to the Duhamel formula, \cite{EY00}, we have the next identity for the evolution operator.
\begin{align}
e^{-iHt}& =e^{-iH_{0}t}-i\lambda \int_{0}^{t}e^{-iH(t-s)}Ve^{-iH_{0}s}ds \label{eq:duhamel}
\end{align} 
By applying successively this identity we can expand the evolution operator in orders of $\lambda$.
Thus we can write it as follows:
\begin{align}
e^{-iHt}=&\sum_{n=0}^{M}(-i\lambda)^n\Gamma_{n}(t) +(-i\lambda)^{M+1}\tilde{\Gamma}_{M+1}(t)
\label{eq:gammaprop} \\
|\psi_{t}\rangle=&\sum_{n=0}^{M}|\psi^{n}_{t}\rangle +|\phi^{M+1}_{t}\rangle
\label{eq:psiprop}
\end{align} 
with
\begin{align}
\Gamma_{n}(t) & = \int_{0}^{t}\int_{0}^{t-s1}\dots \int_{0}^{t-\sum_{j=1}^{n-1}s_{j}}ds_{1} \dots ds_{n}
e^{-iH_{0}(t-\sum_{j=1}^{n}s_{j})}Ve^{-iH_{0}s_{n}}V  \dots  e^{-iH_{0}s_{1}} \notag \\
& = \int_{0}^{t}\int_{0}^{t}\dots \int_{0}^{t}ds_{0} \dots ds_{n}
e^{-iH_{0}s_{0}}Ve^{-iH_{0}s_{1}} \dots e^{-iH_{0}s_{n}}\delta(t-\sum_{j=0}^{n}s_{j}) \label{eq:gamma}\\
\tilde{\Gamma}_{M+1}(t) & =\int_{0}^{t}\int_{0}^{t}\dots \int_{0}^{t}ds_{0} \dots ds_{M+1}
e^{-iHs_{0}}Ve^{-iH_{0}s_{1}} \dots e^{-iH_{0}s_{M+1}}\delta(t-\sum_{j=0}^{M+1}s_{j}) \notag\\
&= \int_{0}^{t} ds e^{-i(t-s)H}V\Gamma_{M}(s)\label{eq:gammatilde}
\end{align}
$|\phi^{M+1}_{t}\rangle$ is the  error term of the time evolved wave function.
We adopt the following notation for multiple time integrals: 
\begin{align}
 \int_{0}^{t}\dots \int_{0}^{t} ds_{0}\dots ds_{n} \delta(t-\sum_{i=0}^{n}s_{i})= \int [ds_{n}] \notag
\end{align}
We are interested in calculating the time evolution of the observables $\hat{P}^{1}$ and $\hat{P}^{2}$ given by Eqs. (\ref{p1}) and (\ref{p2}).
Using the expansion of Eq. (\ref{eq:gammaprop}) until the $M^{\text{th}}$ order, the time evolution of the observables $P^{x_{0},N}_{t}$ is
\begin{align}
 P^{x_{0},N}_{t}=& P^{x_{0},M,N}_{t}+R^{x_{0},M,N}_{t} \label{eq:PplusR}
\end{align}
with
\begin{align}
P^{x_{0},M,N}_{t} & =\sum_{E_{0}}\sum_{n,m=0}^{M} (i\lambda)^{m}(-i\lambda)^{n}
\langle \psi_{0}| \Gamma_{m}^{\dagger}(t) |x_{0},E_{0}\rangle \langle x_{0},E_{0}|\Gamma_{n}(t)|\psi_{0}\rangle \notag \\
& = \sum_{n,m=0}^{M} \lambda^{m+n} i^{m}(-i)^{n}
\sum_{E_{0},E_{n},E'_{m}=1}^{N} 
 \psi_{0}^{*}(x'_{m},E'_{m}) \psi_{0}(x_{n},E_{n}) 
 \langle x'_{m},E'_{m}| \Gamma_{m}^{\dagger}(t) |x_{0},E_{0}\rangle \langle x_{0},E_{0}|\Gamma_{n}(t)|x_{n},E_{n} \rangle
\label{eq:wignerlfullexpansion}
\end{align}
 $R^{x_{0},M,N}_{t}$ encodes the remainder of the evolution. It is our error term in the evolution of the probability $P^{x_{0},N}_{t}$ and has the following form:
\begin{align}
 R^{x_{0},M,N}_{t}= \sum_{n=0}^{M}\langle \psi^{n}_{t}|\hat{P}^{x_{0}}|\phi^{M+1}_{t}\rangle
+\sum_{n=0}^{M}\langle \phi^{M+1}_{t}|\hat{P}^{x_{0}}|\psi^{m}_{t}\rangle
+\langle \phi^{M+1}_{t}|\hat{P}^{x_{0}}|\phi^{M+1}_{t}\rangle \label{eq:Rest}
\end{align}
We will compute $P^{x_{0},M,N}_{t}$ in the limit $N\rightarrow \infty$ and the Van Hove limit,
$t\xrightarrow{\lambda^2t=T<\infty}\infty$.
In section \ref{sec:error1} we will show that remainder goes to $0$. 
\begin{align}
\lim_{M\rightarrow \infty}\lim_{t\rightarrow \infty}\lim_{N\rightarrow \infty} \left|R^{x_{0},M,N}_{t}\right|=0
\end{align}
This implies that we can use $P^{x_{0},M,N}_{t}$ to obtain the evolution of $P^{x_{0}}_{t}$. That is, 
\begin{align}
\lim_{t\rightarrow \infty}\lim_{N\rightarrow \infty} P^{x_{0},N}_{t}=& \lim_{M\rightarrow \infty}\lim_{t\rightarrow \infty}\lim_{N\rightarrow \infty} P^{x_{0},M,N}_{t} \notag
\end{align}
When inserting Eq. (\ref{eq:gamma}) in $\langle x_{0},E_{0}|\Gamma_{n}(t)|x_{n},E_{n}\rangle$  and identities after each interaction term we obtain the following:
\begin{align}
(-i)^{n}\langle x_{0},E_{0}|\Gamma_{n}(t)|x_{n},E_{n}\rangle & = 
\prod_{i=1}^{n-1} \sum_{E_{i},x_{i}} 
K^{n}(t,\{E_{i}\})L^{n}(\{x_{i},E_{i}\})
\end{align}
with
\begin{align} 
K^{n}(t,\{E_{i}\}) & = (-i)^{n} \int [ds_{n}] e^{-iE_{0}s_{0}} e^{-iE_{1}s_{1}} \dots e^{-iE_{n}s_{n}}
\label{firstK}\\
L^{n}(\{x_{i},E_{i}\}) & =
 \langle x_{0},E_{0} |V|x_{1},E_{1}\rangle
 \langle x_{1},E_{1} |V|x_{2},E_{2}\rangle 
 \dots \langle x_{n-1},E_{n-1} |V|x_{n},E_{n}\rangle
\end{align}
We denote by $\{x_{i},E_{i}\}$ the set of all energy variables, $\{E_{0}\dots E_{n}\}$, and position variables, $\{x_{0},\dots x_{n}\}$.
We have then for $P^{x_{0},M,N}_{t}$:
\begin{align}
P^{x_{0},M,N}_{t}
 = & \sum_{n,m=0}^{M}\lambda^{n+m} 
\sum_{\{E_{i},E'_{j}\}} \sum_{\{x_{i},x'_{j}\}_{0}}\psi_{0}^{*}(x'_{m},E'_{m}) \psi_{0}(x_{n},E_{n})
 K^{n}(t,\{E_{i}\})\bar{K}^{m}(t,\{E'_{i}\})  L^{n}(\{x_{i},E_{i}\}) \bar{L}^{m}(\{x'_{i},E'_{i}\})
\label{eq:explicitP} 
\end{align}
where we have taken up  the following notation :
\begin{align}
 \prod_{i=0}^{n}\sum_{E_{i}}\prod_{j=1}^{m}\sum_{E'_{j}} &=\sum_{\{E_{i},E'_{j}\}} \label{sum1}\\
 \prod_{i=1}^{n}\sum_{x_{i}}\prod_{j=1}^{m}\sum_{x'_{j}}&=\sum_{\{x_{i},x'_{j}\}_{0}} \label{sum2}
\end{align}
The subscript $0$ in Eq. (\ref{sum2}) denotes the fact that we are not summing over $x_{0}$. From Eq. (\ref{eq:wignerlfullexpansion}) we see we have $E_{0}'=E_{0}$ and $x_{0}=x'_{0}$.
The $L^{n}(\{x_{i},E_{i}\})$ function  is the statistical weight given to this history or process by the random interaction. It carries no time dependency.
Since we want to calculate the average and  the randomness is all encoded in the $L^{n}\bar{L}^{m}$ factor we will calculate $\mathbb{E}[\bar{L}^{m}(\{x'_{i},E'_{i}\})L^{n}(\{x_{i},E_{i}\})]$. The purpose of the next section is to characterize this average.
%%%%%%%%%%%%%%%%%%%%%%%%%%%%%%%%%%%%%%%%%%%%%%%%%%%%%%%%%%%%%%%%%%
%%%%%%%%%%%%%%%%%%%%%%%%%%%%%%%%%%%%%%%%%%%%%%%%%%%%%%%%%%%%%%%%%
\section{AVERAGING AND GRAPHS}
\label{sect:AveragingAndGraphs}
The main purpose in this section is to introduce graphs as representations of contributions to the average we want to calculate such that averaging will turn out to be a sum over different graphs. We will introduce three classes of graphs just as in \cite{EY00}. First we differentiate between crossing and non-crossing graphs. Non-crossing graphs are also called planar graphs, \cite{TH2,BIP+78}. They carry individually more weight then crossing graphs and thus crossing graphs will vanish in the limit $N\rightarrow \infty$.
We recall here Wick's theorem.
\begin{thm}
\label{thm:wicks}
Say we have $2k$ random Gaussian variables denotes by $X_{i}$, $1\leq i \leq 2k$, and say we have $Y=X_{1}X_{2} \dots X_{2k}$. Denote by  $\pi(2k)$ a list of pairs of all the elements of the set s, $s=(1,2 \dots,2k)$. We  have then
\begin{align}
\mathbb{E}[Y]=\sum_{\pi(2k)} \prod_{(i,j)\in \pi(2k)} \mathbb{E}[X_{i}X_{j}]
\label{eq:wignerpairs}
\end{align}
where $(i,j)$ refers to a pair of the list $\pi(2k)$. $\pi(2k)$ thus defines a graph on the set $\tilde{s}=\{X_{1},\dots X_{2k}\}$. 
\end{thm}
We set $X_{E_{i},E_{j}}(x_{i},x_{j})=\langle x_{i},E_{i}|V|x_{j},E_{j}\rangle$ which allows us to write the product of random variables \\ $\bar{L}^{m}(\{x'_{j},E'_{j}\})L^{n}(\{x_{i},E_{i}\})$ as
\begin{align}
\bar{L}^{m}(\{x'_{i},E'_{i}\})L^{n}(\{x_{i},E_{j}\})
& = X_{E'_{m},E'_{m-1}}(x'_{m},x'_{m-1})\dots X_{E'_{1},E_{0}}(x'_{1},x_{0}) X_{E_{0},E_{1}}(x_{0},x_{1})\dots X_{E_{n-1},E_{n}}(x_{n-1},x_{n}) \label{eq:primednonprimedvar} \\
& = X_{1} \dots X_{m} X_{m+1}\dots X_{n+m} 
\end{align}
We now apply theorem \ref{thm:wicks} to $\mathbb{E}[(\bar{L}^{m}(\{x'_{j},E'_{j}\}))L^{n}(\{x_{i},E_{i}\})]$.
\begin{align}
\mathbb{E}\left[(\bar{L}^{m}(\{x'_{j},E'_{j}\}))L^{n}(\{x_{i},E_{i}\})\right] & =
\sum_{\pi(n,m)}\prod_{(l,p)\in \pi(n,m)} \mathbb{E}[X_{l}X_{p}] \notag \\
& =\sum_{\pi(n,m)} C_{\pi(n,m)}(\{E_{i},E'_{j}\},\{x_{i},x'_{j}\})
\label{eq:two}
\end{align}
with 
\begin{align}
C_{\pi(n,m)}(\{E_{i},E'_{j}\},\{x_{i},x'_{j}\}) & = \prod_{(l,p)\in \pi(n,m)} \mathbb{E}[X_{l}X_{p}] \label{eq:GraphFunction}
\end{align}
and
\begin{align}
\mathbb{E}[X_{E'_{i+1},E'_{i}}(x'_{i+1},x'_{i})X_{E_{j},E_{j+1}}(x_{j},x_{j+1})] & =\frac{1}{N}  \delta_{E'_{i+1},E_{j+1}} \delta_{E'_{i},E_{j}} \delta_{x'_{i+1},x_{j+1}} \delta_{x'_{i},x_{j}}  \label{eq:OnePair} \\
\mathbb{E}[X_{E_{i},E_{i+1}}(x_{i},x_{i+1})X_{E_{j},E_{j+1}}(x_{j},x_{j+1})] & =\frac{1}{N}  \delta_{E_{i},E_{j+1}} \delta_{E_{i+1},E_{j}} \delta_{x_{i},x_{j+1}} \delta_{x_{i+1},x_{j}}  \label{eq:OnePairp}
\end{align} 
$\pi(n,m)$ is then a list of pairs of the set $\bar{s}$, or a graph, and $(l,p)$ is a pair of the list.
We say the order of a graph on $\bar{s}$ is the length of the set $\bar{s}$.  The order of $\pi(n,m)$ is then $n+m$.
We call $C_{\pi(n,m)}(\{E_{i},E'_{j}\},\{x_{i},x'_{j}\})$ the graph function.
Notice that by Eqs. (\ref{eq:OnePair}) and (\ref{eq:OnePairp})  the graph function can be split into a graph function depending on $\{E_{i},E'_{j}\}$ times a graph function depending on $\{x_{i},x'_{j}\}$.
\begin{align}
 C_{\pi(n,m)}(\{E_{i},E'_{j}\},\{x_{i},x'_{j}\})= C^{1}_{\pi(n,m)}(\{E_{i},E'_{j}\})C^{2}_{\pi(n,m)}(\{x_{i},x'_{j}\}) \label{eq:2Cs}
\end{align}
where $C^{1}_{\pi(n,m)}(\{E_{i},E'_{j}\})$ is a product of the $\delta$ functions in $E_{i}$ and $E'_{j}$ divided by $N^{\frac{n+m}{2}}$ and 
$C^{2}_{\pi(n,m)}(\{x_{i},x'_{j}\})$ is a product of the $\delta$ functions in $x_{i}$ and $x'_{j}$.
We make the following definitions to classify the possible graphs: 
\begin{definition}
Say we have 
a graph, $\pi(n,m)$, on $\bar{s}$ where 
\begin{align}
 \bar{s}=\{X_{1},\dots , X_{n+m}\}
\end{align}
\begin{itemize}
 \item[1]The average of a pair $X_{i}X_{j}$ is called a inner contraction if $i,j\leq m$ or when $i,j> m$. \\
It will be called an outer contraction if $i\leq m$ and $j>m$ or when  $i> m$ and $j\leq m$. 
 \item[2]The average of a pair $X_{i}X_{j}$ is called a next neighboring contraction (nn-contraction) if $j=i+1$.
 \item[3]If we have a contraction, between $X_{i}$ and $X_{j}$, and a contraction between $X_{k}$ and $X_{l}$ and $i<k<j<l$ then we call this a crossing.
 \item[4]If we have a contraction, between $X_{i}$ and $X_{j}$, and a contraction between $X_{k}$ and $X_{l}$ and, $i<k<l<j\leq m$ or $m<i<k<l<j$ then we call this a nest.
\end{itemize}
\label{def:contractions}
\end{definition}
From definition \ref{def:contractions} we see that in Eq. (\ref{eq:GraphFunction}) we have an inner contraction whenever both random variables are dependent on primed variables or when both depend on non primed variables. 
With these definitions we make three classes of graphs as in \cite{EY00}.
\begin{definition}
\label{def:graphs}
Say we have %the set $\bar{s}$ and 
a graph $\pi(n,m)$ on $\bar{s}$.
\begin{itemize}
 \item[C. G.] We say the graph is a crossing  graph (c-graph) if it possesses at least one crossing and we call a graph a non-crossing graph (nc-graph) if it possesses no crossings.
The set of all c-graphs of order $(n,m)$ is denoted by $\mathcal{G}_{2}(n,m)$ and the set of all c- graphs is denoted by $\mathcal{G}_{2}$. An example of a crossing graph is the previous Fig. \ref{figure2}.
 \item[N.G.] We say the graph is a nested graph (n-graph) if it is a nc-graph and possesses at least one nest and call a graph a non-nested graph (nn-graph) if it possesses none.
The set of all n-graphs of order $(n,m)$ is denoted by $\mathcal{G}_{1}(n,m)$ and the set of all n-graphs is denoted by $\mathcal{G}_{1}$. An example of a nested graph is shown in \ref{figure3} 
 \item[S. G.] We say the graph is a simple graph (s-graph) if it is a nc and nn-graph. 
The set of all s-graphs of order $(n,m)$ is denoted by $\mathcal{G}_{0}(n,m)$ and the set of all s-graphs is denoted by $\mathcal{G}_{0}$. 
An example of a simple graph is shown in Fig. (\ref{figure4}).
\end{itemize}
\end{definition}
\begin{figure}[h]
\begin{center}
\includegraphics[width=15cm]{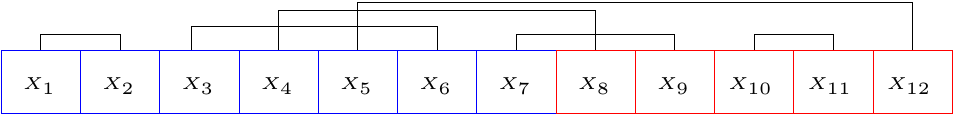}
\end{center}
\caption{Example of a crossing graph}
\label{figure2}       
\end{figure}
\begin{figure}[h]
\begin{center}
\includegraphics[width=15cm]{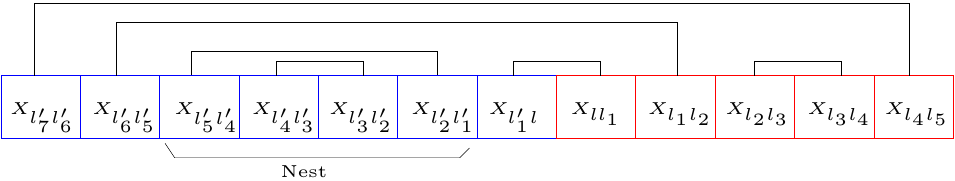}
\end{center}
\caption{Example of a nested graph}
\label{figure3}       
\end{figure}
\begin{figure}[h]
\begin{center}
\includegraphics[width=15cm]{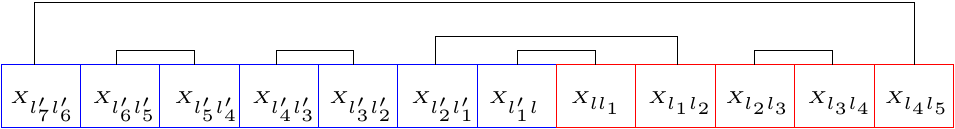}
\end{center}
\caption{Example of a simple graph}
\label{figure4}       
\end{figure}

Notice that simple graphs are build from next neighboring contractions and outer contractions since no crossing nor nests are allowed. The total number of graphs of order $n+m$ is $\frac{\left(n+m\right)!}{\frac{n+m}{2}!2^{n+m}}$ while the number of non crossing graphs is equal to the $\left(\frac{n+m}{2}\right)^{\text{th}}$ catalan number. This one is bounded by $C^{n+m}$, where $C$ is a constant.
From the definitions of simple, nested and crossing graphs we note that these classes are mutually exclusive and cover the set of all possible graphs. We have then the following identity:
\begin{align}
\sum_{\pi(n,m)}=\sum_{\pi(n,m)\in \mathcal{G}_{0}}+\sum_{\pi(n,m)\in \mathcal{G}_{1}}+\sum_{\pi(n,m)\in \mathcal{G}_{2}}
\end{align}
Thus from Eq. (\ref{eq:two})
\begin{align}
\mathbb{E}[(\bar{L}^{m}(\{x'_{j},E'_{j}\}))L^{n}(\{x_{i},E_{i}\})]
& =\sum_{a=0}^{2}\sum_{\pi(n,m)\in \mathcal{G}_{a} }C_{\pi(n,m)}(\{E_{i},E'_{j}\},\{x_{i},x'_{j}\})
\label{eq:twoprime}
\end{align}
From Eq. (\ref{eq:GraphFunction}),  (\ref{eq:OnePair}) and (\ref{eq:OnePairp}) we see that a graph function of order $n+m$  is a product of Kronecker delta functions divided by $N^{\frac{n+m}{2}}$. Thus not all the energy variables and position variables are independent. When summing up over these variables the more of these that are independent the larger the sum will become. 
This motivates  the following definitions:
\begin{definition}
\label{def:ABK}
For a graph function $C_{\pi(n,m)}\left(\{E_{i},E'_{j}\},\{x_{i},x'_{j}\}\right)$ we define: \\
$\mathcal{A}_{\pi(n,m)}=\{
\text{The set of independent variables of the set $\{E_{i},E'_{j}\}$ given by the graph function $C_{\pi(n,m)}$}
\} $\\
$\mathcal{B}_{\pi(n,m)}=\{
\text{The set of dependent variables of the set $\{E_{i},E'_{j}\}$  given by the graph function  $C_{\pi(n,m)}$}
\} $\\
$\kappa_{\pi(n,m)}=$\{Number of independent variables we have of the set $\{E_{i},E'_{j}\}$ given by the graph function $C_{\pi(n,m)}$\}
\end{definition}
With these definitions we have then
\begin{align}
 \sum_{\{E_{j},E'_{i}\}}=\sum_{\mathcal{A}_{\pi}}\sum_{\mathcal{B}_{\pi}}
\end{align}
From the definitions \ref{def:ABK} and Eq. (\ref{eq:2Cs}) we have 
\begin{align}
 \sum_{\{E_{j},E'_{j}\}\in\mathcal{B}_{\pi(n,m)}} C^{1}_{\pi(n,m)}\left(\{E_{j},E'_{j}\}\right) =\frac{1}{N^{\frac{n+m}{2}}} \label{eq:sumdependentvar}
\end{align}
With these definitions we can estimate certain sums of graphs functions.
We first prove the following property of the graph function:
\begin{thm}
\label{thm:endsmeet}
For any graph function $C_{\pi(n,m)}(\{E_{i},E'_{i}\},\{x_{i},x'_{i}\})$ we have 
\begin{align}
 C_{\pi(n,m)}(\{E_{i},E'_{j}\},\{x_{i},x'_{j}\})\propto \delta_{E'_{m},E_{n}}\delta_{x'_{m},x_{n}}
\end{align}
\end{thm}
\begin{proof}{Theorem \ref{thm:endsmeet}} \\
First we relabel the $n+m+1$ energy variables $\{E_{i},E'_{j}\}$ to $\{E_{i}\}$ and the position variables $\{x_{i},x'_{j}\}$ to $\{x_{i}\}$. Thus in the set $\{E_{i}\}$ and $\{x_{i}\}$ the index $i$ runs from $0$ to $n+m$.
Then Eqs. (\ref{eq:OnePair}) become
\begin{align}
\mathbb{E}[X_{E_{i},E_{i+1}}(x_{i},x_{i+1})X_{E_{j},E_{j+1}}(x_{j},x_{j+1})] & =\frac{1}{N}  \delta_{E_{i},E_{j+1}} \delta_{E_{i+1},E_{j}} \delta_{x_{i},x_{j+1}} \delta_{x_{i+1},x_{j}}  \label{eq:OnePair2}
\end{align}
where $i$ and $j$ can take on the values $0,\dots n+m-1$.
Thus the from the graph function we have that for each $i$ there is a unique $j$ such that
\begin{align} 
E_{i}=E_{j+1}\\
E_{i+1}=E_{j}
\end{align}
Note that $j,i\neq n+m$. 
Since for each $i$ there is a unique $j$ when summing over $i$ form $0$ to $n+m-1$ and summing also the corresponding and unique $j$ 
we obtain 
\begin{align} 
\sum_{i=0}^{n+m-1}E_{i}+\sum_{j=0}^{n+m-1}E_{j}&=\sum_{i=0}^{n+m-1}E_{i+1}
\sum_{j=0}^{n+m-1}E_{j+1}
\end{align}
and so 
\begin{align}
 E_{0}=E_{n+m}
\end{align}
Following the same reasoning for the $x_{i}$ variables we have 
\begin{align}
 x_{0}=x_{n+m}
\end{align}
This implies then that for the set $\{E_{i},E'_{j}\}$ we have $E_{n}=E'_{m}$ and for the set $\{x_{i},x'_{j}\}$ we have $x_{n}=x'_{m}$.
\end{proof}
From this we have
\begin{align}
\mathbb{E}[(\bar{L}^{m}(\{x'_{j},E'_{j}\}))L^{n}(\{x_{i},E_{i}\})] & =
\sum_{\pi(n,m)} C_{\pi(n,m)}(\{E_{i},E'_{j}\},\{x_{i},x'_{j}\}) \notag \\
& \propto \delta_{E'_{m},E_{n}}\delta_{x'_{m},x_{n}} \label{eq:proptodelta}
\end{align}
We now turn to proving the following essential theorem.
\begin{thm}
\label{thm:kappa}
If $\pi(n,m)$ is a non crossing graph then $\kappa_{\pi(n,m)}=\frac{n+m}{2}+1$.
If $\pi(n,m)$ is a  crossing graph then $\kappa_{\pi(n,m)}\leq\frac{n+m}{2}-1$.
\end{thm}
\begin{proof}{Theorem\ref{thm:kappa}}\\
From Eq. (2.12) we have that
\begin{align}
\sum_{\{x_{i},x'_{j}\}}\sum_{\{E_{i},E'_{j}\}} \bar{L}^{m}\left(\{x'_{j},E'_{j}\}\right) L^{n}\left(\{x_{i},E_{i}\}\right)
 \delta_{E_{n},E'_{m}}\delta_{x_{n},x'_{m}}=\text{Tr}\left[V^{n+m}\right]
\end{align}
 Therefore after averaging we get
\begin{align}
\sum_{\{x_{i},x'_{j}\}}\sum_{\{E_{i},E'_{j}\}} \sum_{\pi(n,m)} C_{\pi(n,m)}(\{E_{i},E'_{j}\},\{x_{i},x'_{j}\})
 \delta_{E_{n},E'_{m}}\delta_{x_{n},x'_{m}}=\mathbb{E}\left[\text{Tr}\left[V^{n+m}\right]\right]
\end{align}
From Eq. (\ref{eq:proptodelta}) we have
\begin{align}
\sum_{\{x_{i},x'_{j}\}}\sum_{\{E_{i},E'_{j}\}} \sum_{\pi(n,m)} C_{\pi(n,m)}(\{E_{i},E'_{j}\},\{x_{i},x'_{j}\})
=\mathbb{E}\left[\text{Tr}\left[V^{n+m}\right]\right]
\end{align}
Thus
\begin{align}
& \mathbb{E}\left[\text{Tr}\left[V^{n+m}\right]\right]\notag\\
&=\sum_{\pi(n,m)\in \mathcal{G}_{0,1}} 
\sum_{\{x_{i},x'_{j}\}}\sum_{\{E_{i},E'_{j}\}}  C_{\pi(n,m)}(\{E_{i},E'_{j}\},\{x_{i},x'_{j}\})
+
\sum_{\pi(n,m)\in \mathcal{G}_{2}} 
\sum_{\{x_{i},x'_{j}\}}\sum_{\{E_{i},E'_{j}\}}  C_{\pi(n,m)}(\{E_{i},E'_{j}\},\{x_{i},x'_{j}\})
\end{align}
By the definitions of dependent and independent variables we have 
\begin{align}
\sum_{\{x_{i},x'_{j}\}}\sum_{\{E_{i},E'_{j}\}}  C_{\pi(n,m)}(\{E_{i},E'_{j}\},\{x_{i},x'_{j}\})
& =
\sum_{\mathcal{A}_{\pi(n,m)}}
\sum_{\mathcal{B}_{\pi(n,m)}}  C_{\pi(n,m)}(\{E_{i},E'_{j}\},\{x_{i},x'_{j}\})
\notag\\
 & = \sum_{ \mathcal{A}_{\pi(n,m)}}\frac{1}{N^{\frac{n+m}{2}}}\notag\\
&=\frac{N^{\kappa_{\pi}}}{N^{\frac{n+m}{2}}}
\end{align}
where we have applied Eq. (\ref{eq:sumdependentvar}) in going from the first to the second line.
Thus
\begin{align}
 \mathbb{E}\left[\text{Tr}\left[V^{n+m}\right]\right]&=\sum_{\pi(n,m)\in \mathcal{G}_{0,1}} 
\frac{N^{\kappa_{\pi}}}{N^{\frac{n+m}{2}}}
+\sum_{\pi(n,m)\in \mathcal{G}_{2}} \frac{N^{\kappa_{\pi}}}{N^{\frac{n+m}{2}}}\label{eq:NewResult}
\end{align}
We know by \cite{TH2} that the leading contribution in $N$ of the average of the trace of a product of random matrices comes from non crossing graphs (planar diagrams), $\mathcal{G}_{0}$ and $\mathcal{G}_{1}$. That is:
\begin{align}
 \mathbb{E}\left[\text{Tr}\left[V^{n+m}\right]\right]= N \sum_{\pi(n+m)\in \mathcal{G}_{0,1}} 1+O(N^{-1}) \label{eq:OldResult}
\end{align}
Comparing (\ref{eq:OldResult}) and (\ref{eq:NewResult})
we see we must have $\kappa_{\pi(n,m)}=\frac{n+m}{2}+1$ for non crossing graphs. In order for the contribution of crossing graphs to be of the order of $N^{-1}$ or less we must have for crossing graphs $\kappa_{\pi(n,m)}\leq \frac{n+m}{2}-1$.
%\qed
\end{proof}
%%%%%%%%%%%%%%%%%%%%%%%%%%%%%%%%%%%%%%%%%%%%%
%%%%%%%%%%%%%%%%%%%%%%%%%%%%%%%%%%%%%%%%%%%%%
\section{ANALYSIS OF PROPAGATORS}
\label{sec:analisis}
In this section we mainly will write the average of the time evolution of our observable, $\mathbb{E}\left[P^{x_{0},M,N}_{t}\right]$, in a more convenient form for the analysis.
Starting from Eq. (\ref{eq:explicitP}) and inserting the expression for $\mathbb{E}\left[L^{n}\bar{L}^{m}\right]$ of Eq. (\ref{eq:twoprime}) we obtain 
\begin{align}
\mathbb{E}\left[P^{x_{0},M,N}_{t}\right]
 = & \sum_{n,m=0}^{M}\lambda^{n+m}\sum_{\{x_{i},x'_{j}\}_{0}} 
\sum_{\{E_{i},E'_{j}\}} \psi_{0}^{*}(x'_{m},E'_{m}) \psi_{0}(x_{n},E_{n})
 K^{n}(t,\{E_{i}\})\bar{K}^{m}(t,\{E'_{i}\}) \notag\\
& \times 
\sum_{a=0}^{2} \sum_{\pi(n,m)\in \mathcal{G}_{a}} C_{\pi(n,m)}(\{E_{i},E'_{j}\},\{x_{i},x'_{j}\}) \label{eq:begin}
\end{align}
By theorem \ref{thm:endsmeet} we have in Eq (\ref{eq:begin}) that the graph function will impose that $E'_{m}=E_{n}$ and $x'_{m}=x_{n}$.
Thus we can implement this relationship in the product $\psi_{0}^{*}(x'_{m},E'_{m}) \psi_{0}(x_{n},E_{n})$
\begin{align}
\mathbb{E}\left[P^{x_{0},M,N}_{t}\right] = & \sum_{n,m=0}^{M}\lambda^{n+m}
\sum_{\{E_{i},E'_{j}\}} 
 K^{n}(t,\{E_{i}\})\bar{K}^{m}(t,\{E'_{i}\})
\sum_{a=0}^{2} \sum_{C_{\pi}(n,m)\in \mathcal{G}_{a}}
\sum_{\{x_{i},x'_{j}\}_{0}} P^{x_{n}}_{0}(E_{n})C_{\pi(n,m)}(\{E_{i},E'_{j}\},\{x_{i},x'_{j}\})
\label{eq:begin2}
\end{align}
We can split the contributions to $\mathbb{E}\left[P^{x_{0},M,N}_{t}\right]$ in Eq. (\ref{eq:begin})
according to the three different type of graphs. This is the index $a$ in Eq. (\ref{eq:begin}).
\begin{align}
\mathbb{E}\left[P^{x_{0},M,N}_{t}\right] & = P^{x_{0},M,N}_{0,t}+P^{x_{0},M,N}_{1,t}+P^{x_{0},M,N}_{2,t}
\end{align}
with
\begin{align}
P^{x_{0},M,N}_{a,t}
 = & \sum_{n,m=0}^{M}\lambda^{n+m}
\sum_{\{E_{i},E'_{j}\}} 
 K^{n}(t,\{E_{i}\})\bar{K}^{m}(t,\{E'_{i}\}) 
\sum_{\pi(n,m)\in \mathcal{G}_{a}} 
\sum_{\{x_{i},x'_{j}\}_{0}} P^{x_{n}}_{0}(E_{n})C_{\pi(n,m)}(\{E_{i},E'_{j}\},\{x_{i},x'_{j}\}) \label{eq:Paxt}
\end{align}
and $a$ can take up the values $0$, $1$ or $2$.
We introduce a different representation of $P^{x_{0},M,N}_{a,t}$ that will turn out useful later on. We call this the $\alpha$-representation.
Starting from Eq. (\ref{firstK}) we use the following identities
\begin{align}
\delta(t-\sum_{j=0}^{n}s_{j}) & =\int_{-\infty}^{\infty}d\alpha e^{-i\alpha(t-\sum_{j=0}^{n}s_{j})} e^{\eta(t-\sum_{j=0}^{n}s_{j})}\label{deltaform} \\
\int_{0}^{\infty}ds e^{-is(\omega-i\eta)} & = \frac{-i}{\omega-i\eta} 
\end{align}
with $\eta>0$ and  obtain 
\begin{align}
& K^{n}(t,\{E_{i}\})=i\int_{-\infty}^{\infty}d \alpha e^{-i\alpha t} e^{\eta t} 
\frac{-1}{E_{0}-\alpha -i\eta} \frac{-1}{E_{1}-\alpha -i\eta}\dots \frac{-1}{E_{n}-\alpha -i\eta}
\label{eq:propagator2}
\end{align}
The same can be done for $\bar{K}^{m}(t,\{E'_{j}\})$
and so we can rewrite $K^{n}(t,\{E_{j}\})$ and $\bar{K}^{m}(t,\{E'_{j}\})$ as
\begin{align}
 \bar{K}^{m}(t,\{E'_{i}\})&=-i\int_{-\infty}^{\infty}d \beta e^{i\beta t} e^{\eta t} 
\prod_{j=0}^{m}\frac{-1}{E'_{j}-\beta +i\eta}  \label{eq:Kbar}\\ 
K^{n}(t,\{E_{i}\})&=i\int_{-\infty}^{\infty}d \alpha e^{-i\alpha t} e^{\eta t}
\prod_{j=0}^{n}\frac{-1}{E_{j}-\alpha -i\eta}  \label{eq:K}
\end{align}
The product of $\bar{K}^{m}$ and $K^{n}$, for example in Eq. (\ref{eq:Paxt}), has then $n+m+2$ propagators.
Remember that $E_{0}'=E_{0}$.
We set  $\eta=t^{-1}$ so that the exponential term $e^{\eta t}$ is bounded by a constant.
Inserting Eqs. (\ref{eq:Kbar}) and (\ref{eq:K})  in Eq. (\ref{eq:Paxt}) we obtain the following expression for $P^{x_{0},M,N}_{a,t}$:
\begin{align}
P^{x_{0},M,N}_{a,t}
 = & \sum_{n,m=0}^{M}\lambda^{n+m}
\sum_{\{E_{i},E'_{j}\}} 
\int_{-\infty}^{\infty} d\beta d \alpha e^{-i(\alpha-\beta) t}  e^{\eta 2t} 
\prod_{j=0}^{n} \frac{-1}{E_{j}-\alpha -i\eta} 
\prod_{j=0}^{m}\frac{-1}{E'_{j}-\beta +i\eta}
\notag \\
\times & \sum_{\pi(n,m)\in \mathcal{G}_{a}} 
\sum_{\{x_{i},x'_{j}\}_{0}} P^{x_{n}}_{0}(E_{n})C_{\pi(n,m)}(\{E_{i},E'_{j}\},\{x_{i},x'_{j}\}) \label{eq:Pat}%\\
\end{align}
From the definitions of dependent and independent variables of section \ref{sect:AveragingAndGraphs} we have 
\begin{align}
P^{x_{0},M,N}_{a,t}
 = & \sum_{n,m=0}^{M}\lambda^{n+m}\sum_{\pi(n,m)\in \mathcal{G}_{a}}
\sum_{ \mathcal{A}_{\pi(n,m)}}\sum_{ \mathcal{B}_{\pi(n,m)}} 
\int_{-\infty}^{\infty} d\beta d \alpha e^{-i(\alpha-\beta) t}  e^{\eta 2t} 
\prod_{j=0}^{n} \frac{-1}{E_{j}-\alpha -i\eta} 
\prod_{j=0}^{m}\frac{-1}{E'_{j}-\beta +i\eta}
\notag \\
\times &  
\sum_{\{x_{i},x'_{j}\}_{0}} P^{x_{n}}_{0}(E_{n})C_{\pi(n,m)}(\{E_{i},E'_{j}\},\{x_{i},x'_{j}\})
 \label{eq:Pat2}%\\
\end{align}
The sum over $\{x_{i},x'_{j}\}_{0}$ is a sum over all elements excluding $x_{0}$. Because of the form of the interaction, Eq. (\ref{eq:interactionM}), we have that if $x_{j}=1(2)$ then $x_{j+1}=2(1)$. That is 
\begin{align}
\langle x_{j},E_{j}|V|x_{j+1},E_{j+1}\rangle \propto 1-\delta_{x_{j+1},x_{j}} 
\end{align}
Thus we have $x_{0}=x_{2},x_{4} \dots$ and $x_{0}\neq x_{1},x_{3} \dots$.
The same holds for $x'_{i}$, $x_{0}=x'_{2},x'_{4},\dots$ and $x'_{i}$, $x_{0}\neq x'_{1},x'_{3},\dots$. Therefore if $n$  is even then $x_{n}=x_{0}$ and if $n$ is odd then $x_{n}\neq x_{0}$. 
We thus define
\begin{align}
 P_{0}(x_{0},E_{n})=\left\{ \begin{array}{ll}
P^{\bar{x}_{0}}_{0}(E_{n}) & \text{if $n$ odd}\\
 P^{x_{0}}_{0}(E_{n}) & \text{if $n$ even}
\end{array} \right.\label{eq:Px0}
\end{align}
where $\bar{x}_{0}=1$ if $x_{0}=2$ and vice versa.
We have then
\begin{displaymath}
\sum_{\{x_{i},x'_{j}\}_{0}} P^{x_{n}}_{0}(E_{n})C_{\pi(n,m)}(\{E_{i},E'_{j}\},\{x_{i},x'_{j}\})
 = \left\{ \begin{array}{ll}
 P^{x_{0}}_{0}(E_{n})\sum_{\{x_{i},x'_{j}\}_{0}}C_{\pi(n,m)}(\{E_{i},E'_{j}\},\{x_{i},x'_{j}\})
 & \text{if $n$ even}\\
P^{\bar{x}_{0}}_{0}(E_{n})\sum_{\{x_{i},x'_{j}\}_{0}}C_{\pi(n,m)}(\{E_{i},E'_{j}\},\{x_{i},x'_{j}\})
 & \text{if $n$ odd}
\end{array} \right.
\end{displaymath}
with the definition in Eq. (\ref{eq:Px0}) we get
\begin{align}
\sum_{\{x_{i},x'_{j}\}_{0}} P^{x_{n}}_{0}(E_{n})C_{\pi(n,m)}(\{E_{i},E'_{j}\},\{x_{i},x'_{j}\})= P_{0}(x_{0},E_{n})\sum_{\{x_{i},x'_{j}\}_{0}} C_{\pi(n,m)}(\{E_{i},E'_{j}\},\{x_{i},x'_{j}\})
\end{align}
%%%%%%%%%%%%%%%%%%%%%%%%%%%%%%%%%%%%%%%%
As discussed in section \ref{sect:AveragingAndGraphs} the graph function, $C_{\pi(n,m)}(\{E_{i},E'_{j}\},\{x_{i},x'_{j}\})$, is  a product of Kronecker delta functions divided by $N^{\frac{n+m}{2}}$. Thus only a part of the variables $\{E_{i},E'_{j}\}$ are independent. 
This means that when having a sum of the type
\begin{align}
\sum_{ \mathcal{A}_{\pi(n,m)}}\sum_{ \mathcal{B}_{\pi(n,m)}} 
\sum_{\{x_{i},x'_{j}\}_{0}} P_{0}(x_{0},E_{n})C_{\pi(n,m)}(\{E_{i},E'_{j}\},\{x_{i},x'_{j}\}) \prod_{i}f(E_{i})\prod_{j}g(E'_{j}) \label{eq:exampleC}
\end{align}
which we have in Eq. (\ref{eq:Pat2}), each independent variable $E_{l}$ of the set  $\{E_{i},E'_{j}\}$
will appear a certain amount of times in $f$ and $g$, which we denote by $k_{l}$ and $p_{l}$.
That is if we relabel the independent energy variables by $\omega_{j}$, Eq. (\ref{eq:exampleC}) has the following form:
\begin{align}
(\ref{eq:exampleC})
=\prod_{j=1}^{\kappa_{\pi(n,m)}}\sum_{\omega_{j}} \left( \frac{P_{0}(x_{0},\omega_{\kappa_{\pi}})}{N^{\frac{n+m}{2}}}\prod_{j=1}^{\kappa_{\pi}}f^{k_{j}}(\omega_{j})g^{p_{j}}(\omega_{j})\right)
\end{align}
where $k_{j}$ and $p_{j}$ depends on the graph and $\kappa_{\pi(n,m)}$ is equal to the number of independent variables.
We labelled the independent variable related to $E_{n}$ (the one referring to the initial data ) by $\omega_{\kappa_{\pi}}$.
In Eq. (\ref{eq:Pat2}) the functions $f$ and $g$ are  the propagators $\frac{-1}{E_{j}-\alpha -i\eta} $ and $\frac{-1}{E'_{j}-\beta +i\eta}$. We call $k_{j}$ and $p_{j}$  the left and right multiplicity of the independent variable $\omega_{j}$, and $k_{j}+p_{j}$ the multiplicity of $\omega_{j}$.
We obtain  then from Eq. (\ref{eq:Pat2})
\begin{align}
P^{x_{0},M,N}_{a,t}
 = & \sum_{n,m=0}^{M}\lambda^{n+m}\sum_{\pi(n,m)\in \mathcal{G}_{a}}  
\int_{-\infty}^{\infty} d\beta d \alpha e^{-i(\alpha-\beta) t}  e^{\eta 2t} \notag\\
\times & \prod_{j=1}^{\kappa_{\pi(n,m)}}\sum_{\omega_{j} \in \mathcal{A}_{\pi}}\left(\frac{1}{N^{\frac{n+m}{2}}} P_{0}(x_{0},\omega_{\kappa_{\pi}}) 
\prod_{j=1}^{\kappa_{\pi(n,m)}}\left( \left(\frac{-1}{\omega_{j}-\alpha -i\eta}\right)^{k_{j}} 
\left(\frac{-1}{\omega_{j}-\beta +i\eta}\right)^{p_{j}} \right)\right)\label{eq:PNaxt}
\end{align}
This is then called the $\alpha$-representation.
The information about $p_{j}$ and $k_{j}$ lies in the graph $\pi(n,m)$ but since the amount of propagators from a term of order $n+m$ of the expansion is $n+m+2$, as can be seen in  Eq. (\ref{eq:Pat2}), we must have
\begin{align}
\sum_{j=1}^{\kappa_{\pi(n,m)}}\left(k_{j}+p_{j}\right)=n+m+2 \label{eq:propagatorcount}
\end{align}
We define $Q^{N}_{\pi(n,m)}(t,\lambda,x_{0})$ such that 
\begin{align}
P^{x_{0},M,N}_{a,t}
 = & \sum_{n,m=0}^{M}\sum_{\pi(n,m)\in \mathcal{G}_{a}}  
Q^{N}_{\pi(n,m)}(t,\lambda,x_{0}) \label{eq:qN}
\end{align}
$Q^{N}_{\pi(n,m)}$ is the contribution of the graph $\pi(n,m)$ to the probability to be at $x_{0}$.
Notice that if we sum over $x_{0}$ and over $a$ in Eq. (\ref{eq:Pat2}) we obtain the squared norm of $\sum_{n=0}^{M} |\psi_{t}^{n}\rangle$. Thus $\sum_{x_{0}}Q^{N}_{\pi(n,m)}(t,\lambda,x_{0})$ is the contribution of the graph $\pi(n,m)$ to the norm of the wave vector.
%%%%%%%%%%%%%%%%%%%%%%%%%%%%%%%%%%%%%%%
\section{CROSSING GRAPHS}
\label{sec:crossing}
We now prove the following lemma for crossing graphs 
\begin{lem}
\label{lem:CrossContribution}
The contribution of crossing graphs to the time evolution of the observable $\hat{P}^{x_{0}}$ tends to zero in the limit $N\rightarrow \infty$. That is 
\begin{align}
\lim_{N\rightarrow \infty} P^{x_{0},M,N}_{a=2,t} & =0 
\end{align}
\end{lem}
\begin{proof}{Lemma \ref{lem:CrossContribution}} \\
By inspecting Eq. (\ref{eq:PNaxt}) we see that we have a factor of $N^{-\frac{n+m}{2}}$ and a sum over $\kappa_{\pi}$ energy variables, where one sum is weighted by $P_{0}(x_{0},\omega_{\kappa})$. If $\kappa_{\pi}= \frac{n+m}{2}+1$ then each sum is weighted by a $N^{-1}$ factor except one that is weighted by $P_{0}(x_{0},\omega_{\kappa})$, the initial probability distribution. All the sums would be finite. But if $\kappa_{\pi}< \frac{n+m}{2}+1$ then a factor of $N^{-1}$ could be extracted and so this term would vanish. This is the case for crossing graphs.
From Eq. (\ref{eq:Paxt}) we have
\begin{align}
\left|P^{x_{0},M,N}_{2,t}\right|
 \leq & \sum_{n,m=0}^{M}\lambda^{n+m}\sum_{\pi(n,m)\in \mathcal{G}_{a}} 
\sum_{\{E_{i},E'_{j} \}} 
\left| K^{n}(t,\{E_{i}\})\bar{K}^{m}(t,\{E'_{i}\}) \right|%\notag \\
%\times & 
 \sum_{\{x_{i},x'_{j} \}_{0}} P^{x_{n}}_{0}(E_{n})C_{\pi(n,m)}(\{E_{i},E'_{j}\},\{x_{i},x'_{j}\}) 
\end{align}
From Eq. (\ref{firstK}) we have that
\begin{align}
\left|  K^{n}(t,\{E_{i}\})\right| \leq \frac{t^{n}}{n!} \\
\left|  \bar{K}^{m}(t,\{E'_{i}\})\right| \leq \frac{t^{m}}{m!}
\end{align}
Thus 
\begin{align}
\lim_{N\rightarrow \infty}\left|P^{x_{0},M,N}_{2,t}\right|
 \leq & \lim_{N\rightarrow \infty} \sum_{n,m=0}^{M}\lambda^{n+m} \frac{t^{n+m}}{n!m!}
\sum_{\pi(n,m)\in \mathcal{G}_{2}}\sum_{\{E_{j},E'_{i}\}^{n,m}} \sum_{\{x_{i},x'_{j} \}^{n,m}_{0}} P^{x_{n}}_{0}(E_{n})C_{\pi(n,m)}(\{E_{i},E'_{j}\},\{x_{i},x'_{j}\})\notag \\
 \leq & \lim_{N\rightarrow \infty} \sum_{n,m=0}^{M}\lambda^{n+m} \frac{t^{n+m}}{n!m!}
\sum_{\pi(n,m)\in \mathcal{G}_{2}}\sum_{\{E_{j},E'_{i},x_{i},x'_{j}\}\in \mathcal{A}_{\pi(n,m) }} P^{x_{n}}_{0}(E_{n})\frac{1}{N^{\frac{n+m}{2}}}\notag \\
 \leq & \lim_{N\rightarrow \infty} \sum_{n,m=0}^{M}\lambda^{n+m} \frac{t^{n+m}}{n!m!}
\sum_{\pi(n,m)\in \mathcal{G}_{2}} \frac{N^{\kappa_{\pi}-1}}{N^{\frac{n+m}{2}}}\notag \\
 \leq & \lim_{N\rightarrow \infty} \frac{1}{N^{2}}\sum_{n,m=0}^{M}\lambda^{n+m} \frac{t^{n+m}}{n!m!}
\sum_{\pi(n,m)\in \mathcal{G}_{2}} 1 \notag \\
\leq & 0
\end{align}
\end{proof}
\section{NESTED GRAPHS}
\label{sec:nested}
Nested graphs are non crossing graphs. If there are no crossing it means that when there is a contraction between two elements, for example $X_{E_{g},E_{g+1}}(x_{g},x_{g+1})$ and $X_{E_{h},E_{h+1}}(x_{h},x_{h+1})$, then the elements in between these two, in the total product $\bar{L}^{m}L^{n}$, can only contract among themselves and thus the energy variables in between $E_{g+1}$ and $E_{h}$ are independent of all the others. This mean that in the sum
\begin{align}
\sum_{\{E_{i},E'_{j}\}}P_{0}^{x_{n}}(E_{n}) C_{\pi(n,m)}(\{E_{i},E'_{j}\},\{x_{i},x'_{j}\}) \prod_{i}f(E_{i})\prod_{j}g(E'_{j}) 
=\prod_{j=1}^{\kappa_{\pi(n,m)}}\sum_{\omega_{j}}  \frac{1}{N^{\frac{n+m}{2}}}\prod_{j}f^{k_{j}}(\omega_{j})g^{p_{j}}(\omega_{j}) \label{eq:exam}
\end{align}
there are independent variables, $\omega_{j}$, for which $k_{j}$ or $p_{j}$ is $0$.
The simplest example is that of an nn contraction. If $X_{E_{g},E_{g+1}}(x_{g},x_{g+1})$ contracts with $X_{E_{g+1},E_{g+2}}(x_{g+1},x_{g+2})$ this imposes according to Eq. (\ref{eq:OnePairp}) the relationship
\begin{align}
E_{g}&=E_{g+2} \notag \\
E_{g+1}&=E_{g+1}
\end{align}
The second equation is of course superflues but since there are no more random elements in the product that depend on $E_{g+1}$ there are no more contraction which could relate $E_{g+1}$ to another energy variable. Thus
 $E_{g+1}$ is independent of the rest. Then in the product of Eq. (\ref{eq:exam}) if $\omega_{l}=E_{g+1}$ we would have $p_{l}=0$ and $k_{l}=1$ and thus a term  $f(\omega_{l})$.
The particularity of nested graphs is that there is at least one independent variable $\omega_{j}$ such that $p_{j}=0$ and $k_{j}>1$ or $k_{j}=0$ and $p_{j}>1$. This is easy to see as follows.
Say we have a nested graph, that is suppose we have no crossings and that 
$X_{E_{g},E_{g+1}}(x_{g},x_{g+1})$ and $X_{E_{h},E_{h+1}}(x_{h},x_{h+1})$ contract with $ g+1< h$. 
We have then $E_{g+1}=E_{h}$. 
Since no elements $X_{E_{j},E_{j+1}}(x_{j},x_{j+1})$ with $g<j<h$ can contract with a primed random variable $E_{h}$ cannot be equal to a primed $E'_{j}$ variable. Thus for the independent variable $\omega_{l}=E_{g+1}$ we will have $p_{l}=0$ and  have $k_{l}\geq 2$.
Simple graphs do not have such independent variables.\\
According to theorem \ref{thm:kappa} we have for non crossing graphs in Eq. (\ref{eq:PNaxt}) $\kappa_{\pi(n,m)}=\frac{n+m}{2}+1$. Thus for non crossing graphs we have
\begin{align}
P^{x_{0},M,N}_{a,t}
 = & \sum_{n,m=0}^{M}\lambda^{n+m}\sum_{\pi(n,m)\in \mathcal{G}_{a}}
\int_{-\infty}^{\infty} d\beta d \alpha e^{-i(\alpha-\beta) t}  e^{\eta 2t} \notag\\
\times & \prod_{j=1}^{\kappa_{\pi}}\sum_{\omega_{j} \in \mathcal{A}_{\pi}}\left(\frac{1}{N^{\kappa_{\pi}-1}} P_{0}(x_{0},\omega_{\kappa_{\pi}}) 
\prod_{j=1}^{\kappa_{\pi}}\left( \left(\frac{-1}{\omega_{j}-\alpha -i\eta}\right)^{k_{j}} 
\left(\frac{-1}{\omega_{j}-\beta +i\eta}\right)^{p_{j}} \right)\right)\label{eq:PNaxt2}
\end{align}
From Eq. (\ref{eq:PNaxt2}) %and by Eq. (\ref{eq:Nconvergence2}) 
we have in the limit $N\rightarrow \infty$ 
\begin{align}
P^{x_{0},M}_{a,t}=&\lim_{N\rightarrow \infty }P^{x_{0},M,N}_{a,t} \notag \\
 = & \sum_{n,m=0}^{M}\lambda^{n+m}\sum_{\pi(n,m)\in \mathcal{G}_{a}}  
\int_{-\infty}^{\infty} d\beta d \alpha e^{-i(\alpha-\beta) t}  e^{\eta 2t} \notag \\
\times & \prod_{j=1}^{\kappa_{\pi}}\int d\omega_{j} \left(P_{0}(x_{0},\omega_{\kappa_{\pi}}) 
 \prod_{j=1}^{\kappa_{\pi}}\left(\left(\frac{-1}{\omega_{j}-\alpha -i\eta}\right)^{k_{j}} 
\left(\frac{-1}{\omega_{j}-\beta +i\eta}\right)^{p_{j}} \right)\right)\label{eq:contPaxt}
\end{align}
We will now prove the following theorem:
\begin{thm}
\label{thm:nested}
In the Van Hove limit ($\lambda^{2}t=T<\infty$) the contribution from nested graphs ($\mathcal{G}_{1}$) to the average of the time evolution of the observable, $\mathbb{E}\left[P^{x_{0},M}_{t}\right]$, vanishes. That is
\begin{align}
\lim_{t\rightarrow \infty}^{t\lambda^{2}=T}P^{x_{0},M}_{1,t}=0
\end{align}
\end{thm}
\begin{proof}{Theorem \ref{thm:nested}}\\
We define the following:
\begin{align}
 Q_{\pi(n,m)}(t,\lambda,x_{0})
= &  \lambda^{n+m}
\int_{-\infty}^{\infty} d\beta d \alpha e^{-i(\alpha-\beta) t}  e^{\eta 2t}  \notag \\
\times & \prod_{j=1}^{\kappa_{\pi}} \int d\omega_{j}\left(
P_{0}(x_{0},\omega_{\kappa_{\pi}}) 
\prod_{j=1}^{\kappa_{\pi}}\left(
 \left(\frac{-1}{\omega_{j}-\alpha -i\eta}\right)^{k_{j}} 
\left(\frac{-1}{\omega_{j}-\beta +i\eta}\right)^{p_{j}}\right)\right)
\label{eq:Qfunction}
\end{align}
This is just the limit of $Q^{N}_{\pi(n,m)}(t,\lambda,x_{0})$ as $N\rightarrow \infty$ of Eq. (\ref{eq:qN}). Notice that this is the form of $Q^{N}_{\pi(n,m)}$ for all non crossing graphs.
For $a=0,1$ we have then from Eq. (\ref{eq:contPaxt})
\begin{align}
\lim_{N\rightarrow \infty }P^{x,M,N}_{a,t}= \sum_{n,m=0}^{M}\sum_{\pi(n,m)\in \mathcal{G}_{a}}
 Q_{\pi(n,m)}(t,\lambda,x) \label{eq:PxatOfQ}
\end{align}
Since we are considering a nested graph there exists an $\omega_{l}$ such that either $k_{l}=0$ and $p_{l}>1$ or vice versa. We can thus perform the integration over this variable. From Eq. (\ref{eq:Qfunction}) we have
\begin{align}
 Q_{\pi(n,m)}(t,\lambda,x_{0}) 
= &  
\lambda^{n+m} \int_{-\infty}^{\infty} d\beta d \alpha  e^{-i(\alpha-\beta) t}e^{\eta 2t} \left(\int d\omega_{l} 
 \left(\frac{-1}{\omega_{l}-\alpha -i\eta}\right)^{k_{l}}\right)\notag \\
\times &
\prod_{j=1,\neq l}^{\kappa_{\pi}} \int d\omega_{j}\left( P_{0}(x_{0},\omega_{\kappa_{\pi}}) 
\prod_{j=1,\neq l}^{\kappa_{\pi}} \left( \left(\frac{-1}{\omega_{j}-\alpha -i\eta}\right)^{k_{j}} 
\left(\frac{-1}{\omega_{j}-\beta +i\eta}\right)^{p_{j}} \right)\right)\notag 
\end{align}
We can perform the integration over $\omega_{l}$ first. After taking the absolute value we use inequality (\ref{ine:two}) to bound almost all integrations over $\omega_{j}$. We can only apply inequality (\ref{ine:two}) to those integrations where $k_{j}+p_{j}\geq 2$ and we do so except for two variables which we denote  $\omega_{1}$ and $\omega_{2}$. For these variables we bound the set of propagators by 
\begin{align}
 \left|\frac{1}{\omega_{1}-\alpha-i\eta}\right|^{k_{1}} \left|\frac{1}{\omega_{1}-\beta+i\eta}\right|^{p_{1}}
\leq \eta^{-\left(k_{1}+p_{1}-2\right)}  \left|\frac{1}{\omega_{1}-\alpha-i\eta}\right| \left|\frac{1}{\omega_{1}-\beta+i\eta}\right|
\end{align}
When $k_{j}=0$ and $p_{j}=1$ we can bound the integral by a $\left|\log \eta\right|$ term. We denote by $n'$ the number of cases in which  $k_{j}=0$ and $p_{j}=1$ or $k_{j}=1$ and $p_{j}=0$. 
This corresponds then to the number of propagators with multiplicity equal to $1$. We denote by $\bar{n}+1$ the number of propagators of multiplicity higher then $1$. Thus we have $\bar{n}+n'+1=\kappa_{\pi(n,m)}$. For non crossing graphs we have then $\bar{n}+n'+1=\frac{n+m}{2}+1$. 
Therefore we have $n'<\frac{n+m}{2}$.
We have then : 
\begin{align}
\left| Q_{\pi(n,m)}(t,\lambda,x) \right| \leq &\lambda^{n+m} \int_{-\infty}^{\infty} d\beta d \alpha  e^{\eta 2t}
 \eta^{-\sum_{j=1, \neq l}\left(k_{j}+p_{j}-1\right)+2} \left(\left|\frac{-1}{1-\alpha -i\eta}\right|^{k_{l}-1}+\left|\frac{-1}{-\alpha -i\eta}\right|^{k_{l}-1} \right)
\left|\log \eta\right|^{n'} \notag \\
\times & \int d \omega_{1} d \omega_{2} 
\left|\frac{1}{\omega_{1}-\alpha -i\eta}\right| 
\left|\frac{1}{\omega_{1}-\beta +i\eta}\right|
\left|\frac{1}{\omega_{2}-\alpha -i\eta}\right| 
\left|\frac{1}{\omega_{2}-\beta +i\eta}\right| \notag
\end{align}
The sum over $j$ in the exponent of $\eta$ should be over those $j$ for which $k_{j}+p_{j}\geq 2$ but since we are summing $k_{j}+p_{j}-1$ we can extend to all $j$ since if $k_{j}+p_{j}=1$ then we are summing zero.
By using inequality (\ref{eq:4kintegrals}) and remembering that $\eta=t^{-1}$ we get
\begin{align}
\left| Q_{\pi(n,m)}(t,\lambda,x) \right| \leq &\lambda^{n+m}  e^{\eta 2t}
 \eta^{-\sum_{j=1, \neq l}\left(k_{j}+p_{j}-1\right)+2}
\frac{C\left|\log \eta\right|^{3+n'}}{\eta^{k_{l}-1}}
\notag  \\
\leq &C \lambda^{n+m} 
 t^{\sum_{j=1}\left(k_{j}+p_{j}-1\right)-2} 
\left|\log t \right|^{3+n'}
\label{eq:QN}
\end{align}
Since for non crossing graphs  $\kappa_{\pi(n,m)}=\frac{n+m}{2}+1$, we have from Eq. (\ref{eq:propagatorcount})
\begin{align}
\sum_{j=1}^{\kappa_{\pi}}\left(k_{j}+p_{j}-1\right)&=(n+m+2)-\kappa_{\pi} \notag \\
&=\frac{n+m}{2} +1 \notag
\end{align}
Inserting the last equation in Eq. (\ref{eq:QN}) and maximizing $n'$ by $\frac{n+m}{2}$ gives
\begin{align}
\left| Q_{\pi(n,m)}(t,\lambda,x) \right| 
\leq & \left(CT\right)^{\frac{n+m}{2}} \frac{\left(\log t\right)^{3+\frac{n+m}{2}}}{t} \label{eq:Qnestedbound}
\end{align}
with $\lambda^2t=T$.
This vanishes in the Van Hove limit.
\begin{align}
\lim_{t\rightarrow \infty}\left| Q_{\pi(n,m)}(t,\lambda,x) \right| \leq &\lim_{t\rightarrow \infty} C T^{\frac{n+m}{2}} \frac{(\log t)^{3+\frac{n+m}{2}}}{t} \notag \\
=& 0 \label{eq:qtozero}
\end{align}
Thus from Eqs. (\ref{eq:PxatOfQ}) and (\ref{eq:qtozero}),  and from the fact the number of nested graphs of length $n+m$ is less then $c^{n+m}$, with $c$ a constant ,
we have for the contribution of the nested graphs the following bound
\begin{align}
\left|P^{x}_{1,t}\right|
\leq & \sum_{n,m=0}^{M}\sum_{\pi(n,m)\in \mathcal{G}_{1}}  \left| Q_{\pi(n,m)}(t,\lambda,x) \right|
\notag\\
\leq & \sum_{n,m=0}^{M} \left(C'T\right)^{\frac{n+m}{2}} \frac{(\log t)^{3+\frac{n+m}{2}}}{t}
\notag \\
\lim_{t\rightarrow \infty }\left|P^{x}_{1,t}\right| =& 0
\end{align}
Thus the contribution of a nested graph vanishes.
\end{proof}
\section{SIMPLE GRAPHS}
\label{sec:simple}
As mentioned, simple graphs are graphs which are build from nn contractions and outer contractions.
This means that in Eq. (\ref{eq:Qfunction}) we can either have $k_{j}=1$ and $p_{j}=0$ (nn contractions) or
we can have $k_{j}\geq 1$ and $p_{j}\geq 1$. We can see this as follows.
Say we have two outer contractions between $X_{E_{q},E_{q+1}}$ and $X_{E'_{p+1},E'_{p}}$ and between $X_{E_{g},E_{g+1}}$ and $X_{E'_{h+1},E'_{h}}$. and suppose there are no outer contractions in between these two outer contractions. That is  there is no outer contraction between $X_{E_{a},E_{a+1}}$ and $X_{E'_{b+1},E'_{b}}$ with $g+1<a<q$ and $h+1<b<p$. 
we take the outer contraction between  $X_{E_{q},E_{q+1}}$ and $X_{E'_{p+1},E'_{p}}$ to be the $(j+1)^{\text{th}}$ outer contraction and the contraction between  $X_{E_{g},E_{g+1}}$ and $X_{E'_{h+1},E'_{h}}$ to be the $j^{\text{th}}$. We are thus counting from the inside to the outside.
If the graph is a simple graph then all elements $X_{E_{a},E_{a+1}}$ with 
$g+1<a<q$ contract amongst each other and form only nn contractions. The same is valid for all elements $X_{E'_{b+1},E'_{b}}$ with $h+1<b<p$. By Eq. (\ref{}) we have then 
\begin{align}
E_{q}=& E_{q-2}
= E_{q-4} \dots
= E_{g+1} 
= E'_{h+1} 
= E'_{h+3} 
\dots = E'_{p} \notag
\end{align}
All of the other variables, $E_{q-1},E_{q-3}\dots E_{g+2}, E'_{h+2},E'_{h+4}\dots E'_{p-1}$, are independent and of multiplicity $1$.
If we set $\omega_{j}=E_{q}$ as the independent variable with respectively 
$k_{j}$ and $p_{j}$ as left and right multiplicities then there are $k_{j}-1$ 
variables in between the two outer contractions of multiplicity $1$ on the left and $p_{j}-1$ on the right.
Thus in Eq. (\ref{eq:exam}) there will be in the product a term 
\begin{align}
 f^{k_{j}}(\omega_{j}) g^{p_{j}}(\omega_{j})\prod_{l=1}^{k_{j}-1}f(\omega_{j}^{l})\prod_{l=1}^{p_{j}-1}g(\bar{\omega}^{l}_{j})  \label{eq:form}
\end{align}
In view of this we will change our notation.
We denote by $\bar{n}$ the number of outer contractions for a graph and so there are $\bar{n}+1$ independent variables of multiplicity higher then $1$. We set $n'$ to be the number of variables of multiplicity equal to $1$. 
Instead of $k_{j}$ denoting the multiplicity of any independent variable we will denote by $k_{j}+1$ the multiplicity of independent variables with multiplicity higher then $1$. Eq. (\ref{eq:form}) then becomes
\begin{align}
 f^{k_{j}+1}(\omega_{j}) g^{p_{j}+1}(\omega_{j})\prod_{l=1}^{k_{j}}f(\omega_{j}^{l})\prod_{l=1}^{p_{j}}g(\bar{\omega}^{l}_{j})  \label{eq:form2}
\end{align}
Notice that the set of numbers $\bar{n},\{k_{j},p_{j}\}$ determine uniquely the simple graph and so for each set there is unique graph.
We will introduce this notation in Eq. (\ref{eq:Qfunction}). 
In between each two outer contractions we will have a product of the form of Eq. (\ref{eq:form2}) and so we have from Eq. (\ref{eq:Qfunction})
\begin{align}
Q_{\pi(n,m)} (t,\lambda,x)&= \lambda^{n+m}
\int d\alpha d\beta 
e^{-i(\alpha -\beta)t}e^{2\eta t}  \label{eq:simpleQ} \\
\times &
\prod_{j=0}^{\bar{n}}\int d\omega_{j} P_{0}(x,\omega_{\bar{n}})\prod_{j=0}^{\bar{n}}\left(
\left(\frac{1}{\omega_{j}-\alpha -i\eta}\right)^{k_{j}+1}
\left(\frac{1}{\omega_{j}-\beta +i\eta} \right)^{p_{j}+1} 
\Theta^{k_{j}}(\alpha,\eta)\bar{\Theta}^{p_{j}}(\beta,\eta)\right) \notag
\end{align}
with
\begin{align}
 \Theta(\alpha,\eta)=\int d \omega \frac{-1}{\omega-\alpha-i\eta} \label{eq:ThetaFunction}
\end{align}
$k_{j}$ and $p_{j}$  here are not those from section \ref{sec:nested}. 
Since a graph of order $n+m$  has $n+m+2$ propagators and that there are $\frac{n+m}{2}+1$ independent variables we have in this new notation
\begin{align}
 \sum_{j=0}^{\bar{n}}\left(k_{j}+p_{j}+2\right)+n'=& n+m+2 \label{eq:kpnm}  \\
 \bar{n}+1+n'=& \frac{n+m}{2}+1 \label{eq:kpnm2}
\end{align}
and 
\begin{align}
P^{x,M}_{0,t}&=\sum_{n,m=0}^{M}\sum_{\pi(n,m)\in \mathcal{G}_{0}}Q_{\pi(n,m)}(t,\lambda,x) \label{eq:PatM} 
\end{align}
We now will  prove that in Eqs. (\ref{eq:simpleQ}) and (\ref{eq:PatM}) the function $\Theta(\alpha,\eta)$
can be replaced by $\Theta(\omega_{\bar{n}})$, with $\Theta(\omega_{\bar{n}})=\lim_{\eta\rightarrow \infty}\Theta(\omega_{\bar{n}},\eta)$,   such that the error goes to zero in the van Hove limit. 
We define then 
\begin{align}
 \tilde{P}^{x,M}_{0,t}&=\sum_{n,m=0}^{M}\sum_{\pi(n,m)\in \mathcal{G}_{0}}\tilde{Q}_{\pi(n,m)}(t,\lambda,x) \label{eq:PMattilde} \\ 
\tilde{Q}_{\pi(n,m)} (t,\lambda,x)&= \lambda^{n+m}
\int d\alpha d\beta 
e^{-i(\alpha -\beta)t}e^{2\eta t}   \notag \\
\times &
\prod_{j=0}^{\bar{n}}\int d\omega_{j}P_{0}(x,\omega_{\bar{n}}) \prod_{j=0}^{\bar{n}}\left(
\left(\frac{1}{\omega_{j}-\alpha -i\eta}\right)^{k_{j}+1}
\left(\frac{1}{\omega_{j}-\beta +i\eta} \right)^{p_{j}+1} 
\Theta^{k_{j}}(\omega_{\bar{n}})\bar{\Theta}^{p_{j}}(\omega_{\bar{n}})\right) \label{eq:PMattilde2}
\end{align}
Thus we will analyze the difference $\left|\Delta Q_{\pi(n,m)}(t,\lambda,x)\right|=\left|Q_{\pi(n,m)}(t,\lambda,x)-\tilde{Q}_{\pi(n,m)}(t,\lambda,x)\right|$. For briefness we denote $\Delta Q_{\pi(n,m)}(t,\lambda,x)$ by $\Delta Q_{\pi}$ in this section.
\begin{align}
\Delta Q_{\pi} &=Q_{\pi(n,m)}(t,\lambda,x)-\tilde{Q}_{\pi(n,m)}(t,\lambda,x) \notag \\
= \lambda^{n+m}&
\int d\alpha d\beta 
e^{-i(\alpha -\beta)t}e^{2 \eta t}   \prod_{j=0}^{\bar{n}}\int d\omega_{j}P_{0}(x,\omega_{\bar{n}}) \prod_{j=0}^{\bar{n}}\left(
\left(\frac{1}{\omega_{j}-\alpha -i\eta}\right)^{k_{j}+1}
\left(\frac{1}{\omega_{j}-\beta +i\eta} \right)^{p_{j}+1}\right)\notag \\
\times &
\left(\prod_{j=0}^{\bar{n}}\Theta^{k_{j}}(\alpha,\eta)\bar{\Theta}^{p_{j}}(\beta,\eta)-\prod_{j=0}^{\bar{n}}\Theta^{k_{j}}(\omega_{\bar{n}})\bar{\Theta}^{p_{j}}(\omega_{\bar{n}})\right) \label{DeltaQ} 
\end{align}
 and show that it tends to zero in the van Hove limit such that $\left|\tilde{P}^{x,M}_{0,t}-P^{x,M}_{0,t}\right|$ tends to zero in the limit.
We will prove the following theorem:
\begin{thm}
\label{thm:simple}
In the Van Hove limit we have
\begin{align}
\lim_{t\rightarrow \infty}^{t\lambda^{2}=T} \left(\tilde{P}^{x,M}_{0,t}-P^{x,M}_{0,t}\right)=0 
\end{align}
\end{thm}
\begin{proof}{Theorem \ref{thm:simple}} \\
To prove this theorem we will first bound $\Delta Q_{\pi}$.
First we bound the term in the second line of Eq. (\ref{DeltaQ}).
We split this term as follows:
\begin{align}
\Theta^{K}(\alpha,\eta)\bar{\Theta}^{P} (\beta,\eta) -
\Theta^{K} (\omega_{\bar{n}})\bar{\Theta}^{P}(\omega_{\bar{n}})
 =&  \bar{\Theta}^{P} (\beta,\eta)
\sum_{l=0}^{K-1}
\Theta^{K-l-1}(\alpha,\eta) \Theta^{l} (\omega_{\bar{n}})
\left(\Theta(\alpha,\eta)-\Theta(\omega_{\bar{n}})\right) \notag \\
+ & 
 \Theta^{K} (\omega_{\bar{n}})
\sum_{l=0}^{P-1}
\bar{\Theta}^{P-l-1}(\beta,\eta) \bar{\Theta}^{l} (\omega_{\bar{n}})
\left(\bar{\Theta}(\beta,\eta)-\bar{\Theta}(\omega_{\bar{n}})\right)
 \notag
\end{align}
We bound the two parts as follows:
\begin{align}
& \left| 
\Theta^{K}(\alpha,\eta)\bar{\Theta}^{P} (\beta,\eta) -
\Theta^{K} (\omega_{\bar{n}})\bar{\Theta}^{P}(\omega_{\bar{n}})
\right|   \leq  A+ B\\ 
A&=\left|\bar{\Theta}^{P} (\beta,\eta)\right|
\sum_{l=0}^{K-1}
\left|\Theta^{K-l-1}(\alpha,\eta) \Theta^{l} (\omega_{\bar{n}})\right|
\left|\Theta(\alpha,\eta)-\Theta(\omega_{\bar{n}})\right|
\notag \\
&=K\left|\log\eta\right|^{P+K-1}
\left|\Theta(\alpha,\eta)-\Theta(\omega_{\bar{n}})\right|
\notag \\
&=\left|\log \eta\right|^{K+P}
\left|\Theta(\alpha,\eta)-\Theta(\omega_{\bar{n}})\right|
\notag \\
B&= \left|\Theta^{K} (\omega_{\bar{n}})\right|
\sum_{l=0}^{P-1}
\left|\bar{\Theta}^{P-l-1}(\beta,\eta) \bar{\Theta}^{l} (\omega_{\bar{n}})\right|
\left|\bar{\Theta}(\beta,\eta)-\bar{\Theta}(\omega_{\bar{n}})\right| \notag\\
&= C^{K}
\sum_{l=0}^{P-1}
\left|\log \eta\right|^{P-1}\left|\bar{\Theta}(\beta,\eta)-\bar{\Theta}(\omega_{\bar{n}})\right| \notag \\
&= \left|\log \eta\right|^{K+P}\left|\bar{\Theta}(\beta,\eta)-\bar{\Theta}(\omega_{\bar{n}})\right| \notag
\end{align}
Thus
\begin{align}
\left|\prod_{j=0}^{\bar{n}}\Theta^{k_{j}}(\alpha,\eta)\bar{\Theta}^{p_{j}}(\beta,\eta)-\prod_{j=0}^{\bar{n}}\Theta^{k_{j}}(\omega_{\bar{n}})\bar{\Theta}^{p_{j}}(\omega_{\bar{n}})\right| 
\leq &\left|\log\eta\right|^{\sum_{j}\left(k_{j}+p_{j}\right)}
\left(\left|\bar{\Theta}(\beta,\eta)-\bar{\Theta}(\omega_{\bar{n}})\right|+
\left|\Theta(\alpha,\eta)-\Theta(\omega_{\bar{n}})\right| \right) \notag \\
\leq &\left|\log\eta\right|^{\frac{n+m}{2}+1}\left(\left|\bar{\Theta}(\beta,\eta)-\bar{\Theta}(\omega_{\bar{n}})\right|+\left|\Theta(\alpha,\eta)-\Theta(\omega_{\bar{n}})\right|\right) \notag
\end{align}
where we have by Eq. (\ref{eq:kpnm})  $\sum_{j}\left(k_{j}+p_{j}\right)\leq\frac{n+m}{2}+1$.
Inserting this in Eq. (\ref{DeltaQ}) we obtain
\begin{align}
\left|\Delta Q_{\pi}\right| &= \left|\log\eta\right|^{\frac{n+m}{2}+1}\lambda^{n+m}
\int d\alpha d\beta 
\prod_{j=0}^{\bar{n}}\int d\omega_{j}  P_{0}(x,\omega_{\bar{n}})   \prod_{j=0}^{\bar{n}}\left(
\left|\frac{1}{\omega_{j}-\alpha -i\eta}\right|^{k_{j}+1}
\left|\frac{1}{\omega_{j}-\beta +i\eta} \right|^{p_{j}+1}\right)\notag \\
\times &\left( \left|\Theta(\alpha,\eta)-\Theta(\omega_{\bar{n}})\right| + \left|\bar{\Theta}(\beta,\eta)-\bar{\Theta}(\omega_{\bar{n}})\right|\right)
\end{align}
We will bound the part including $\left|\Theta(\alpha,\eta)-\Theta(\omega_{\bar{n}})\right|$ as the part with $\left|\bar{\Theta}(\beta,\eta)-\bar{\Theta}(\omega_{\bar{n}})\right|$ can be done analoguesly. We have then by using  inequality (\ref{eq:thetabound})
\begin{align}
\left|\Delta Q_{\pi}\right| \leq& \left|\log\eta\right|^{\frac{n+m}{2}+1}\lambda^{n+m}
\int d\alpha d\beta 
\prod_{j=0}^{\bar{n}}\int d\omega_{j}  P_{0}(x,\omega_{\bar{n}}) 
\prod_{j=0}^{\bar{n}}\left(
\left|\frac{1}{\omega_{j}-\alpha -i\eta}\right|^{k_{j}+1}
\left|\frac{1}{\omega_{j}-\beta +i\eta} \right|^{p_{j}+1}\right) \notag \\
\times &\left|\omega_{\bar{n}}-\alpha-i\eta\right| \left(\frac{1}{\left|1-\alpha-i\eta\right|}+\frac{1}{\left|1-\omega_{\bar{n}}\right|}+
\frac{1}{\left|\alpha+i\eta\right|}+\frac{1}{\left|\omega_{\bar{n}}\right|}\right)
\end{align}
We first bound the term including $\frac{1}{\left|1-\alpha-i\eta\right|}$ by using inequality (\ref{ine:two}) for the integrations over all $\omega_{j}$
except $\omega_{\bar{n}}$ and $\omega_{0}$.
\begin{align}
\mathcal{P}_{1} = &\left|\log\eta\right|^{\frac{n+m}{2}}\lambda^{n+m}
\int d\alpha d\beta 
\prod_{j=0}^{\bar{n}}\int d\omega_{j} P_{0}(x, \omega_{\bar{n}})  
 \prod_{j=0}^{\bar{n}}\left(
\left|\frac{1}{\omega_{j}-\alpha -i\eta}\right|^{k_{j}+1}
\left|\frac{1}{\omega_{j}-\beta +i\eta} \right|^{p_{j}+1}\right)
 \left|\frac{\omega_{\bar{n}}-\alpha-i\eta}{1-\alpha-i\eta}\right| \notag \\
\leq & \left|\log\eta\right|^{\frac{n+m}{2}}\lambda^{n+m}
\eta^{-\sum_{j=0}^{\bar{n}}\left(k_{j}+p_{j}+1\right)+2} \int d\alpha d\beta \int d\omega_{\bar{n}} d\omega_{0} \notag \\
\times & P_{0}(x, \omega_{\bar{n}})
 \left(
\left|\frac{1}{\omega_{0}-\alpha -i\eta}\right|
\left|\frac{1}{\omega_{0}-\beta +i\eta} \right|
\left|\frac{1}{\omega_{\bar{n}}-\alpha -i\eta}\right|
\left|\frac{1}{\omega_{\bar{n}}-\beta +i\eta} \right|
\left|\frac{\omega_{\bar{n}}-\alpha-i\eta}{1-\alpha-i\eta}\right|\right) \notag \\
\leq & \left|\log\eta\right|^{\frac{n+m}{2}}\lambda^{n+m}
 \eta^{-\sum_{j}\left(k_{j}+p_{j}+1\right)+2}
 \notag \\
\times & \int d\alpha d\beta \int d\omega_{\bar{n}} d\omega_{0}
P_{0}(x,\omega_{\bar{n}}) 
 \left(
\left|\frac{1}{\omega_{0}-\alpha -i\eta}\right|
\left|\frac{1}{\omega_{0}-\beta +i\eta} \right|
\left|\frac{1}{\omega_{\bar{n}}-\beta +i\eta} \right|
\left| \frac{1}{1-\alpha-i\eta}\right|\right) \notag 
\end{align}
By bounding succesivly the integrations over $\omega_{\bar{n}}$, $\beta$, $\omega_{0}$ and $\alpha$ by $\left|\log \eta\right|$ we obtain
\begin{align}
\mathcal{P}_{1}
\leq & \left|\log\eta\right|^{\frac{n+m}{2}+5}\lambda^{n+m}
 \eta^{-\sum_{j}\left(k_{j}+p_{j}+1\right)+2}
 \notag 
\end{align}
Setting $\eta=t^{-1}$ and using Eq. (\ref{eq:kpnm}) we obtain
in the van Hove limit ($\lambda^{2}t=T<\infty$)
\begin{align}
\mathcal{P}_{1}\leq & 
\left(CT\right)^{\frac{n+m}{2}}
\frac{\left(\log t\right)^{\frac{n+m}{2}+5}}{t} \notag 
\end{align}
We analyze the term including $\frac{1}{\left|\omega_{\bar{n}}\right|}$, which is the second type of term, by applying the same strategy as for $\mathcal{P}_{1}$
\begin{align}
 \mathcal{P}_{2}=& \left|\log\eta\right|^{\frac{n+m}{2}+1}\lambda^{n+m}
\int d\alpha d\beta 
\prod_{j=0}^{\bar{n}}\int d\omega_{j} P_{0}(x,\omega_{\bar{n}})  
 \prod_{j=0}^{\bar{n}}\left(
\left|\frac{1}{\omega_{j}-\alpha -i\eta}\right|^{k_{j}+1}
\left|\frac{1}{\omega_{j}-\beta +i\eta} \right|^{p_{j}+1}
\right)\left|\frac{\omega_{\bar{n}}-\alpha-i\eta }{\omega_{\bar{n}}}\right| \notag \\
\leq & \left|\log\eta\right|^{\frac{n+m}{2}+1}\lambda^{n+m}
 \eta^{-\sum_{j}\left(k_{j}+p_{j}+1\right)+2}\int d\alpha d\beta \int d\omega_{\bar{n}} d\omega_{0} 
\notag \\
\times &  P_{0}(x,\omega_{\bar{n}})
 \left(
\left|\frac{1}{\omega_{0}-\alpha -i\eta}\right|
\left|\frac{1}{\omega_{0}-\beta +i\eta} \right|
\left|\frac{1}{\omega_{\bar{n}}-\alpha -i\eta}\right|
\left|\frac{1}{\omega_{\bar{n}}-\beta +i\eta} \right|
 \left|\frac{\omega_{\bar{n}}-\alpha-i\eta}{\omega_{\bar{n}}}\right|\right) \notag \\
\leq &\left|\log\eta\right|^{\frac{n+m}{2}+1}\lambda^{n+m}
 \eta^{-\sum_{j}\left(k_{j}+p_{j}+1\right)+2} \int d\alpha d\beta \int d\omega_{\bar{n}} d\omega_{0}\notag \\
\times & P_{0}(x,\omega_{\bar{n}}) 
 \left(
\left|\frac{1}{\omega_{0}-\alpha -i\eta}\right|
\left|\frac{1}{\omega_{0}-\beta +i\eta} \right|
\left|\frac{1}{\omega_{\bar{n}}-\beta +i\eta} \right|
\left|\frac{1}{\omega_{\bar{n}}}\right|\right) \notag 
\end{align}
Once again we bound the integrations over $\alpha$, $\omega_{0}$ and $\beta$ by $\left|\log \eta\right|$ .Since $P_{0}(x,\omega_{\bar{n}})$ is taken to be zero around the edges $0$ and $1$ in some $\epsilon$ interval we have
\begin{align}
 \mathcal{P}_{2}\leq &\left|\log\eta\right|^{\frac{n+m}{2}+4}\lambda^{n+m}
 \eta^{-\sum_{j}\left(k_{j}+p_{j}+1\right)+2}
 \int d\omega_{\bar{n}} P_{0}(x,\omega_{\bar{n}})
\left| \frac{1}{\omega_{\bar{n}}}\right| \notag \\
\leq & \left|\log\eta\right|^{\frac{n+m}{2}+4}\lambda^{n+m}\epsilon^{-1}
 \eta^{-\sum_{j}\left(k_{j}+p_{j}+1\right)+2}
 \notag \\
\leq & \epsilon^{-1}
\left(CT\right)^{\frac{n+m}{2}}\frac{ \left|\log t\right|^{\frac{n+m}{2}+4}}{t}
 \notag 
\end{align}
The remaining terms can be analyzed in a similar manner.
Therefore we have 
\begin{align}
\left|\Delta Q_{\pi(n,m)}(t,\lambda,x)\right|\leq &\left(CT\right)^{\frac{n+m}{2}}
\frac{\left(\log t\right)^{\frac{n+m}{2}+5}}{t} \label{eq:DeltaQ}
\end{align}
From Eqs. (\ref{eq:PatM}), (\ref{eq:PMattilde2}) and (\ref{eq:DeltaQ}) we have then
\begin{align}
\lim_{t\rightarrow \infty}\left|\tilde{P}^{x,M}_{0,t}-P^{x,M}_{0,t}\right|\leq &
\lim_{t\rightarrow \infty}\sum_{n,m=0}^{M}\sum_{\pi(n,m)\in \mathcal{G}_{0}}\left|\Delta Q_{\pi(n,m)}(t,\lambda,x)\right| \notag \\
\leq &\lim_{t\rightarrow \infty}\sum_{n,m=0}^{M}\sum_{\pi(n,m)\in \mathcal{G}_{0}} \left(CT\right)^{\frac{n+m}{2}}
\frac{\left(\log t\right)^{\frac{n+m}{2}+5}}{t} \notag\\
=& 0 \notag 
\end{align}
\end{proof}
\section{EFFECTIVE EQUATION}
\label{sec:effective}
We will now derive the effective equations by calculating the Van Hove limit of $\tilde{P}^{x,M}_{0,t}$.
We recapitulate here our previous results.
In section \ref{sec:crossing} and \ref{sec:nested} we showed that  
\begin{align}
\lim_{N\rightarrow \infty} P^{x,M,N}_{2,t}=0 \\
\lim_{t\rightarrow \infty }\lim_{N\rightarrow \infty} P^{x,M,N}_{1,t}=0 
\end{align}
and in section \ref{sec:simple} we showed that in the Van Hove limit
\begin{align}
\lim_{t\rightarrow \infty}\lim_{N\rightarrow \infty}P_{0,t}^{x,M,N}&=
\lim_{t\rightarrow \infty}\lim_{N\rightarrow \infty}\tilde{P}_{0,t}^{x,M,N} \notag\\
= & \sum_{n,m=0}^{M}\sum_{\pi(n,m)\in \mathcal{G}_{0}}\lim_{t\rightarrow \infty}\tilde{Q}_{\pi(n,m)}(t,\lambda,x) \label{eq:QsimpleP}
\end{align} 
with $\tilde{Q}_{\pi(n,m)}(t,\lambda,x)$ defined through Eq. (\ref{eq:PMattilde2}).
Thus we have
\begin{align}
\lim_{M\rightarrow \infty}\lim_{t\rightarrow \infty}\lim_{N\rightarrow \infty}\mathbb{E}\left[P_{t}^{x,M,N}\right]&=\lim_{M\rightarrow \infty}\sum_{n,m=0}^{M}\sum_{\pi(n,m)\in \mathcal{G}_{0}}\lim_{t\rightarrow \infty}\tilde{Q}_{\pi(n,m)}(t,\lambda,x)
\end{align}
Since a simple graph is completely characterized by the numbers $\bar{n}$, $k_{j}$ and $p_{j}$ 
the sum over all simple graphs $\pi(n,m)$ is a sum over $\bar{n}$, $k_{j}$ and $p_{j}$ such that 
\begin{align}
2 \sum_{j=0}^{\bar{n}}\left(k_{j}+p_{j}+1\right)= & n+m+2 \leq 2M+2 \label{cond1}\\
\sum_{j=0}^{\bar{n}}\left(2p_{j}+1\right)= & n+1 \leq M+1\label{cond2}\\
\sum_{j=0}^{\bar{n}}\left(2k_{j}+1\right)= & m+1 \leq M+1\label{cond3}
\end{align}
The sum 
$\sum_{n,m=0}^{M}\sum_{\pi(n,m)\in \mathcal{G}_{0}} $ in Eq. (\ref{eq:QsimpleP})
is then a sum over $\bar{n}$, $k_{j}$ and $p_{j}$ such that the inequalities are satisfied.
We denote this as follows
\begin{align}
\sum_{n,m=0}^{M}\sum_{\pi(n,m)\in \mathcal{G}_{0}}= \sum_{\bar{n},\{k_{j},p_{j}\}=0}^{c} 
\end{align}
where the superscript $c$ refers to the fact that the conditions of Eqs. (\ref{cond1}), (\ref{cond2}) and (\ref{cond3}) have to be satisfied. In the limit $M\rightarrow \infty$ these conditions are always satisfied and so 
we get 
\begin{align}
\lim_{M \rightarrow\infty}P^{x,M}_{T}
=&\lim_{M\rightarrow \infty}\sum_{n,m=0}^{M}\sum_{\pi(n,m)\in \mathcal{G}_{0}}\lim_{t\rightarrow \infty}\tilde{Q}_{\pi(n,m)}(t,\lambda,x)
\\ 
=& 
\sum_{\bar{n},\{k_{j},p_{j}\}=0}^{\infty}\lim_{t\rightarrow \infty}\tilde{Q}(t,\lambda,x,\bar{n},\{k_{j},p_{j}\})
\end{align}
were $\tilde{Q}(t,\lambda,x,\bar{n},\{k_{j},p_{j}\})$ is given by Eq. (\ref{eq:PMattilde2}) when expressing $n$ and $m$ as a function of $\bar{n}$, $k_{j}$ and $p_{j}$.
\begin{align}
\tilde{Q}(t,\lambda,x,\bar{n},\{k_{j},p_{j}\})=& 
\lambda^{2\bar{n}+\sum_{j=0}^{\bar{n}}\left(k_{j}+p_{j}\right)}
\int d\alpha d\beta 
e^{-i(\alpha -\beta)t}e^{\eta t}\prod_{j=0}^{\bar{n}}\int d\omega_{j} P_{0}(x,\omega_{\bar{n}})  \notag \\
\times &
\prod_{j=0}^{\bar{n}}\left(
\left(\frac{1}{\omega_{j}-\alpha -i\eta}\right)^{k_{j}+1}
\left(\frac{1}{\omega_{j}-\beta +i\eta} \right)^{p_{j}+1} 
\Theta^{k_{j}}(\omega_{\bar{n}})\bar{\Theta}^{p_{j}}(\omega_{\bar{n}})\right) \end{align}
If the the rest, Eq.(\ref{eq:Rest}), vanishes in these limits then we have derived our solution in the limits considered. 
\begin{align}
\lim_{t\rightarrow \infty}\lim_{N\rightarrow \infty} \mathbb{E}\left[P^{x}_{t}\right]=\lim_{M \rightarrow\infty}P^{x,M}_{T}
\end{align}
By the identity
\begin{align}
\left( \frac{1}{\omega_{j}-\alpha-i\eta}\right)^{k_{j}+1}=\frac{i^{k_{j}+1}}{k_{j}!}\int dse^{-i(\omega_{j}-\alpha-i\eta)s}s^{k_{j}} \label{eq:IDomega}
\end{align}
 we get 
\begin{align}
\tilde{Q}&= \lambda^{2\bar{n}+\sum_{j=0}^{\bar{n}}\left(k_{j}+p_{j}\right)}
\int d\alpha d\beta 
e^{-i(\alpha -\beta)t}e^{2\eta t}\prod_{j=0}^{\bar{n}}\int d\omega_{j} P_{0}(x,\omega_{\bar{n}}) \label{Qdefsimple} \\
\times & \prod_{j=0}^{\bar{n}}\left(
\frac{i^{k_{j}+1}}{k_{j}!}\int ds_{j}e^{-i(\omega_{j}-\alpha-i\eta)s_{j}}s_{j}^{k_{j}}\Theta^{k_{j}}(\omega_{\bar{n}})\right)
\prod_{j=0}^{\bar{n}}\left(\frac{\left(-i\right)^{p_{j}+1}}{p_{j}!}\int d\tau_{j}e^{-i(\omega_{j}-\beta+i\eta)\tau_{j}}\tau_{j}^{p_{j}}
\bar{\Theta}^{p_{j}}(\omega_{\bar{n}})\right) \notag
\end{align}
We can sum up over each  $k_{j}$ and $p_{j}$ by grouping the $\lambda^2$, $s_{j}$ and $\Theta$. We obtain
\begin{align}
\tilde{Q}(t,\lambda,x,\bar{n})=&\sum_{k_{j},p_{j}=0}^{\infty}Q(t,\lambda,x,k_{j},p_{j},\bar{n})\notag \\
= &
\sum_{\bar{n}=0}^{\infty}\lambda^{2\bar{n}}
\prod_{j=0}^{\bar{n}}\int d\omega_{j}\int d\alpha d\beta 
e^{-i(\alpha -\beta)t}e^{\eta t} P_{0}(x,\omega_{\bar{n}}) 
 \prod_{j=0}^{\bar{n}}\left(
\int ds_{j}e^{-i(\omega_{j}-\alpha-i\eta)s_{j}}
e^{i s_{j}\lambda^{2}\Theta(\omega_{\bar{n}})}
\right) \notag \\
\times & \prod_{j=0}^{\bar{n}} \left(
\int d\tau_{j}e^{i(\omega_{j}-\beta+i\eta)\tau_{j}}
e^{-i\tau_{j} \lambda^{2}\bar{\Theta}(\omega_{\bar{n}})}
\right) \notag
\end{align}
Integrating over $\alpha$ and $\beta$ we get
\begin{align}
\tilde{Q}(t,\lambda,x,\bar{n})
= &
\lambda^{2\bar{n}}
\prod_{j=0}^{\bar{n}}\int d\omega_{j} P_{0}(x,\omega_{\bar{n}}) 
 \prod_{j=0}^{\bar{n}}\left(
\int ds_{j} \delta\left(t-\sum_{j}s_{j}\right)
e^{-i(\omega_{j})s_{j}}
e^{i s_{j}\lambda^{2}\Theta(\omega_{\bar{n}})}
\right) \notag \\
\times & \prod_{j=0}^{\bar{n}} \left(
\int d\tau_{j}
e^{i(\omega_{j})\tau_{j}}\delta\left(t-\sum_{j}\tau_{j}\right)
e^{-i\tau_{j} \lambda^{2}\bar{\Theta}(\omega_{\bar{n}})}
\right) \notag \\
=& 
\lambda^{2\bar{n}}
\prod_{j=0}^{\bar{n}}\int d\omega_{j} P_{0}(x,\omega_{\bar{n}})e^{i t\lambda^{2}\left(\Theta(\omega_{\bar{n}})-\bar{\Theta}(\omega_{\bar{n}})\right)} \notag \\
\times & \prod_{j=0}^{\bar{n}}\left(
\int ds_{j} \delta\left(t-\sum_{j}s_{j}\right)
e^{-i\omega_{j}s_{j}}
\right) 
 \prod_{j=0}^{\bar{n}} \left(
\int d\tau_{j}
e^{i\omega_{j}\tau_{j}}\delta\left(t-\sum_{j}\tau_{j}\right)
\right) \label{eq:Qtl} 
\end{align}
By the following change of variables
\begin{align}
 a_{j}=\frac{s_{j}+\tau_{j}}{2} \label{eq:a}\\
 b_{j}=\frac{s_{j}-\tau_{j}}{2} \label{eq:b}
\end{align}
we get
\begin{align}
(\ref{eq:Qtl})&= 
\lambda^{2\bar{n}}
\prod_{j=0}^{\bar{n}}\int d\omega_{j} P_{0}(x,\omega_{\bar{n}})e^{i t\lambda^{2}\left(\Theta(\omega_{\bar{n}})-\bar{\Theta}(\omega_{\bar{n}})\right)}
 \prod_{j=0}^{\bar{n}}\left(
\int_{0}^{t} da_{j} \delta\left(t-\sum_{j}a_{j}\right)
\int_{-a_{j}}^{a_{j}} db_{j}
e^{-i\omega_{j}b_{j}}\delta\left(\sum_{j}b_{j}\right)
\right)   \notag
\end{align}
and by the following change of variable and identity 
\begin{align}
 \alpha_{j}&=\lambda^{2}a_{j} \notag\\
  i\left(\Theta(\omega_{\bar{n}})-\bar{\Theta}(\omega_{\bar{n}})\right)&=2Im\left[\lim_{\eta\rightarrow\infty}\int d\omega\left(\frac{-1}{\omega-\omega_{\bar{n}}-i\eta}\right)\right] \notag \\
&= 2\pi
\end{align}
we obtain
\begin{align}
\tilde{Q}(t,\lambda,x,\bar{n})&= 
e^{-2\pi T}\prod_{j=0}^{\bar{n}}\int d\omega_{j} P_{0}(x,\omega_{\bar{n}}) 
 \left(\prod_{j=0}^{\bar{n}}
\int_{0}^{T} d\alpha_{j} \delta\left(T-\sum_{j}\alpha_{j}\right)
  \left( \prod_{j=0}^{\bar{n}}
\int_{-\alpha_{j}\frac{t}{T}}^{\alpha_{j}\frac{t}{T}} db_{j}
e^{-i\omega_{j}b_{j}}\delta\left(\sum_{j}b_{j}\right)
\right) \right)\notag \\ 
&= 
e^{-2\pi T}\prod_{j=0}^{\bar{n}}\int d\omega_{j} P_{0}(x,\omega_{\bar{n}}) 
 \left(\prod_{j=0}^{\bar{n}}
\int_{0}^{T} d\alpha_{j} \delta\left(T-\sum_{j}\alpha_{j}\right)
  \left( \prod_{j=0}^{\bar{n}-1}
\int_{-\alpha_{j}\frac{t}{T}}^{\alpha_{j}\frac{t}{T}} db_{j}
e^{-i\left(\omega_{j}-\omega_{\bar{n}}\right)b_{j}}\chi(b_{j})
\right) \right)\label{eq:almost} 
\end{align}
Where 
\begin{align}
 \chi(b_{j})=\left\{ \begin{array}{ll}
 1 & \text{if $-\alpha_{j}\frac{t}{T}<b_{j}<\alpha_{j}\frac{t}{T}$ and $-\alpha_{0}\frac{t}{T}<\sum_{j=0}^{\bar{n}-1}b_{j}<\alpha_{0}\frac{t}{T}$}\\
0 & \text{else} 
\end{array} \right.\notag
\end{align}
We have then in the limit $t\rightarrow\infty$
\begin{align}
\lim_{t\rightarrow \infty}\prod_{j=0}^{\bar{n}}\int d\omega_{j} P_{0}(x,\omega_{\bar{n}})   
\left( \prod_{j=0}^{\bar{n}-1}
\int_{-\alpha_{j}\frac{t}{T}}^{\alpha_{j}\frac{t}{T}} db_{j}
e^{-i\left(\omega_{j}-\omega_{\bar{n}}\right)b_{j}}\chi(b_{j})
\right)&=\prod_{j=0}^{\bar{n}}\int d\omega_{j} P_{0}(x,\omega_{\bar{n}})   
\prod_{j=0}^{\bar{n}-1} 2\pi\delta(\omega_{j}-\omega_{\bar{n}}) \notag\\
&=\left(2\pi\right)^{\bar{n}}\int d\omega_{\bar{n}} P_{0}(x,\omega_{\bar{n}})
\label{eq:deltaomega} 
\end{align}
Inserting Eq. (\ref{eq:deltaomega}) in Eq. (\ref{eq:almost}) we get 
\begin{align}
\tilde{Q}(T,x,\bar{n})&=\lim_{t \rightarrow \infty}\tilde{Q}(t,\lambda,x,\bar{n}) \notag \\
&= 
e^{-2\pi T}\int d\omega_{\bar{n}} P_{0}(x,\omega_{\bar{n}})  \frac{\left(2\pi T \right)^{\bar{n}}}{\bar{n}!}
 \notag 
\end{align}
We note the following about  $P_{0}(x,\omega_{\bar{n}})$. From Eq. (\ref{cond2}) we see that if $n$ is even then $\bar{n}$ is even and if $n$ is odd so must be $\bar{n}$.
According to Eq. (\ref{eq:Px0}) $P_{0}(x,\omega_{\bar{n}})=P^{x}_{0}(\omega_{\bar{n}})$ if $n$ is even and so also if $\bar{n}$ is even.
$P_{0}(x,\omega_{\bar{n}})=P^{\bar{x}}_{0}(\omega_{\bar{n}})$ if $n$ is odd or if equivalently if $\bar{n}$ is odd. 
Depending on whether $ \bar{n}$ is odd or even  we have
\begin{align}
 \int d\omega_{\bar{n}} P_{0}(x,\omega_{\bar{n}})=\left\{ \begin{array}{ll}
P_{0}(x) & \text{if $\bar{n}$ even}\\
P_{0}(\bar{x}) & \text{if $\bar{n}$ odd}
\end{array} \right.
\end{align}
And so 
\begin{align}
P^{x}_{T}= & \sum_{\bar{n}=0}^{\infty}Q(T,x,\bar{n}) \notag\\
=&e^{-2\pi T}\left(P_{0}(x)\sum_{\bar{n}=0}^{\infty}  \frac{\left(2\pi T \right)^{2\bar{n}}}{(2\bar{n})!}+
P_{0}(\bar{x})\sum_{\bar{n}=0}^{\infty}  \frac{\left(2\pi T \right)^{2\bar{n}+1}}{(2\bar{n}+1)!}
\right) \notag \\
=& \frac{P_{0}(x)+P_{0}(\bar{x})}{2}+\left(P_{0}(x)-
P_{0}(\bar{x})\right)\frac{e^{-4\pi T}}{2}
\end{align}
These solutions satisfy the following rate equations 
\begin{align}
 \frac{d}{dT}P^{1}_{T}&=-4\pi\left(P^{1}_{T}-P^{2}_{T}\right) \\
 \frac{d}{dT}P^{2}_{T}&=-4\pi\left(P^{2}_{T}-P^{1}_{T}\right)
\end{align}
%%%%%%%%%%%%%%%%%%%%%%%%%%%%%%%%%%%%%%%%%%%%%%%%%%%%%%%%%%%%%
%%%%%%%%%%%%%%%%%%%%%%%%%%%%%%%%%%%%%%%%%%%%%%%%%%%%%%%%%%%%%
%%%%%%%%%%%%%%%%%%%%%%%%%%%%%%%%%%%%%%%%%%%%%%%%%%%%%%%%%%%%%
\section{ANALYSIS OF THE ERROR}
\label{sec:error1}
In this section we analyze the error term, Eq. (\ref{eq:Rest}). By the form of Eq. (\ref{eq:Rest}) we see that if the norm of $| \phi^{M+1}_{t} \rangle$ vanishes in the limits considered
\begin{align}
\lim_{M\rightarrow \infty}\lim_{t\rightarrow \infty}\lim_{N\rightarrow \infty}\mathbb{E}\left[\langle \phi^{M+1}_{t}| \phi^{M+1}_{t} \rangle\right]=0 \label{eq:errorlimit}
\end{align}
then the error term will also vanish and this is what we will show.
Up to now we have expanded our solution up to the $M^{th}$ term and derived the equation it would follow when taking $M\rightarrow \infty$. Thus we shall prove that when taking $M\rightarrow \infty$ for the error term this one vanishes. In order to do this we will expand the error term until the $M(t)^{th}$ order, where $M(t)$ now depends on $t$ which is scaled with the coupling constant. 
\begin{align}
 | \phi^{M+1}_{t} \rangle&= \sum_{n=M+1}^{M(t)}|\psi_{t}^{n}\rangle+ | \phi^{M(t)+1}_{t} \rangle \notag \\
&= |\tilde{\phi}_{t}^{M,M(t)}\rangle+ | \phi^{M(t)+1}_{t} \rangle \label{eq:twophi}
\end{align}
To prove Eq. (\ref{eq:errorlimit}) we shall prove that the norm of the two terms in Eq. (\ref{eq:twophi}) vanish.
\begin{align}
\lim_{M\rightarrow \infty}\lim_{t\rightarrow \infty}\lim_{N\rightarrow \infty}\mathbb{E}\left[\langle \tilde{\phi}^{M,M(t)}_{t}| \tilde{\phi}^{M,M(t)}_{t} \rangle\right]&=0 \label{eq:MMt}\\
\lim_{M\rightarrow \infty}\lim_{t\rightarrow \infty}\lim_{N\rightarrow \infty}\mathbb{E}\left[\langle \phi^{M(t)+1}_{t}| \phi^{M(t)+1}_{t} \rangle\right]&=0 \label{eq:Mt}
\end{align}
We focus now on proving Eq. (\ref{eq:MMt}). Since $|\tilde{\phi}^{M,M(t)}_{t} \rangle$ 
is a sum of $|\psi_{t}^{n}\rangle$ vectors we can write it down as a function of $Q_{\pi(n,m)}$ function.
 We already have some usefull bounds on the different type of graphs.
We will use the bounds of  Eq. (\ref{eq:Qnestedbound}) on nested graphs and also Eq. (\ref{eq:DeltaQ}) for a part of the simple graphs.
We will thus look to bound the remaining part of simple graphs.
Thus we turn to bound $\tilde{Q}_{\pi(n,m)}$ from Eq. (\ref{eq:PMattilde2}). 
\begin{thm}
\label{thm:Qtildenm}
For simple graphs we have the following bound for $\tilde{Q}_{\pi(n,m)}(t,\lambda,x)$
defined in Eq. (\ref{eq:PMattilde2}):
 \begin{align}
\left|  \tilde{Q}_{\pi(n,m)}(t,\lambda,x)\right|\leq \frac{\left(C\lambda^{2}t\right)^{\frac{n+m}{2}}}{\left(\frac{n+m}{2}!\right)^{a}}
 \end{align}
with $a<\frac{1}{2}$.
\end{thm}
\begin{proof}{Theorem \ref{thm:Qtildenm}}\\
From Eq. (\ref{eq:PMattilde2}) we have 
\begin{align}
&\left|\tilde{Q}_{\pi(n,m)} (t,\lambda,x)\right|\notag \\
= &\left|\lambda^{n+m}
\prod_{j=0}^{\bar{n}}\int_{0}^{1} d\omega_{j}\int_{-\infty}^{\infty} d\alpha d\beta 
e^{-i(\alpha -\beta)t}e^{\eta t} P_{0}(x,\omega_{\bar{n}})
\Theta^{k_{j}}(\omega_{\bar{n}})\bar{\Theta}^{p_{j}}(\omega_{\bar{n}})
\prod_{j=0}^{\bar{n}}\left(
\left(\frac{1}{\omega_{j}-\alpha -i\eta}\right)^{k_{j}+1}
\left(\frac{1}{\omega_{j}-\beta +i\eta} \right)^{p_{j}+1} 
\right)\right| \notag  \\
\leq &\lambda^{n+m}
\prod_{j=0}^{\bar{n}}\int_{0}^{1} d\omega_{j}
\left|\int_{-\infty}^{\infty} d\alpha d\beta 
e^{-i(\alpha -\beta)t}e^{\eta t} P_{0}(x,\omega_{\bar{n}})
\prod_{j=0}^{\bar{n}}\left(
\left(\frac{1}{\omega_{j}-\alpha -i\eta}\right)^{k_{j}+1}
\left(\frac{1}{\omega_{j}-\beta +i\eta} \right)^{p_{j}+1} 
\right)\right|^{1-a} \notag  \\
\times &\left|\int_{-\infty}^{\infty} d\alpha d\beta 
e^{-i(\alpha -\beta)t}e^{\eta t} P_{0}(x,\omega_{\bar{n}})
\prod_{j=0}^{\bar{n}}\left(
\left(\frac{1}{\omega_{j}-\alpha -i\eta}\right)^{k_{j}+1}
\left(\frac{1}{\omega_{j}-\beta +i\eta} \right)^{p_{j}+1} 
\right)\right|^{a} \left|\Theta^{k_{j}}(\omega_{\bar{n}})\bar{\Theta}^{p_{j}}(\omega_{\bar{n}})\right|
\label{eq:Qerror}
\end{align}
With the following relations fulfilled:
\begin{align}
\sum_{j=0}^{\bar{n}}\left(k_{j}+p_{j}+2\right)+n'&=n+m+2 \notag \\
\bar{n}+1+n'&=\frac{n+m}{2}+1 \notag
\end{align}
Therefore we have $\sum_{j=0}^{\bar{n}}\left(k_{j}+p_{j}+1\right)=\frac{n+m}{2}+1$.
 The $\Theta$ functions are bounded by a constant and so have no importance.
Since we have 
\begin{align}
\left|\int_{-\infty}^{\infty} d\alpha 
e^{-i\alpha t}e^{\eta t}
\prod_{j=0}^{\bar{n}}\left(
\left(\frac{1}{\omega_{j}-\alpha -i\eta}\right)^{k_{j}+1}
\right)\right|\leq \frac{t^{\sum_{j=0}^{\bar{n}}k_{j}+\bar{n}}}{\left(\sum_{j=0}^{\bar{n}}k_{j}+\bar{n}\right)!}
\end{align}
 we can easily see that the following bound holds
\begin{align}
 & \left|\int_{-\infty}^{\infty} d\alpha d\beta 
e^{-i(\alpha -\beta)t}e^{\eta t} P_{0}(x,\omega_{\bar{n}})
\prod_{j=0}^{\bar{n}}\left(
\left(\frac{1}{\omega_{j}-\alpha -i\eta}\right)^{k_{j}+1}
\left(\frac{1}{\omega_{j}-\beta +i\eta} \right)^{p_{j}+1} 
\right)\right|^{a} & \notag \\
\leq & \frac{\left(t\right)^{\left(\sum_{j}\left(k_{j}+p_{j}\right)+2\bar{n}\right)a}}{
\left(\left(\sum_{j=0}^{\bar{n}}k_{j}+\bar{n}\right)!\left(\sum_{j=0}^{\bar{n}}p_{j}+\bar{n}\right)!\right)^{a}} \notag \\
\leq & \frac{\left(Ct\right)^{\left(\sum_{j}\left(k_{j}+p_{j}\right)+2\bar{n}\right)a}}{
\left(\sum_{j}\left(k_{j}+p_{j}\right)+2\bar{n}\right)!^{a}} \notag \\
\leq & \frac{\left(Ct\right)^{\left(\sum_{j}\left(k_{j}+p_{j}\right)+2\bar{n}\right)a}}{
\left(\frac{n+m}{2}\right)!^{a}} \label{eq:powera}
\end{align}
We use this to bound the second line of Eq. (\ref{eq:Qerror}).
We bound the remaining in Eq. (\ref{eq:Qerror}) by bounding the integrals over $\omega_{j}$ as follows.
\begin{align}
&\left|\lambda^{n+m}
\prod_{j=0}^{\bar{n}}\int_{0}^{1} d\omega_{j}
\left(\int_{-\infty}^{\infty} d\alpha d\beta 
e^{-i(\alpha -\beta)t}e^{\eta t} P_{0}(x,\omega_{\bar{n}})
\prod_{j=0}^{\bar{n}}\left(
\left(\frac{1}{\omega_{j}-\alpha -i\eta}\right)^{k_{j}+1}
\left(\frac{1}{\omega_{j}-\beta +i\eta} \right)^{p_{j}+1} 
\right)\right)\right|^{1-a} \notag  \\
\leq &\lambda^{n+m}
\prod_{j=0}^{\bar{n}}\int_{0}^{1} d\omega_{j}
\int_{-\infty}^{\infty} d\alpha d\beta
\prod_{j=0}^{\bar{n}}
\left|\frac{1}{\omega_{j}-\alpha -i\eta}\right|^{\left(k_{j}+1\right)\left(1-a\right)} 
\left|\frac{1}{\omega_{j}-\beta +i\eta} \right|^{\left(p_{j}+1\right)\left(1-a\right)} \label{eq:nextlineone}
\end{align}
We bound $\left|\frac{1}{\omega_{j}-\alpha -i\eta}\right|^{\left(k_{j}\right)\left(1-a\right)} 
\left|\frac{1}{\omega_{j}-\beta +i\eta} \right|^{\left(p_{j}\right)\left(1-a\right)}$ %$\left(k_{j}+p_{j}\right)(1-a)$
by $\left(\frac{1}{\eta}\right)^{\left(k_{j}+p_{j}\right)(1-a)}$ and for $j\geq 2$ we use Eq. (\ref{eq:oneminusa}) to bound the integration over $\omega_{j}$ of the remaining propagators. That is in applying Eq. (\ref{eq:oneminusa})in this case we have using $\delta=2\left(1-a\right)$ in Eq. (\ref{eq:oneminusa}).
\begin{align}
(\ref{eq:nextlineone})\leq &\lambda^{n+m}
\int d\omega_{0}d\omega_{1}
\int_{-\infty}^{\infty} d\alpha d\beta
\left|\frac{1}{\omega_{0}-\alpha -i\eta}\right|^{1-a} 
\left|\frac{1}{\omega_{0}-\beta +i\eta} \right|^{1-a} 
\left|\frac{1}{\omega_{1}-\alpha -i\eta}\right|^{1-a} 
\left|\frac{1}{\omega_{1}-\beta +i\eta} \right|^{1-a} \notag \\
\times &\left(\frac{1}{\eta}\right)^{\left(k_{0}+p_{0}+k_{1}+p_{1}\right)\left(1-a\right)}
\prod_{j=2}^{\bar{n}}
\left(\frac{1}{\eta}\right)^{\left(k_{j}+p_{j}+2\right)\left(1-a\right)-1} \notag 
\end{align}
For $\omega_{0}$ and $\omega_{1}$ we use first Eq. (\ref{eq:ABdelta}) and then 
Eq. (\ref{eq:oneminusa})
\begin{align}
(\ref{eq:nextlineone})\leq &\lambda^{n+m}\left(\frac{1}{\eta}\right)^{1-2a}
\left(\frac{1}{\eta}\right)^{\left(k_{0}+p_{0}+k_{1}+p_{1}\right)\left(1-a\right)}
\prod_{j=2}^{\bar{n}}
\left(\frac{1}{\eta}\right)^{\left(k_{j}+p_{j}+2\right)\left(1-a\right)-1} \notag \\
\leq &\lambda^{n+m}
\left(\frac{1}{\eta}\right)^{-a\left(\sum_{j=0}^{\bar{n}}\left(k_{j}+p_{j}\right)+2\bar{n}\right)+\sum_{j}\left(k_{j}+p_{j}+1\right)-1} \label{eq:powera2}
\end{align}
Combining this estimate with Eq. (\ref{eq:powera}) in Eq. (\ref{eq:Qerror}) we get 
\begin{align}
\left|\tilde{Q}_{\pi(n,m)} (t,\lambda,x)\right|\leq \frac{\left(C\lambda^2t\right)^{\frac{n+m}{2}}}{\frac{n+m}{2}!^{a}}\label{eq:Qboundeda}
\end{align}
\end{proof}
We will now bound
$
\mathbb{E}\left[\langle\tilde{\phi}^{M,M(t)}_{t}|\tilde{\phi}^{M,M(t)}_{t}\rangle\right] 
$
with
\begin{align}
|\tilde{\phi}^{M,M(t)}_{t}\rangle= \sum_{n=M+1}^{M(t)}|\psi_{t}^{n}\rangle
\end{align}
We have 
\begin{align}
\mathbb{E}\left[\langle\tilde{\phi}^{M,M(t)}_{t}|\tilde{\phi}^{M,M(t)}_{t}\rangle\right]&=\sum_{n,m=M+1}^{M(t)}\mathbb{E}\left[\langle \psi^{m}_{t}|\sum_{x}\hat{P}^{x}|\psi^{n}_{t}\rangle\right] \notag \\
&=\sum_{n,m=M+1}^{M(t)}
\sum_{\pi(n,m)}
\sum_{x=1,2} Q^{N}_{\pi(n,m)}(t,\lambda,x) \notag
\end{align}
Similar to the proof in (\ref{lem:CrossContribution}) we have that the contributions from crossing graphs will vanish. Thus we obtain
\begin{align}
&\lim_{N\rightarrow\infty}\mathbb{E}\left[\langle\tilde{\phi}^{M,M(t)}_{t}|\tilde{\phi}^{M,M(t)}_{t}\rangle\right] 
=\sum_{n,m=M+1}^{M(t)}\left(\sum_{\pi(n,m)\in \mathcal{G}_{0}}\sum_{x}Q_{\pi(n,m)} (t,\lambda,x)
+\sum_{\pi(n,m)\in \mathcal{G}_{1}}\sum_{x}Q_{\pi(n,m)} (t,\lambda,x) \right) \notag \\
=&\sum_{n,m=M+1}^{M(t)}\left(
\sum_{\pi(n,m)\in \mathcal{G}_{0}}\sum_{x}\left(\tilde{Q}_{\pi(n,m)} (t,\lambda,x)+\Delta Q_{\pi(n,m)} (t,\lambda,x) \right)+\sum_{\pi(n,m)\in \mathcal{G}_{1}}\sum_{x}Q_{\pi(n,m)} (t,\lambda,x)\right) \notag 
\end{align}
Using the bounds of Eqs. (\ref{eq:Qnestedbound}), (\ref{eq:DeltaQ}) and  (\ref{eq:Qboundeda}) we get
\begin{align}
\lim_{N\rightarrow\infty}\mathbb{E}\left[\langle\tilde{\phi}^{M,M(t)}_{t}|\tilde{\phi}^{M,M(t)}_{t}\rangle\right] \leq&\sum_{n,m=M+1}^{M(t)}\left((CT)^{\frac{n+m}{2}}\frac{\log^{\frac{n+m}{2}+2}(t)}{t} + \frac{\left(CT\right)^{\frac{n+m}{2}}}{\left(\frac{n+m}{2}\right)!^{a}} +(CT)^{\frac{n+m}{2}}\frac{\log^{\frac{n+m}{2}+5}(t)}{t}\right) \notag \\
\leq&\sum_{n,m=M+1}^{M(t)}\left(  \frac{\left(C\lambda^{2}t\right)^{\frac{n+m}{2}}}{\left(\frac{n+m}{2}\right)!^{a}}\right) 
+(CT)^{M(t)}\frac{\log^{M(t)+5}(t)}{t} M(t)^{2}
\end{align}
We choose $M(t)=\gamma\frac{\log(t)}{\log\log t}$ with $\gamma <1$ and take $a=\frac{1}{4}$. 
We also set $\log t=x$. We have then 
\begin{align}
\lim_{N\rightarrow \infty}
\mathbb{E}\left[\langle\tilde{\phi}^{M,M(t)}_{t}|\tilde{\phi}^{M,M(t)}_{t}\rangle\right] \leq &
 \sum_{n,m=M+1}^{\infty}\frac{\left(CT\right)^{\frac{n+m}{2}}}{\frac{n+m}{2}!^{\frac{1}{4}}}+\frac{\gamma^{2}x^2}{\log^{2}x}\frac{x^{\gamma\frac{x}{\log x}+5}}{e^{x}}\notag\\
\end{align}
For large enough $M$ we have then 
\begin{align}
\lim_{N\rightarrow \infty} \mathbb{E}\left[\langle\tilde{\phi}^{M,M(t)}_{t}|\tilde{\phi}^{M,M(t)}_{t}\rangle\right] \leq &
\frac{\left(CT\right)^{M}}{M!^{\frac{1}{4}}}
+\frac{\gamma^{2}x^2}{\log^{2}x}
e^{\gamma x -x+4\log x}\notag\\
\end{align}
and so 
\begin{align}
\lim_{M\rightarrow \infty }\lim_{x\rightarrow \infty}\lim_{N\rightarrow \infty}\mathbb{E}\left[\langle\tilde{\phi}^{M,M(t)}_{t}|\tilde{\phi}^{M,M(t)}_{t}\rangle\right] =&0
\end{align}
%%%%%%%%%%%%%%%%%%%%%%%%%%%%%%%%%%%%%%%%%%%%%%%%%%%%%%%%%%%%%%
What we have left to bound is the average of $|\phi_{t}^{M(t)}\rangle$.
We shall drop the $t$ dependency of $M$  for now.
According to Eq. (\ref{eq:gammatilde}) we have
\begin{align}
 |\phi^{M+1}_{t}\rangle &= \tilde{\Gamma}_{M+1}(t)|\psi_{0}\rangle \notag \\
&=-i\lambda \int_{0}^{t} d s e^{-iH(t-s)}V|\psi_{s}^{M}\rangle
\end{align}
We will follow \cite{EY00} in bounding this term. That is we will divide the time integration in $\kappa$ parts, where $\kappa$ will eventually depend on $t$, and expand each piece of the time integrations once again using the Duhamel formula, Eq. (\ref{eq:duhamel}). We will thus extract again a term which is a succession of free evolutions and one which will depend on the whole evolution.
We have then 
\begin{align}
 |\phi^{M+1}_{t}\rangle 
&=-i\lambda\sum_{j=1}^{\kappa}e^{-i\left(t-\theta_{j+1}\right)H} \int_{\theta_{j}}^{\theta_{j+1}} d s e^{-iH(\theta_{j+1}-s)}V|\psi_{s}^{M}\rangle
\end{align}
Where 
\begin{align}
 \theta_{j}&=\frac{j t}{\kappa} \label{eq:thetaj} \\
 \theta_{j+1}-\theta_{j}&=\frac{t}{\kappa} \label{eq:thetajdif}
\end{align}
We have the following expansion for each $e^{-i(\theta_{j}-s)H}$ from Eq. (\ref{eq:gammaprop}):
\begin{align}
e^{-i(\theta_{j}-s)H}=\sum_{n=0}^{M_{0}}\left(-i\lambda\right)^{n}\Gamma_{n}(\theta_{j}-s)+\tilde{\Gamma}_{M_{0}+1}(\theta_{j}-s) 
\end{align}
and so 
\begin{align}
 |\phi^{M+1}_{t}\rangle &= |\psi^{1}_{M,M_{0},\kappa}(t)\rangle+|\psi^{2}_{M,M_{0},\kappa}(t)\rangle \notag \\
|\psi^{1}_{M,M_{0},\kappa}(t)\rangle
&=-i\lambda\sum_{j=1}^{\kappa}\sum_{n=0}^{M_{0}}e^{-i\left(t-\theta_{j+1}\right)H} \int_{\theta_{j}}^{\theta_{j+1}} d s \left(-i\lambda\right)^{n}\Gamma_{n}(\theta_{j+1}-s)V|\psi_{s}^{M}\rangle \label{eq:psi1}\\
|\psi^{2}_{M,M_{0},\kappa}(t)\rangle&=-i\lambda\sum_{j=1}^{\kappa}e^{-i\left(t-\theta_{j+1}\right)H}\int_{\theta_{j}}^{\theta_{j+1}} d s \tilde{\Gamma}_{M_{0}+1}(\theta_{j+1}-s)V|\psi_{s}^{M}\rangle \label{eq:psi2}
\end{align}
$|\psi^{2}_{M,M_{0},\kappa}(t)\rangle$ has $M+M_{0}+2$ products of random matrices.
We define
\begin{align}
 |\psi_{M,n,\kappa,\theta_{j}}(\tilde{s})\rangle=& \left(-i\lambda\right)^{n+1}\int_{\theta_{j}}^{\tilde{s}}ds\int_{0}^{\tilde{s}-s}
 \left[ds_{n}\right] e^{-is_{0}H_{0}}V\dots e^{-is_{n}H_{0}}\delta\left(\tilde{s}-s-\sum_{j=0}^{n}s_{j}\right)V|\psi_{s}^{M}\rangle \label{psiaux} \\
=& \left(-i\lambda\right)^{M+n+1}\int_{\theta_{j}}^{\tilde{s}}ds
\Gamma_{n}(\tilde{s}-s)
V\Gamma_{M}(s)|\psi_{0}\rangle
\end{align}
$ |\psi_{M,n,\kappa,\theta_{j}}(\tilde{s})\rangle$ has $M+n+1$ random matrices and 
$n+M+2$ propagators.
With the definition of Eq. (\ref{psiaux}) we can rewrite Eqs. (\ref{eq:psi1}) and (\ref{eq:psi2}) as
\begin{align}
|\psi^{1}_{M,M_{0},\kappa}(t)\rangle
&=\sum_{j=0}^{\kappa}\sum_{n=0}^{M_{0}}e^{-i\left(t-\theta_{j+1}\right)H}  |\psi_{M,n,\kappa,\theta_{j}}(\theta_{j+1})\rangle \label{psi1}\\
|\psi^{2}_{M,M_{0},\kappa}(t)\rangle&=-i\lambda\sum_{j=0}^{\kappa}e^{-i\left(t-\theta_{j+1}\right)H}\int_{\theta_{j}}^{\theta_{j+1}} d \tilde{s}  e^{-i(\theta_{j+1}-\tilde{s})H}V |\psi_{M,M_{0},\kappa,\theta_{j}}(\tilde{s}) \rangle \label{psi2}
\end{align}
We first bound $|\psi^{2}_{M,M_{0},\kappa}(t)\rangle$ through the following theorem:
\begin{thm}
\label{thm:psi2}
We have the following bound for the norm of $|\psi^{2}_{M,M_{0},\kappa}(t)\rangle$ in the limit $N\rightarrow \infty$:
\begin{align}
\lim_{N\rightarrow \infty}\mathbb{E}\left[\langle \psi^{2}_{M,M_{0},\kappa}(t)|\psi^{2}_{M,M_{0},\kappa}(t)\rangle\right] \leq & \left(C\lambda^2t\right)^{M+M_{0}+1}\frac{t
\log^{M+M_{0}+6}(t)}{\kappa^{M_{0}-M-1}}
\end{align}
\end{thm}
\begin{proof}{Theorem \ref{thm:psi2}}\\
From Eq. (\ref{psi2}) we have 
\begin{align}
\langle \psi^{2}_{M,M_{0},\kappa}(t) |\psi^{2}_{M,M_{0},\kappa}(t)\rangle
&\leq 
\lambda^{2}\sum_{j=1}^{\kappa}\int_{\theta_{j}}^{\theta_{j+1}} d \tilde{s} 
\sum_{l=1}^{\kappa}\int_{\theta_{l}}^{\theta_{l+1}} d \tilde{s}' 
\langle\psi_{M,M_{0},\kappa,\theta_{j}}(\tilde{s}')|
Ve^{i\left(t-\tilde{s}'\right)H}  
e^{-i\left(t-\tilde{s}\right)H}V |\psi_{M,M_{0},\kappa,\theta_{j}}(\tilde{s}) \rangle
\notag\\
&\leq\lambda^{2}t^{2}\text{supp}_{\theta_{j},\tilde{s}}\{
\langle\psi_{M,M_{0},\kappa,\theta_{j}}(\tilde{s}) |V^{2}|\psi_{M,M_{0},\kappa,\theta_{j}}(\tilde{s}) \rangle\} \label{vsquared}
\end{align}
We can once again rewrite the average, $\mathbb{E}\left[\right\langle\psi_{M,M_{0},\kappa,\theta_{j}}(\tilde{s}) |V^{2}|\psi_{M,M_{0},\kappa,\theta_{j}}(\tilde{s}) \rangle]$, as a sum over graph evaluated functions starting from Eq. (\ref{psiaux}). 
In addition to the $2(M+M_{0}+1)$ random matrices that come from the expansion we have $2$ random matrices.
When inserting 
Eq. (\ref{psiaux}) in  
Eq. (\ref{vsquared})  the resulting expression has $2(M_{0}+M+2)$ random matrices but $2(M_{0}+M+2)$ propagators. 
In our previous sections and definitions of $Q_{\pi(n,m)}$ we had, for the expansion of the order $n+m$, $n+m$ random matrices and $n+m+2$ propagators.
Since now we have $2$ extra random matrices the number of random matrices equals the number of propagators.
Analoguesly to how it was done in section \ref{sec:nested} and \ref{sec:simple} we can introduce a $Q_{\pi(M_{0}+M+2,M_{0}+M+2)}(\theta_{j},\tilde{s},\lambda)$ function that encodes the contribution of the graph $\pi(M_{0}+M+2,M_{0}+M+2)$ to the average. The fact that we have $2$ extra random matrices will modify a bit the relationships we had.
We can use the $\alpha$-representation two times in Eq. (\ref{psiaux}), one for the explicit $\delta$ function and one for the delta function in $|\psi^{M}_{s}\rangle$.
For the explicit one we have
\begin{align}
 \delta(\tilde{s}-s-\sum_{j=0}^{n}s_{j})=\int d\tilde{\alpha} e^{-i\tilde{\alpha}\left(\tilde{s}-s-\sum_{j}s_{j}\right)+\tilde{\eta}\left(\tilde{s}-s-\sum_{j}s_{j}\right)}
\end{align}
We have then
\begin{align}
\langle \tilde{E}_{0},\tilde{x}_{0} |\psi_{M,M_{0},\kappa,\theta_{j}}(\tilde{s})\rangle=& \left(-i\lambda\right)^{M+M_{0}+1}\int_{\theta_{j}}^{\tilde{s}}ds 
\int_{-\infty}^{\infty} d\tilde{\alpha}\int_{-\infty}^{\infty} d\alpha 
e^{-i\tilde{\alpha}\left(\tilde{s}-s\right)}e^{\tilde{\eta}\left(\tilde{s}-s\right)}
e^{-i\alpha s}e^{\eta s} \notag \\
&\sum_{\tilde{E_{1}}\dots E_{M+1},\tilde{x}_{1},\dots x_{M+1}}\frac{1}{\tilde{E}_{0}-\tilde{\alpha}-i\tilde{\eta}} \dots \frac{1}{\tilde{E}_{M_{0}}-\tilde{\alpha}-i\tilde{\eta}} 
\frac{1}{E_{0}-\alpha-i\eta} \dots \frac{1}{E_{M}-\alpha-i\eta} \notag \\
& \langle \tilde{E}_{0},\tilde{x}_{0}|V|\tilde{E}_{1},\tilde{x}_{1}\rangle \dots \langle E_{M},x_{M}|V|E_{M+1},x_{M+1}\rangle\psi_{0}\left(E_{M+1},x_{M+1}\right)\notag \\
=&\left(-i\lambda\right)^{M+M_{0}+1}\int_{\theta_{j}}^{\tilde{s}}ds\sum_{\tilde{E_{1}}\dots E_{M+1},\tilde{x}_{1},\dots x_{M+1}} K^{M_{0}}\left(\tilde{s}-s,\{\tilde{E}_{j}\}\right)
K^{M}\left(s,\{E_{j}\}\right) \notag \\
& \langle \tilde{E}_{0},\tilde{x}_{0}|V|\tilde{E}_{1},\tilde{x}_{1}\rangle \dots \langle E_{M},x_{M}|V|E_{M+1},x_{M+1}\rangle\psi_{0}\left(E_{M+1},x_{M+1}\right)
\end{align}
and a similar expression for $\langle \psi_{M,M_{0},\kappa,\theta_{j}}(\tilde{s})|\tilde{E}'_{0},\tilde{x}'_{0} \rangle$, where $\beta$ will stand for $\alpha$ and $\tilde{\beta}$ for $\tilde{\alpha}$.
Thus
\begin{align}
& \mathbb{E}\left[\langle\psi_{M,M_{0},\kappa,\theta_{j}}(\tilde{s}) |V^{2}|\psi_{M,M_{0},\kappa,\theta_{j}}(\tilde{s}) \rangle\right] \notag \\
=&
\lambda^{2\left(M+M_{0}+1\right)}\int_{\theta_{j}}^{\tilde{s}}\int_{\theta_{j}}^{\tilde{s}}d\tau ds \sum 
\bar{K}^{M_{0}}\left(\tilde{s}-\tau,\{\tilde{E}'_{j}\}\right)
\bar{K}^{M}\left(\tau,\{E'_{j}\}\right)
K^{M_{0}}\left(\tilde{s}-s,\{\tilde{E}_{j}\}\right)
K^{M}\left(s,\{E_{j}\}\right) \notag \\
&
\mathbb{E}\left[\langle E'_{M+1},x'_{M+1}|V|E'_{1},x'_{1}\rangle \dots \langle \tilde{E}'_{1},\tilde{x}'_{1}|V|\tilde{E}_{0},\tilde{x}'_{0}\rangle
\langle \tilde{E}'_{0},\tilde{x}'_{0}|V^{2}|\tilde{E}_{0},\tilde{x}_{0}\rangle
\langle \tilde{E}_{0},\tilde{x}_{0}|V|\tilde{E}_{1},\tilde{x}_{1}\rangle \dots \langle E_{M},x_{M}|V|E_{M+1},x_{M+1}\rangle\right] \notag \\
& \psi_{0}^{*}\left(\tilde{E}'_{M+1},\tilde{x}'_{M+1}\right)\psi_{0}\left(E_{M+1},x_{M+1}\right) \label{eq:long}
\end{align}
Where the sum is over all $E_{j}$,$\tilde{E}_{j}$,$E'_{j}$ and $\tilde{E}'_{j}$ variables.
In the limit $N\rightarrow \infty$ crossing graphs will once again not contribute because each one of them has a weight less then or equal to $N^{-2}$.
We have then 
\begin{align}
\lim_{N\rightarrow \infty} \mathbb{E}\left[\langle\psi_{M,M_{0},\kappa,\theta_{j}}(\tilde{s}) |V^{2}|\psi_{M,M_{0},\kappa,\theta_{j}}(\tilde{s}) \rangle\right]
&=\sum_{\pi(M+M_{0}+1,M+M_{0}+1)\in \mathcal{G}_{0,1}} Q_{\pi(M+M_{0}+1,M+M_{0}+1)}(\theta_{j},\tilde{s},\lambda) \label{eq:AVncQ}
\end{align}
For shortness of notation we refer to $Q_{\pi(M+M_{0}+1,M+M_{0}+1)}(\theta_{j},\tilde{s},\lambda)$ as $Q_{\pi}$.
Propagators depending on $\alpha$ come from the right $|\psi^{M}_{s}\rangle$ and those depending on $\tilde{\alpha}$ come from the right $\tilde{\Gamma}_{M_{0}+1}$.
Propagators depending on $\beta$ come from the left $|\psi^{M}_{s}\rangle$ and those depending on $\tilde{\beta}$ come from the left $\tilde{\Gamma}_{M_{0}+1}$.
There are thus $M+1$ propagators depending on $ \alpha$, $M_{0}+1$ depending on $\tilde{\alpha}$, $M+1$  depending on $\beta$ and $M_{0}+1$ depending on $\tilde{\beta}$.
When averaging in Eq. (\ref{eq:long}) and taking only non crossing graphs we will once again have that the number of independent variables is half of the length of the graph plus $1$ (theorem \ref{thm:kappa}). That is $M_{0}+M+3$. The number of independent energy variables will then be equal to $M_{0}+M+3$. Nevertheless we notice that  not all independent energy variables must have a set of propagators associated.
Previously we had in between each random matrix a propagator which meant that each energy variables (dependent or independent) was associated with a propagator. We see from Eq. (\ref{eq:long}) that there is no propagator in between the $V^{2}$ and so if the graphs is such that the variables in between 
this product is independent it will have no propagator associated. Therefore the sum or integration over this variable will be $1$ and so we could omit it. 
Therefore, depending on the graph, the number of independent energy variables can be either $M_{0}+M+3$ or $M_{0}+M+2$. 
We have then as in section \ref{sec:simple} and \ref{sec:nested}
\begin{align}
Q_{\pi}=& \lambda^{2(M+M_{0}+1)}
\int_{\theta_{j}}^{\tilde{s}} d\tau ds\int d\alpha d\tilde{\alpha}d\beta d\tilde{\beta} e^{-i(\alpha+i\eta)s}
e^{-i(\tilde{\alpha}+i\tilde{\eta})(\tilde{s}-s)}e^{i(\beta-i\eta)\tau}e^{i(\tilde{\beta}-i\tilde{\eta})(\tilde{s}-\tau)} \notag\\
\times& \prod_{j=1}^{\bar{n}}\int d\omega_{j}\left(\frac{1}{\omega_{j}-\tilde{\alpha}-i\tilde{\eta}}\right)^{a_{j}}
\left(\frac{1}{\omega_{j}-\alpha-i\eta}\right)^{b_{j}}
\left(\frac{1}{\omega_{j}-\beta+i\eta}\right)^{c_{j}} 
\left(\frac{1}{\omega_{j}-\tilde{\beta}+i\tilde{\eta}}\right)^{d_{j}}
\notag \\
\times& \prod_{j=1}^{n'}\int d\omega'_{j} \left(\frac{1}{\omega'_{j}-\gamma_{j}-i\eta_{j}}\right) \label{eq:Qabcd}
\end{align}
where
the graph $\pi(M+M_{0}+2,M+M_{0}+2)$ determines the multiplicities $a_{j}$, $b_{j}$, $c_{j}$ and $d_{j}$. The propagators of multiplicity one are dependent on $\omega'_{j}$ and there are $n'$ of them. $\gamma_{j}$ can take on the values $\alpha$, $\tilde{\alpha}$, $\beta$ or $\tilde{\beta}$. This dependents on where the propagator is located and thus depends on the graph. $\eta_{j}$ can take on the values $\eta$, $\tilde{\eta}$, $-\eta$ or $-\tilde{\eta}$ depending on which value $\gamma_{j}$ take on. 
The following relations have to be satisfied:
\begin{align}
 \sum_{j=0}^{\bar{n}}\left(a_{j}+b_{j}+c_{j}+d_{j}\right)+n'&=2\left(M+M_{0}+2\right) \label{consprop}\\
2(M_{0}+1)&\leq \sum_{j=0}^{\bar{n}}\left(a_{j}+d_{j}\right)+n' \label{consacprop} \\
M+M_{0}+2&\leq n'+\bar{n}+1\leq M+M_{0}+3 \label{consnn}
\end{align}
Eq.(\ref{consprop}) expresses that fact that there are $2(M_{0}+M+2)$ propagators.
Eq.(\ref{consacprop}) expresses that fact that there are $2M_{0}+2$ propagators depending on $\tilde{\alpha}$ and $\tilde{\beta}$. Since $n'$ counts all propagators with multiplicity equal to $1$ there is an inequality sign. 
Eq.(\ref{consnn}) expresses that fact that the number of independent variables 
varies between two possibilities as explained earlier.
We set the following:
\begin{align}
\eta&= t^{-1} \label{eta1}\\
\tilde{\eta}&=\left(\theta_{j+1}-\theta_{j}\right)^{-1}\geq t^{-1}\label{eta2}
\end{align}
This choice guarantees that the exponentials in Eq. (\ref{eq:Qabcd}) do not diverge since $\theta_{j}<s<\tilde{s}$. Also $\tilde{\eta}-\eta\geq t^{-1}$.
We fisrt integrate over $s$ and $\tau$. 
\begin{align}
 Q_{\pi}=& \lambda^{2(M+M_{0}+1)}
\int d\alpha d\tilde{\alpha}d\beta d\tilde{\beta} \frac{e^{-i(\alpha+i\eta)s}
e^{-i(\tilde{\alpha}+i\tilde{\eta})(\tilde{s}-s)}e^{i(\beta-i\eta)\tau}e^{i(\tilde{\beta}-i\tilde{\eta})(\tilde{s}-\tau)}\Big|_{s=\theta}^{\tilde{s}}\Big|_{\tau=\theta}^{\tilde{s}}} {\left(\alpha-\tilde{\alpha}-i(\eta-\tilde{\eta})\right)\left(\beta-\tilde{\beta}-i(\eta-\tilde{\eta})\right)} \notag\\
\times& \prod_{j=0}^{\bar{n}}\int d\omega_{j}\left(\frac{1}{\omega_{j}-\tilde{\alpha}-i\tilde{\eta}}\right)^{a_{j}}
\left(\frac{1}{\omega_{j}-\alpha-i\eta}\right)^{b_{j}}
\left(\frac{1}{\omega_{j}-\tilde{\beta}+i\tilde{\eta}}\right)^{c_{j}}
\left(\frac{1}{\omega_{j}-\beta+i\eta}\right)^{d_{j}} \notag \\
\times& \prod_{j=1}^{n'}\int d\omega'_{j} \left(\frac{1}{\omega'_{j}-\gamma_{j}-i\eta_{j}}\right) \notag 
\end{align}
We bound $Q_{\pi}$ by taking the absolute value inside the rest of the integrals.
Integrations over propagators of multiplicity one are bounded by $\left|\log \eta \right|$.
Using inequality (\ref{ine:two}) on the integrations over $\omega_{j}$ with $j\neq 0$, and using inequality (\ref{eq:ABlog}) on the integration over the remaining variables we obtain:
\begin{align}
\left| Q_{\pi}\right|
\leq& \lambda^{2(M+M_{0}+1)}
\int d\alpha d\tilde{\alpha}d\beta d\tilde{\beta}\int d\omega_{0}
\left|\frac{1}{\alpha-\tilde{\alpha}-i(\eta-\tilde{\eta})}\right|
\left|\frac{1}{\beta-\tilde{\beta}-i(\eta-\tilde{\eta})}\right|\log^{n'}(\eta)  \notag\\
\times &\left|\frac{1}{\omega_{0}-\alpha-i\eta}\right|\left|\frac{1}{\omega_{0}-\beta-i\eta}\right|
\prod_{j=0}^{\bar{n}}\left(
\left(\frac{1}{\tilde{\eta}}\right)^{a_{j}+d_{j}-1}
\left(\frac{1}{\eta}\right)^{b_{j}+d_{j}}\right) \left(\frac{1}{\tilde{\eta}}\right)^{a_{0}+d_{0}}
\left(\frac{1}{\eta}\right)^{b_{0}+c_{0}-2}
\notag \\
\leq &\lambda^{2(M+M_{0}+1)} \left(\frac{1}{\eta}\right)^{2\left(M+M_{0}+2\right)-n'-\bar{n}-2}\frac{\log^{n'+4}(\eta)
}{\kappa^{\sum_{j=0}^{\bar{n}}\left(a_{j}+c_{j}-1\right)}}
\end{align}
By Eqs. (\ref{consnn}), (\ref{eta1}) and (\ref{eta2}) we have the following bound:
\begin{align}
\left| Q_{\pi}\right|\leq & \left(\lambda^2t\right)^{M+M_{0}+1}\frac{\log^{M+M_{0}+6}(t)
}{\kappa^{\sum_{j=0}^{\bar{n}}\left(a_{j}+c_{j}\right)-\bar{n}-1}}
\end{align}
From Eqs. (\ref{consnn})  and (\ref{consacprop}) we have 
\begin{align}
 \sum_{j=0}^{\bar{n}}\left(a_{j}+c_{j}\right)-\bar{n}-1\geq M_{0}-M-1
\end{align}
Thus
\begin{align}
\left| Q_{\pi}\right|\leq & 
\left(\lambda^2t\right)^{M+M_{0}+1}\frac{\log^{M+M_{0}+6}(t)}{\kappa^{M_{0}-M-1}}\label{kapaQbound}
\end{align}
By inserting Eq. (\ref{kapaQbound}) in Eq. (\ref{eq:AVncQ}) and inserting this in Eq. (\ref{vsquared}) we obtain
\begin{align}
\lim_{N\rightarrow \infty }\mathbb{E}\left[\langle \psi^{2}_{M,M_{0},\kappa}(t) |\psi^{2}_{M,M_{0},\kappa}(t)\rangle\right] \leq 
\left(C\lambda^2t\right)^{M+M_{0}+2}\frac{t
\log^{M+M_{0}+6}(t)}{\kappa^{M_{0}-M-1}}
\end{align}
\end{proof}
We set now back the $t$ dependency of $M$ and $\kappa$ and take $M_{0}$ as follows
\begin{align}
x&=\log t\label{uno} \\
M(t)&=\gamma\frac{x}{\log x} \label{dos}\\
\kappa(t)&= x^{\alpha} \label{tres}\\
M_{0}(t)&=4M(t) \label{quatro}
\end{align}
We have then 
\begin{align}
\mathbb{E}\left[\langle \psi^{2}_{M,M_{0},\kappa}(t) |\psi^{2}_{M,M_{0},\kappa}(t)\rangle\right] \leq &
\left(CT\right)^{5\gamma\frac{x}{\log x}} \frac{e^{x}
x^{\gamma 5\frac{x}{\log x}+6}}{x^{3\alpha\gamma\frac{x}{\log x}-1}} \notag \\
\leq &
x^{7}Exp\left[x\left(1+\gamma 5+5\gamma\frac{\log(CT)}{\log x}-3\alpha\gamma\right) \right]
\end{align}
with a suitable choice of $\alpha$ and $\gamma$ the coefficient of the exponential is negative and this quantity vanishes in the limit $x\rightarrow \infty$.
%%%%%%%%%%%%%%%%%%%%%%%%%%%%%%%%%%%%%%%%%%%%%%%%%%%%%%%%%%%%%%%%%
%%%%%%%%%%%%%%%%%%%%%%%%%%%%%%%%%%%%%%%%%%%%%%%%%%%%%%%%%%%%%%%%%
\newpage
We will now seek to prove the following theorem for the bound of the norm of $|\psi^{1}_{M,M_{0},\kappa}(t)\rangle$: 
\begin{thm}
\label{thm:psi1}
 \begin{align}
\lim_{N\rightarrow \infty}\mathbb{E}\left[\langle \psi^{1}_{M,M_{0},\kappa}(t)|\psi^{1}_{M,M_{0},\kappa}(t)\rangle \right]
&\leq \kappa^{2} M_{0} \sum_{n=0}^{M_{0}}\left(
 \left(C T\right)^{M+n+1}
\frac{\log^{2}t}{\left(M+n\right)!^{\frac{1}{2}}}+
\frac{CT^{M+n+1}\log^{M+n+1} t }{t}\right)
\end{align}
\end{thm}
Mainly the bounds derived here are analogues to the ones derived and used in the previous sections for nested and simple graphs. We will first bound  the contribution of a nested graph, similar to how it was done in section \ref{sec:nested} . Then we will bound a part of a simple graph ($\Delta Q_{\pi}$), similar to how it was done in section \ref{sec:simple}. And  finally we will bound what remains of the simple graph ($\tilde{Q}_{\pi}$). The only difference is that the expression for $Q_{\pi}$ is a bit more complicated.\\
Using the Cauchy-Schwartz inequality on 
$\langle \psi^{1}_{M,M_{0},\kappa}(t)|\psi^{1}_{M,M_{0},\kappa}(t)\rangle$ when replacing $|\psi^{1}_{M,M_{0},\kappa}(t)\rangle$ by the expression in  Eq. (\ref{psi1}) we obtain
\begin{align}
\mathbb{E}\left[\langle \psi^{1}_{M,M_{0},\kappa}(t)|\psi^{1}_{M,M_{0},\kappa}(t)\rangle \right]
&\leq\kappa M_{0} \sum_{j=1}^{\kappa}\sum_{n=0}^{M_{0}}
\mathbb{E}\left[\langle \psi_{M,n,\kappa,\theta_{j}}(\theta_{j+1}) |\psi_{M,n,\kappa,\theta_{j}}(\theta_{j+1})\rangle \right] \label{eq:Avsquared}
\end{align}
In the large $N$ limit we have
\begin{align}
\lim_{N\rightarrow \infty }
\mathbb{E}\left[\langle \psi_{M,n,\kappa,\theta_{j}}(\theta_{j+1}) |\psi_{M,n,\kappa,\theta_{j}}(\theta_{j+1})\rangle \right]
&= \sum_{\pi(M+n+1,M+n+1)\in \mathcal{G}_{0},\mathcal{G}_{1}} Q_{\pi(M+n+1,M+n+1)}(\theta_{j},\theta_{j+1},\lambda,\kappa)\label{eq:Avsquared2}
\end{align}
The crossing graphs do not contribute once again because their individual contribution is of the order of $N^{-2}$.
For briefness of notation we refer now to $Q_{\pi(M+n+1,M+n+1)}(\theta_{j},\theta_{j+1},\lambda,\kappa)$ with $Q_{\pi}$ and have the following expression for it:
\begin{align}
 Q_{\pi}=& \lambda^{2(M+n+1)}
\int_{\theta}^{\tilde{s}} d\tau ds\int d\alpha d\tilde{\alpha}d\beta d\tilde{\beta} e^{-i(\alpha+i\eta)s}
e^{-i(\tilde{\alpha}+i\tilde{\eta})(\tilde{s}-s)}e^{i(\beta-i\eta)\tau}e^{i(\tilde{\beta}-i\tilde{\eta})(\tilde{s}-\tau)} \notag\\
\times& \prod_{j=0}^{\bar{n}}\int d\omega_{j}\left(\frac{1}{\omega_{j}-\tilde{\alpha}-i\tilde{\eta}}\right)^{  a_{j}}
\left(\frac{1}{\omega_{j}-\alpha-i\eta}\right)^{b_{j}}
\left(\frac{1}{\omega_{j}-\beta+i\eta}\right)^{c_{j}} 
\left(\frac{1}{\omega_{j}-\tilde{\beta}+i\tilde{\eta}}\right)^{d_{j}} \notag \\
\times& \prod_{j=1}^{n'}\int d\omega'_{j} \left(\frac{1}{\omega'_{j}-\gamma_{j}-i\eta}\right) \label{eq:Qtsl}
\end{align}
Eq. (\ref{eq:Qtsl}) is derived analoguesly to how Eq. (\ref{eq:Qabcd}) is derived from Eq. (\ref{psiaux}). 
The following relations are satisfied for non crossing graphs:
\begin{align}
 \sum_{j=0}^{\bar{n}}\left(a_{j}+b_{j}+c_{j}+d_{j}\right)+n'=&2\left(M+n+2\right)\label{ID1} \\
\bar{n}+1+n'=&M+n+2 \label{ID2}
\end{align}
These relations express the fact that there are $2\left(M+n+2\right)$ propagators in emerging from the $(M+n+1)^{\text{th}}$ expansion and that there are $M+n+2$ independent energy variables whenever one has a non-crossing graph.
Once again the graphs $\pi(M+n+1,M+n+1)$ can be either nested or simple graphs and depending on this $Q_{\pi}$ will render different contributions.
We will show that nested graphs have an extra $t^{-1}$ factor.
\subsection{NESTED}
For nested graphs we will now prove the following theorem:
\begin{lem}
\label{lemmanestedgraphs}
If $\pi(M+n+1,M+n+1)$ is a nested graph then we have the following bound for $Q_{\pi}$ in Eq. (\ref{eq:Qtsl}):
\begin{align}
\left|Q_{\pi}\right|\leq\frac{\left(CT\right)^{M+n+1}\log^{M+n+1} t }{t}
\end{align}
\end{lem}
\begin{proof}{Lemma \ref{lemmanestedgraphs}}\\
Starting from Eq. (\ref{eq:Qtsl})
we can perform the $s$ and $\tau$ integrations.
\begin{align}
 Q_{\pi}=& \lambda^{2(M+n+1)}
\int d\alpha d\tilde{\alpha}d\beta d\tilde{\beta} \frac{e^{-i(\alpha+i\eta)s}
e^{-i(\tilde{\alpha}+i\tilde{\eta})(\tilde{s}-s)}e^{i(\beta-i\eta)\tau}e^{i(\tilde{\beta}-i\tilde{\eta})(\tilde{s}-\tau)}\Big|_{s=\theta}^{\tilde{s}}\Big|_{\tau=\theta}^{\tilde{s}}} {\left(\alpha-\tilde{\alpha}-i(\eta-\tilde{\eta})\right)\left(\beta-\tilde{\beta}-i(\eta-\tilde{\eta})\right)} \notag\\
\times& \prod_{j=0}^{\bar{n}}\int d\omega_{j}\left(\frac{1}{\omega_{j}-\tilde{\alpha}-i\tilde{\eta}}\right)^{a_{j}}
\left(\frac{1}{\omega_{j}-\alpha-i\eta}\right)^{b_{j}}
\left(\frac{1}{\omega_{j}-\tilde{\beta}+i\tilde{\eta}}\right)^{c_{j}}
\left(\frac{1}{\omega_{j}-\beta-i\eta}\right)^{d_{j}} \notag \\
\times& \prod_{j=1}^{n'}\int d\omega'_{j} \left(\frac{1}{\omega'_{j}-\gamma_{j}-i\eta_{j}}\right) \notag 
\end{align}
A nest can present itself in two ways. First we can have, for a specific $j$, one of the indices $a_{j}$, $b_{j}$, $c_{j}$ or $d_{j}$ different from $0$ with all of the others equal to $0$. Secondly we can have, for a specific $j$, both $c_{j},d_{j}\neq 0$ and $a_{j},b_{j}= 0$ or $c_{j},d_{j}= 0$ and $a_{j},b_{j}\neq 0$.
The first case can be analyzed similarly to how the original bound for nested graphs was done in section \ref{sec:nested}. We analyze the second case.
If the nest is such that $c_{l}$ and $d_{l}$ are different then zero and $a_{l}=b_{l}=0$ then the nested part of the $Q_{\pi}$ function can be bound as follows:
\begin{align}
 &\left|\int_{0}^{1} d\omega_{l}
\left(\frac{1}{\omega_{l}-\tilde{\beta}+i\tilde{\eta}}\right)^{c_{l}}
\left(\frac{1}{\omega_{l}-\beta+i\eta}\right)^{d_{l}}\right| 
= \left|\int_{0}^{1} d\omega_{l} \int_{0}^{\infty} ds_{1}ds_{2}
e^{i\left(\omega_{l}-\tilde{\beta}+i\tilde{\eta}\right)s_{1}}
\frac{s_{1}^{c_{l}-1}}{\left(c_{l}-1\right)!}
e^{i\left(\omega_{l}-\beta+i\eta\right)s_{2}}
\frac{s_{2}^{d_{l}-1}}{\left(d_{l}-1\right)!}\right| \notag \\
\leq & \left| \int_{0}^{\infty} ds_{1}ds_{2}
\frac{e^{i\left(1-\tilde{\beta}+i\tilde{\eta}\right)s_{1}}e^{i\left(1-\beta+i\eta\right)s_{2}}
s_{1}^{c_{l}-1}s_{2}^{d_{l}-1}}{\left(s_{1}+s_{2}\right)\left(d_{l}-1\right)!\left(c_{l}-1\right)!}\right| +
\left| \int_{0}^{\infty} ds_{1}ds_{2}
\frac{e^{i\left(-\tilde{\beta}+i\tilde{\eta}\right)s_{1}}e^{i\left(-\beta+i\eta\right)s_{2}}
s_{1}^{c_{l}-1}s_{2}^{d_{l}-1}}{\left(s_{1}+s_{2}\right)\left(d_{l}-1\right)!\left(c_{l}-1\right)!}\right| \label{nestmu}
\end{align}
By introducing $\frac{1}{s_{1}+s_{2}}=\int_{0}^{\infty}d\mu e^{-\mu\left(s_{1}+s_{2}\right)}$ and integrating over $s_{1}$ and $s_{2}$ we obtain
\begin{align}
(\ref{nestmu})\leq & \int d\mu \left( \left|\frac{1}{1-\tilde{\beta}+i\left(\tilde{\eta}+\mu\right)}\right|^{c_{l}}
\left|\frac{1}{1-\beta+i\left(\eta+\mu\right)}\right|^{d_{l}}
+\left|\frac{1}{-\tilde{\beta}+i\left(\tilde{\eta}+\mu\right)}\right|^{c_{l}}
\left|\frac{1}{-\beta+i\left(\eta+\mu\right)}\right|^{d_{l}}\right)
\end{align}
For the other integrations over the $\omega_{j}$'s, with $j\geq 2$, in Eq. (\ref{eq:Qtsl}) we use the usual bounds
\begin{align}
\int d\omega_{j} \left|\frac{1}{\omega_{j}-\tilde{\alpha}-i\tilde{\eta}}\right|^{a_{j}}
\left|\frac{1}{\omega_{j}-\alpha-i\eta}\right|^{b_{j}}
\left|\frac{1}{\omega_{j}-\tilde{\beta}+i\tilde{\eta}}\right|^{c_{j}}
\left|\frac{1}{\omega_{j}-\beta-i\eta}\right|^{d_{j}}\leq& \left(\frac{1}{\eta}\right)^{a_{j}+b_{j}+c_{j}+d_{j}-1} \label{usualb}\\
\int d\omega'_{j} \left|\frac{1}{\omega'_{j}-\gamma_{j}-i\eta_{j}}\right|\leq& \log\eta .
\end{align}
Inserting these bounds in Eq. (\ref{eq:Qtsl}) we obtain
\begin{align}
\left| Q_{\pi}\right|\leq& \lambda^{2(M+n+1)}
\int d\mu\int d\alpha d\tilde{\alpha}d\beta d\tilde{\beta}\int d\omega_{0} d \omega_{1} \left|\frac{1}{\alpha-\tilde{\alpha}-i(\eta-\tilde{\eta})}\right| 
\left|\frac{1}{\beta-\tilde{\beta}-i(\eta-\tilde{\eta})}\right|
\notag\\
\times&  \left|\frac{1}{\omega_{0}-\tilde{\alpha}-i\tilde{\eta}}\right|
\left|\frac{1}{\omega_{1}-\alpha-i\eta}\right|
\left|\frac{1}{\omega_{0}-\tilde{\beta}+i\tilde{\eta}}\right|
\left|\frac{1}{\omega_{1}-\beta-i\eta}\right|
 \left|\frac{1}{1-\tilde{\beta}+i\left(\tilde{\eta}+\mu\right)}\right|^{c_{j}}
\left|\frac{1}{1-\beta-i\left(\eta+\mu\right)}\right|^{d_{j}} \notag \\
\times & \eta^{2}\prod_{j=0,j\neq l}^{\bar{n}}\left(\frac{1}{\eta}\right)^{a_{j}+b_{j}+c_{j}+d_{j}-1}\log^{n'}(\eta) . \label{eq:Qtsl2}
\end{align}
Because of Eq. (\ref{eq:thetaj}) and (\ref{eq:thetajdif})  we have $\eta<\tilde{\eta}<\tilde{\eta}-\eta$ and so 
if we replace in Eq. (\ref{eq:Qtsl2}) $\eta-\tilde{\eta}$ by $\eta$ and $\tilde{\eta}$ by $\eta$ the inequality will still hold. 
To bound the integration we successively use inequality (\ref{eq:ABlog}) on the integrations over $\omega_{0}$, $\omega_{1}$, $\alpha$ and $\tilde{\alpha}$ to obtain 
\begin{align}
&\int_{0}^{\infty} d\mu
\int d\alpha d\tilde{\alpha}d\beta d\tilde{\beta}\int d\omega_{0} d \omega_{1} \left|\frac{1}{\alpha-\tilde{\alpha}-i\eta}\right|
\left|\frac{1}{\beta-\tilde{\beta}-i\eta}\right|
 \left|\frac{1}{\omega_{0}-\tilde{\alpha}-i\eta}\right|
\left|\frac{1}{\omega_{1}-\alpha-i\eta}\right| \notag \\
\times &\left|\frac{1}{\omega_{0}-\tilde{\beta}+i\eta}\right|
\left|\frac{1}{\omega_{1}-\beta-i\eta}\right| 
 \left|\frac{1}{1-\tilde{\beta}+i\left(\eta+\mu\right)}\right|^{c_{j}}
\left|\frac{1}{1-\beta-i\left(\eta+\mu\right)}\right|^{d_{j}} \notag  \\
\leq &\left|\log \eta\right|^{4}
\int_{0}^{\infty} d\mu
\int d\beta d\tilde{\beta} 
\left|\frac{1}{\beta-\tilde{\beta}-i\eta}\right|^{2}
 \left|\frac{1}{1-\tilde{\beta}+i\left(\eta+\mu\right)}\right|^{c_{j}}
\left|\frac{1}{1-\beta-i\left(\eta+\mu\right)}\right|^{d_{j}} \notag  
\end{align}
When $\mu$ greater then a constant $C$ this is bounded by $\frac{1} {C^{d_{j}+c_{j}-1}\eta}$ and so we can consider only the region where $\mu$ is bounded.
For $c_{j}$ or $d_{j}$ greater then $2$ we have:
\begin{align}
\leq &\left|\log \eta\right|^{4}\frac{1}{\eta}
\int_{0}^{C} d\mu
\int d\beta d\tilde{\beta} 
\left|\frac{1}{\beta-\tilde{\beta}-i\eta}\right|
 \left|\frac{1}{1-\tilde{\beta}+i\left(\tilde{\eta}+\mu\right)}\right|^{c_{j}}
\left|\frac{1}{1-\beta-i\left(\eta+\mu\right)}\right|^{d_{j}} \notag  \\
\leq &\left|\log \eta\right|^{5}\frac{1}{\eta}
\int_{0}^{C} d\mu
 d\beta 
 \left|\frac{1}{1-\beta+i\eta}\right|
\left|\frac{1}{1-\beta-i\eta}\right|  \left|\frac{1}{\tilde{\eta}+\mu}\right|^{c_{j}-1}
\left|\frac{1}{\eta+\mu}\right|^{d_{j}-1} \notag  \\
\leq &
\left|\log \eta\right|^{6}\left|\frac{1}{\eta}\right|^{2}
\int d\mu
 \left|\frac{1}{\left(\tilde{\eta}+\mu\right)}\right|^{c_{j}-1}
\left|\frac{1}{\left(\eta+\mu\right)}\right|^{d_{j}-1} \notag  \\
\leq &\left|\log \eta\right|^{6}
\left|\frac{1}{\eta}\right|^{c_{j}+d_{j}-1}
\end{align}
For $c_{j}=d_{j}=1$ we have  
\begin{align}
 \leq &\left|\log \eta\right|^{4}
\int_{0}^{C} d\mu
\int d\beta d\tilde{\beta} 
\left|\frac{1}{\beta-\tilde{\beta}-i\eta}\right|^{2}
 \left|\frac{1}{1-\beta-i\left(\eta+\mu\right)}\right|^{2} \notag \\
 \leq &\frac{\left|\log \eta\right|^{4}}{\eta}
\int_{0}^{C} d\mu
\int d\beta
 \left|\frac{1}{1-\beta+i\left(\eta+\mu\right)}\right|^{2} \notag  \\
 \leq &\frac{\left|\log \eta\right|^{5}}{\eta}
\end{align}
Inserting this in Eq. (\ref{eq:Qtsl2}) we obtain
\begin{align}
\left| Q_{\pi}(\theta,\tilde{s},\lambda)\right|\leq& \lambda^{2(M+n+1)}
\eta^{2}\prod_{j=0}^{\bar{n}}\left(\frac{1}{\eta}\right)^{a_{j}+b_{j}+c_{j}+d_{j}-1}\log^{n'}(\eta)\notag \\
\leq& \lambda^{2(M+n+1)}
t^{M+n}\log^{n'}t \\
\leq& T^{M+n+1}
\frac{\log^{M+n+1}t}{t} \label{eq:Qtsl3}
\end{align}
where we have used Eqs. (\ref{ID1}) and (\ref{ID2}) to find the exponent of $t$.
In case that the nest is such that for a specific $l$ only one out of $a_{l}$, $b_{l}$, $c_{l}$ and $d_{l}$ is different then $0$ we can perform directly the integration over $\omega_{l}$ and follow the procedure of section \ref{sec:nested}.
\end{proof}
We now turn to simple graphs.
\subsection{SIMPLE}
If $\pi(M+n+1,M+n+1)$ is a simple graph then in Eq. (\ref{eq:Qtsl}) for each $j$ we have an $a_{j}$ or $b_{j}$ different from $0$ and $c_{j}$ or $d_{j}$ different from $0$.
Similar to how was done in section \ref{sec:simple} we can prove that the contribution of a simple graph can be decomposed in two parts contributing in two different ways.
We define now $\tilde{Q}_{\pi}$ as $Q_{\pi}$ from Eq. (\ref{eq:Qtsl}) but with the propagators $\Theta(\gamma_{j},\eta_{j})=\int d\omega'_{j} \frac{1}{\omega'_{j}-\gamma_{j}-i\eta_{j}}$ replaced by $\Theta(\omega_{\bar{n}})$.
 The difference between $\tilde{Q}_{\pi}$ and $Q_{\pi}$, $\Delta Q_{\pi}$, is then  the following:
\begin{align}
\Delta Q_{\pi}(\theta,\tilde{s},\lambda)=& \lambda^{2(M+n+1)}
\int_{\theta}^{\tilde{s}} d\tau ds\int d\alpha d\tilde{\alpha}d\beta d\tilde{\beta} e^{-i(\alpha+i\eta)s}
e^{-i(\tilde{\alpha}+i\tilde{\eta})(\tilde{s}-s)}e^{i(\beta-i\eta)\tau}e^{i(\tilde{\beta}-i\tilde{\eta})(\tilde{s}-\tau)} \notag\\
\times& \prod_{j=0}^{\bar{n}}\int d\omega_{j}\left(\frac{1}{\omega_{j}-\tilde{\alpha}-i\tilde{\eta}}\right)^{a_{j}}
\left(\frac{1}{\omega_{j}-\alpha-i\eta}\right)^{b_{j}}
\left(\frac{1}{\omega_{j}-\tilde{\beta}+i\tilde{\eta}}\right)^{c_{j}}
\left(\frac{1}{\omega_{j}-\beta-i\eta}\right)^{d_{j}} \notag \\
\times&\left( 
\Theta^{n_{1}}(\tilde{\alpha},\tilde{\eta})\Theta^{n_{2}}(\alpha,\eta)
\bar{\Theta}^{n_{3}}(\beta,\eta)
\bar{\Theta}^{n_{4}}(\tilde{\beta},\tilde{\eta})
-\Theta^{n_{1}+n_{2}}(\omega_{\bar{n}})\bar{\Theta}^{n_{3}+n_{4}}(\omega_{\bar{n}})\right) \label{eq:deltaQ}
\end{align}
with $n_{1}+n_{2}+n_{3}+n_{4}=n'$.
We will now prove the following: 
\begin{lem}
\label{lemmasimplegraphs}
\begin{align} 
\left|\Delta Q_{\pi}(\theta,\tilde{s},\lambda)\right|
\leq & \epsilon^{-1}\left(CT\right)^{M+n+1} \frac{\log^{5+n'}t}{t}
\end{align}
\end{lem}
\begin{proof}{Lemma \ref{lemmasimplegraphs}}\\
As done previously we can rewrite the difference of $\Theta$ functions as follows:
\begin{align}
 &\Theta^{n_{1}}(\tilde{\alpha},\tilde{\eta})\Theta^{n_{2}}(\alpha,\eta)
\bar{\Theta}^{n_{3}}(\beta,\eta)
\bar{\Theta}^{n_{4}}(\tilde{\beta},\tilde{\eta})
-\Theta^{n_{1}+n_{2}}(\omega_{\bar{n}})\bar{\Theta}^{n_{3}+n_{4}}(\omega_{\bar{n}})\notag \\
=& \Theta^{n_{1}}(\tilde{\alpha},\tilde{\eta})\Theta^{n_{2}}(\alpha,\eta)
\bar{\Theta}^{n_{3}}(\beta,\eta)
\left(\bar{\Theta}^{n_{4}}(\tilde{\beta},\tilde{\eta}) -\bar{\Theta}^{n_{4}}(\omega_{\bar{n}})\right) \notag \\
+&\Theta^{n_{1}}(\tilde{\alpha},\tilde{\eta})\Theta^{n_{2}}(\alpha,\eta)
\bar{\Theta}^{n_{4}}(\omega_{\bar{n}})
\left(\bar{\Theta}^{n_{3}}(\beta,\eta)
 -\bar{\Theta}^{n_{3}}(\omega_{\bar{n}})\right) \notag \\
+&\Theta^{n_{1}}(\tilde{\alpha},\tilde{\eta})\bar{\Theta}^{n_{3}}(\omega_{\bar{n}}
\bar{\Theta}^{n_{4}}(\omega_{\bar{n}})
\left(\Theta^{n_{2}}(\alpha,\eta)
 -\Theta^{n_{2}}(\omega_{\bar{n}})\right) \notag \\
+&\Theta^{n_{2}}(\omega_{\bar{n}})\bar{\Theta}^{n_{3}}(\omega_{\bar{n}})
\bar{\Theta}^{n_{4}}(\omega_{\bar{n}})
\left(
\Theta^{n_{1}}(\tilde{\alpha},\tilde{\eta})
 -\Theta^{n_{1}}(\omega_{\bar{n}})\right) \label{eq:thetas}
\end{align}
We define $A$, $B$, $C$ and $D$ to be the first second third and fourth part of the sum in Eq. (\ref{eq:thetas}). We also denote by $\Delta Q^{A}_{\pi}$, $\Delta Q^{B}_{\pi}$, $\Delta Q^{C}_{\pi}$ and $\Delta Q^{D}_{\pi}$ the contribution to 
$\Delta Q_{\pi}(\theta,\tilde{s},\lambda)$ from $A$, $B$, $C$ and $D$ in Eq. (\ref{eq:deltaQ}).
Each difference can the again be expanded as follows:
\begin{align}
\Theta^{n_{1}}(\tilde{\alpha},\tilde{\eta})
 -\Theta^{n_{1}}(\omega_{\bar{n}}) =\left(\Theta(\tilde{\alpha},\tilde{\eta})
 -\Theta(\omega_{\bar{n}})\right)\sum_{p=0}^{n_{1}-1}\Theta^{p}(\omega_{\bar{n}}) \Theta^{n_{1}-1-p}(\tilde{\alpha},\tilde{\eta}) 
\end{align}
Since every $\Theta(\alpha,\eta)$ is bounded by $\left|\log \eta\right|$ we have for the first factor for example
\begin{align}
\left|A\right|=&\left|\Theta^{n_{1}}(\tilde{\alpha},\tilde{\eta})\Theta^{n_{2}}(\alpha,\eta)
\bar{\Theta}^{n_{3}}(\beta,\eta)
\left(\bar{\Theta}^{n_{4}}(\tilde{\beta},\tilde{\eta}) -\bar{\Theta}^{n_{4}}(\omega_{\bar{n}})\right)\right| \notag \\
\leq &
\left|\log \eta\right|^{n_{1}+n_{2}+n_{3}}
\left|\Theta(\tilde{\alpha},\tilde{\eta}) -\Theta\right|
\sum_{p=0}^{n_{4}-1}C^{p} \left|\log\eta\right|^{n_{4}-1-p} \notag \\
\leq &
C^{n_{4}}\left|\log\eta\right|^{n_{1}+n_{2}+n_{3}+n_{4}}
\left|\Theta(\tilde{\alpha},\tilde{\eta}) -\Theta(\omega_{\bar{n}})\right| \label{eq:telescopic}
\end{align}
If we now integrate in Eq. (\ref{eq:deltaQ}) over $s$ and $\tau$ and  use Eq. (\ref{eq:telescopic}) and (\ref{eq:thetabound})  we get 
\begin{align}
\left|\Delta Q^{A}_{\pi}(\theta,\tilde{s},\lambda)\right|
\leq& \lambda^{2(M+n+1)}
\int d\alpha d\tilde{\alpha}d\beta d\tilde{\beta} \prod_{j=0}^{\bar{n}}\int d\omega_{j}
\left|\frac{1}{\alpha-\tilde{\alpha}-i\left(\eta-\tilde{\eta}\right)}\right|
\left|\frac{1}{\beta-\tilde{\beta}-i\left(\eta-\tilde{\eta}\right)}\right|
 \notag\\
\times& \left|\frac{1}{\omega_{j}-\tilde{\alpha}-i\tilde{\eta}}\right|^{a_{j}}
\left|\frac{1}{\omega_{j}-\alpha-i\eta}\right|^{b_{j}}
\left|\frac{1}{\omega_{j}-\tilde{\beta}+i\tilde{\eta}}\right|^{c_{j}}
\left|\frac{1}{\omega_{j}-\beta-i\eta}\right|^{d_{j}} \notag \\
\times& \left|C\log \eta\right|^{n'}
\left|\omega_{\bar{n}}-\tilde{\alpha}-i\tilde{\eta}\right|
\left(\left|\frac{1}{\omega_{\bar{n}}}\right|+\left|\frac{1}{1-\omega_{\bar{n}}}\right|+
\left|\frac{1}{\tilde{\alpha}+i\tilde{\eta}}\right|+\left|\frac{1}{1-\tilde{\alpha}-i\tilde{\eta}}\right|\right)
 \label{eq:deltaQ2}
\end{align}
We bound the integrations over $\omega_{j}$, with $j\neq 0,\bar{n}$, by using Eq. (\ref{usualb}).
\begin{align}
\left|\Delta Q^{A}_{\pi}(\theta,\tilde{s},\lambda)\right|
\leq& \lambda^{2(M+n+1)}
\int d\alpha d\tilde{\alpha}d\beta d\tilde{\beta}\int d\omega_{0}d\omega_{\bar{n}}
\left|\frac{1}{\alpha-\tilde{\alpha}-i\left(\eta-\tilde{\eta}\right)}\right|
\left|\frac{1}{\beta-\tilde{\beta}-i\left(\eta-\tilde{\eta}\right)}\right|
 \notag\\
\times& 
\left|\frac{1}{\omega_{0}-\alpha-i\eta}\right|
\left|\frac{1}{\omega_{\bar{n}}-\tilde{\beta}+i\tilde{\eta}}\right|
\left|\frac{1}{\omega_{0}-\beta-i\eta}\right| \notag \\
\times&
\left(\left|\frac{1}{\omega_{\bar{n}}}\right|+\left|\frac{1}{1-\omega_{\bar{n}}}\right|+
\left|\frac{1}{\tilde{\alpha}+i\tilde{\eta}}\right|+\left|\frac{1}{1-\tilde{\alpha}-i\tilde{\eta}}\right|\right)
\eta^{2}\left|C\log \eta\right|^{n'} \prod_{j=0}^{\bar{n}}\left(\frac{1}{\eta}\right)^{a_{j}+b_{j}+c_{j}+d_{j}-1}
\end{align}
Applying now the bound of Eq. (\ref{eq:ABlog}) multiple times and remembering that the integration over $\omega_{\bar{n}}$ was cut of  and goes from  $1-\epsilon$ and $\epsilon$ because of  our choice of the initial condition we have
\begin{align}
\left|\Delta Q^{A}_{\pi}(\theta,\tilde{s},\lambda)\right|
\leq&\epsilon^{-1}\lambda^{2(M+n+1)} \log^{5}(\eta) C^{n'}\log^{n'}\eta\eta^{2}\prod_{j=0}^{\bar{n}}\left(\frac{1}{\eta}\right)^{a_{j}+b_{j}+c_{j}+d_{j}-1} \notag\\
\leq&\epsilon^{-1}\left(\lambda^{2}
t\right)^{M+n+1} \frac{C^{n'}\log^{5+n'} t}{t}
\end{align}
We have used the  identities of Eqs. (\ref{ID1}) and (\ref{ID2}) to compute the exponent of $t$. We can bound similarly the contributions from $B$, $C$ and $D$.
Thus giving
\begin{align}
\left|\Delta Q_{\pi}(\theta,\tilde{s},\lambda)\right|
\leq&\epsilon^{-1}\left(C\lambda^{2}
t\right)^{M+n+1} \frac{\log^{5+n'}t}{t}
\end{align}
\end{proof}
We now bound $\tilde{Q}_{\pi}$.
\begin{lem}
\label{lem:Qtilde}
\begin{align}
\left| \tilde{Q}_{\pi(M+n+1,M+n+1)}(\theta_{j},\tilde{s},\lambda)\right|\leq\frac{\left(CT\right)^{M+n+1}}{\left(M+n+1\right)!^{a}} 
\end{align}
with $0\leq a<1$.
\end{lem}
\begin{proof}{Lemma \ref{lem:Qtilde}}
We have 
\begin{align}
\tilde{Q}_{\pi}(\theta,\tilde{s},\lambda)=& \lambda^{2(M+n+1)}
\prod_{j=0}^{\bar{n}}\int d\omega_{j}\int_{\theta}^{\tilde{s}} d\tau ds\int d\alpha d\tilde{\alpha}d\beta d\tilde{\beta} e^{-i(\alpha+i\eta)s}
e^{-i(\tilde{\alpha}+i\tilde{\eta})(\tilde{s}-s)}e^{i(\beta-i\eta)\tau}e^{i(\tilde{\beta}-i\tilde{\eta})(\tilde{s}-\tau)} \notag\\
\times& \prod_{j=1}^{\bar{n}}\left(\frac{1}{\omega_{j}-\tilde{\alpha}-i\tilde{\eta}}\right)^{a_{j}}
\left(\frac{1}{\omega_{j}-\alpha-i\eta}\right)^{b_{j}}
\left(\frac{1}{\omega_{j}-\beta+i\eta}\right)^{c_{j}}
\left(\frac{1}{\omega_{j}-\tilde{\beta}+i\tilde{\eta}}\right)^{d_{j}} \notag \\
\times& \Theta^{n_{1}+n_{2}}(\omega_{\bar{n}})\bar{\Theta}^{n_{3}+n_{4}}(\omega_{\bar{n}}) \label{begin}
\end{align}
The proof is similar to that of theorem  \ref{thm:Qtildenm}. Starting from Eq. (\ref{begin}) we get 
\begin{align}
\tilde{Q}_{\pi}(\theta,\tilde{s},\lambda)\leq & \lambda^{2(M+n+1)}
\prod_{j=0}^{\bar{n}}\int d\omega_{j}
\Big|\int_{\theta}^{\tilde{s}} d\tau ds\int d\alpha d\tilde{\alpha}d\beta d\tilde{\beta} e^{-i(\alpha+i\eta)s}
e^{-i(\tilde{\alpha}+i\tilde{\eta})(\tilde{s}-s)}e^{i(\beta-i\eta)\tau}e^{i(\tilde{\beta}-i\tilde{\eta})(\tilde{s}-\tau)} \notag\\
\times& \prod_{j=0}^{\bar{n}}\left(\frac{1}{\omega_{j}-\tilde{\alpha}-i\tilde{\eta}}\right)^{a_{j}}
\left(\frac{1}{\omega_{j}-\alpha-i\eta}\right)^{b_{j}}
\left(\frac{1}{\omega_{j}-\tilde{\beta}+i\tilde{\eta}}\right)^{c_{j}}
\left(\frac{1}{\omega_{j}-\beta-i\eta}\right)^{d_{j}} \notag \\
\times& \Theta^{n_{1}+n_{2}}(\omega_{\bar{n}})\bar{\Theta}^{n_{3}+n_{4}}(\omega_{\bar{n}})\Big|^{1-a} \notag \\
\times&\Big|\int_{\theta}^{\tilde{s}} d\tau ds\int d\alpha d\tilde{\alpha}d\beta d\tilde{\beta} e^{-i(\alpha+i\eta)s}
e^{-i(\tilde{\alpha}+i\tilde{\eta})(\tilde{s}-s)}e^{i(\beta-i\eta)\tau}e^{i(\tilde{\beta}-i\tilde{\eta})(\tilde{s}-\tau)} \notag\\
\times& \prod_{j=0}^{\bar{n}}\left(\frac{1}{\omega_{j}-\tilde{\alpha}-i\tilde{\eta}}\right)^{a_{j}}
\left(\frac{1}{\omega_{j}-\alpha-i\eta}\right)^{b_{j}}
\left(\frac{1}{\omega_{j}-\tilde{\beta}+i\tilde{\eta}}\right)^{c_{j}}
\left(\frac{1}{\omega_{j}-\beta-i\eta}\right)^{d_{j}} \notag \\
\times& \Theta^{n_{1}+n_{2}}(\omega_{\bar{n}})\bar{\Theta}^{n_{3}+n_{4}}(\omega_{\bar{n}})\Big|^{a} \label{final} 
\end{align} 
For the second part , that is the one that is to the power of $a$, we use the $t$-representation
\begin{align}
&\Big|\int_{\theta}^{\tilde{s}} d\tau ds\int d\alpha d\tilde{\alpha}d\beta d\tilde{\beta} e^{-i(\alpha+i\eta)s}
e^{-i(\tilde{\alpha}+i\tilde{\eta})(\tilde{s}-s)}e^{i(\beta-i\eta)\tau}e^{i(\tilde{\beta}-i\tilde{\eta})(\tilde{s}-\tau)} \notag\\
\times&
 \prod_{j=0}^{\bar{n}}\left(\frac{1}{\omega_{j}-\tilde{\alpha}-i\tilde{\eta}}\right)^{a_{j}}
\left(\frac{1}{\omega_{j}-\alpha-i\eta}\right)^{b_{j}}
\left(\frac{1}{\omega_{j}-\tilde{\beta}+i\tilde{\eta}}\right)^{c_{j}}
\left(\frac{1}{\omega_{j}-\beta-i\eta}\right)^{d_{j}} \Theta^{n_{1}+n_{2}}(\omega_{\bar{n}})\bar{\Theta}^{n_{3}+n_{4}}(\omega_{\bar{n}})\Big|^{a}  \notag \\
\leq &\Big| \int d\tau ds \frac{\left(\tilde{s}-s\right)^{\left(\sum_{j=0}^{\bar{n}}a_{j}\right)-1}}{\left(\left(\sum_{j=0}^{\bar{n}}a_{j}\right)-1\right)!}
\frac{s^{\left(\sum_{j=0}^{\bar{n}}b_{j}\right)-1}}{\left(\left(\sum_{j=0}^{\bar{n}}b_{j}\right)-1\right)!}
\frac{\left(\tilde{s}-\tau\right)^{\left(\sum_{j=0}^{\bar{n}}c_{j}\right)-1}}{\left(\left(\sum_{j=0}^{\bar{n}}c_{j}\right)-1\right)!}
\frac{\tau^{\left(\sum_{j=0}^{\bar{n}}d_{j}\right)-1}}{\left(\left(\sum_{j=0}^{\bar{n}}d_{j}\right)-1\right)!}
\Big|^{a} \notag \\
\leq &\frac{\tilde{s}^{a\left(\left(\sum_{j=0}^{\bar{n}}a_{j}+b_{j}+c_{j}+d_{j}\right)-2\right)}}{\left(2M+2n+4\right)!^{a}}
\label{est1}
\end{align}
To bound the first part, that which is to the power of $1-a$, we first integrate over $s$ and $\tau$ and take the absolute value.
\begin{align}
&\lambda^{2(M+n+1)}
\prod_{j=0}^{\bar{n}}\int d\omega_{j}
\Big|\int_{\theta}^{\tilde{s}} d\tau ds\int d\alpha d\tilde{\alpha}d\beta d\tilde{\beta} e^{-i(\alpha+i\eta)s}
e^{-i(\tilde{\alpha}+i\tilde{\eta})(\tilde{s}-s)}e^{i(\beta-i\eta)\tau}e^{i(\tilde{\beta}-i\tilde{\eta})(\tilde{s}-\tau)} \notag\\
\times& \prod_{j=0}^{\bar{n}}\left(\frac{1}{\omega_{j}-\tilde{\alpha}-i\tilde{\eta}}\right)^{\chi_{j}^{a}a_{j}}
\left(\frac{1}{\omega_{j}-\alpha-i\eta}\right)^{\chi_{j}^{b}b_{j}}
\left(\frac{1}{\omega_{j}-\tilde{\beta}+i\tilde{\eta}}\right)^{\chi_{j}^{c}c_{j}}
\left(\frac{1}{\omega_{j}-\beta-i\eta}\right)^{\chi_{j}^{d}d_{j}} \notag \\
\times& \Theta^{n_{1}+n_{2}}(\omega_{\bar{n}})\bar{\Theta}^{n_{3}+n_{4}}(\omega_{\bar{n}})\Big|^{1-a} \notag \\
\leq & \lambda^{2(M+n+1)}
\prod_{j=0}^{\bar{n}}\int d\omega_{j}
\int d\alpha d\tilde{\alpha}d\beta d\tilde{\beta} 
\left|\frac{1}{\alpha-\tilde{\alpha}+i\left(\eta-\tilde{\eta}\right)}\right|^{1-a}
\left|\frac{1}{\beta-\tilde{\beta}+i\left(\eta-\tilde{\eta}\right)}\right|^{1-a} \notag \\
\times& \prod_{j=1}^{\bar{n}}\left|\frac{1}{\omega_{j}-\tilde{\alpha}-i\tilde{\eta}}\right|^{a_{j}(1-a)}
\left|\frac{1}{\omega_{j}-\alpha-i\eta}\right|^{b_{j}(1-a)}
\left|\frac{1}{\omega_{j}-\beta+i\eta}\right|^{c_{j}(1-a)}
\left|\frac{1}{\omega_{j}-\tilde{\beta}+i\tilde{\eta}}\right|^{d_{j}(1-a)} \label{first(1-a)}
\end{align}
where we have omitted the $\Theta$ function since they are bounded by constants.
Similar to Eq. (\ref{ine:two}) we have
\begin{align}
& \int d\omega_{j} \left|\frac{1}{\omega_{j}-\tilde{\alpha}-i\tilde{\eta}}\right|^{a_{j}(1-a)}
\left|\frac{1}{\omega_{j}-\alpha-i\eta}\right|^{b_{j}(1-a)}
\left|\frac{1}{\omega_{j}-\beta+i\eta}\right|^{c_{j}(1-a)}
\left|\frac{1}{\omega_{j}-\tilde{\beta}+i\tilde{\eta}}\right|^{d_{j}(1-a)} \notag \\
\leq& \left(\frac{1}{\eta}\right)^{\left(a_{j}+b_{j}+c_{j}+d_{j}-2\right)(1-a)}
\int d\omega_{j} \left|\frac{1}{\omega_{j}-\alpha-i\eta}\right|^{(1-a)}
\left|\frac{1}{\omega_{j}-\beta-i\eta}\right|^{(1-a)} \notag \\
\leq& \left(\frac{1}{\eta}\right)^{\left(a_{j}+b_{j}+c_{j}+d_{j}-2\right)(1-a)}
\int d\omega_{j} \left|\frac{1}{\omega_{j}-\alpha-i\eta}\right|^{2(1-a)}
 \notag \\
\leq& \left(\frac{1}{\eta}\right)^{\left(a_{j}+b_{j}+c_{j}+d_{j}\right)(1-a)-1}
 \notag 
\end{align}
where we have used Eq. (\ref{eq:oneminusa}) to bound the last integration.
Using this in Eq. (\ref{first(1-a)}) for the integrations over $\omega_{j}$ with $j\neq 0$ and $j\neq\bar{n}$ we obtain:
\begin{align}
(\ref{first(1-a)}) \leq & \lambda^{2(M+n+1)}
\int d\omega_{0}
\int d\omega_{1}
\int d\alpha d\tilde{\alpha}d\beta d\tilde{\beta} 
\left|\frac{1}{\alpha-\tilde{\alpha}+i\left(\eta-\tilde{\eta}\right)}\right|^{1-a}
\left|\frac{1}{\beta-\tilde{\beta}+i\left(\eta-\tilde{\eta}\right)}\right|^{1-a} \notag \\
\times& \left|\frac{1}{\omega_{0}-\tilde{\alpha}-i\tilde{\eta}}\right|^{a_{0}(1-a)}
\left|\frac{1}{\omega_{0}-\alpha-i\eta}\right|^{b_{0}(1-a)}
\left|\frac{1}{\omega_{0}-\tilde{\beta}+i\tilde{\eta}}\right|^{c_{0}(1-a)}
\left|\frac{1}{\omega_{0}-\beta-i\eta}\right|^{d_{0}(1-a)} \notag \\
\times& \left|\frac{1}{\omega_{\bar{n}}-\tilde{\alpha}-i\tilde{\eta}}\right|^{a_{\bar{n}}(1-a)}
\left|\frac{1}{\omega_{\bar{n}}-\alpha-i\eta}\right|^{b_{\bar{n}}(1-a)}
\left|\frac{1}{\omega_{\bar{n}}-\tilde{\beta}+i\tilde{\eta}}\right|^{c_{\bar{n}}(1-a)}
\left|\frac{1}{\omega_{\bar{n}}-\beta-i\eta}\right|^{d_{\bar{n}}(1-a)} \notag \\
\times & \prod_{j=1}^{\bar{n}-1}
\left(\frac{1}{\eta}\right)^{\left(a_{j}+b_{j}+c_{j}+d_{j}\right)(1-a)-1}\notag \\
\leq & \lambda^{2(M+n+1)}
\int d\omega_{0}
\int d\omega_{\bar{n}}
\int d\alpha d\tilde{\alpha}d\beta d\tilde{\beta} 
\left|\frac{1}{\alpha-\tilde{\alpha}+i\left(\eta-\tilde{\eta}\right)}\right|^{1-a}
\left|\frac{1}{\beta-\tilde{\beta}+i\left(\eta-\tilde{\eta}\right)}\right|^{1-a} \notag \\
\times& \left|\frac{1}{\omega_{0}-\tilde{\alpha}-i\tilde{\eta}}\right|^{1-a}
\left|\frac{1}{\omega_{0}-\tilde{\beta}+i\tilde{\eta}}\right|^{1-a} 
\left|\frac{1}{\omega_{\bar{n}}-\alpha-i\eta}\right|^{1-a}
\left|\frac{1}{\omega_{\bar{n}}-\beta-i\eta}\right|^{1-a} \notag \\
\times & 
\left(\frac{1}{\eta}\right)^{\left(a_{0}+b_{0}+c_{0}+d_{0}-2\right)(1-a)
+\left(a_{\bar{n}}+b_{\bar{n}}+c_{\bar{n}}+d_{\bar{n}}-2\right)(1-a)}
\prod_{j=1}^{\bar{n}-1}
\left(\frac{1}{\eta}\right)^{\left(a_{j}+b_{j}+c_{j}+d_{j}\right)(1-a)-1} 
\end{align}
By using inequality of Eq. (\ref{eq:ABdelta}) on the integrations over $\alpha$, $\beta$, $\tilde{\alpha}$ and $\tilde{\beta}$ and the applying inequality of Eq. (\ref{eq:oneminusa}) on the integration over $\omega_{0}$ we obtain
\begin{align}
\leq & \lambda^{2(M+n+1)}
\left(\frac{1}{\eta}\right)^{1-2a}
\left(\frac{1}{\eta}\right)^{\left(a_{0}+b_{0}+c_{0}+d_{0}-2\right)(1-a)
+\left(a_{\bar{n}}+b_{\bar{n}}+c_{\bar{n}}+d_{\bar{n}}-2\right)(1-a)}
\prod_{j=1}^{\bar{n}-1}
\left(\frac{1}{\eta}\right)^{\left(a_{j}+b_{j}+c_{j}+d_{j}\right)(1-a)-1}\notag \\
\leq & \lambda^{2(M+n+1)}
\left(\frac{1}{\eta}\right)^{
-a\left(\sum_{j=0}^{\bar{n}}\left(a_{j}+b_{j}+c_{j}+d_{j}\right)-2\right)
+\sum_{j=0}^{\bar{n}}\left(a_{j}+b_{j}+c_{j}+d_{j}\right)-\bar{n}-2
}\label{est2}
\end{align}
Combining the estimates of Eqs. (\ref{est1}) and (\ref{est2}) in Eq. (\ref{final}) we obtain
\begin{align}
\left|\tilde{Q}_{\pi}(\theta,\tilde{s},\lambda)\right|\leq & \lambda^{2(M+n+1)}
\frac{\left(C\lambda^{2}t\right)^{M+n+1}}{\left(2M+2n+4\right)!^{a}}
\end{align}
\end{proof}
We can now prove theorem \ref{thm:psi1} .
\begin{proof}{Theorem \ref{thm:psi1}}\\
By Eq. (\ref{eq:Avsquared2}) and lemmas \ref{lemmanestedgraphs} and \ref{lemmasimplegraphs} we have 
\begin{align}
& \lim_{N\rightarrow \infty }
\mathbb{E}\left[\langle \psi_{M,n,\kappa,\theta_{j}}(\theta_{j+1}) |\psi_{M,n,\kappa,\theta_{j}}(\theta_{j+1})\rangle \right]
= \sum_{\pi(M+n+1,M+n+1)\in \mathcal{G}_{1}} Q_{\pi(M+n+1,M+n+1)}(\theta_{j},\theta_{j+1},\lambda) \notag \\
+&\sum_{\pi(M+n+1,M+n+1)\in \mathcal{G}_{0}} \tilde{Q}_{\pi(M+n+1,M+n+1)}(\theta_{j},\theta_{j+1},\lambda)+ \Delta  Q_{\pi(M+n+1,M+n+1)}(\theta_{j},\theta_{j+1},\lambda)\notag \\
&\leq  \frac{\left(CT\right)^{M+n+1}\log^{M+n+1}t}{t}
+\frac{\left(CT\right)^{M+n+1}\log^{M+n+1+5}t}{t}+
\frac{\left(CT\right)^{M+n+1}}{\left(M+n+1\right)!^{\frac{1}{2}}}
\end{align}
Inserting this in Eq. (\ref{eq:Avsquared}) we obtain
\begin{align}
\lim_{N\rightarrow \infty}\mathbb{E}\left[\langle \psi^{1}_{M,M_{0},\kappa}(t)|\psi^{1}_{M,M_{0},\kappa}(t)\rangle \right]
&\leq 2\kappa^{2} M_{0} \sum_{n=0}^{M_{0}}\frac{\left(CT\right)^{M+n+1}\log^{M+n+6}t}{t}\notag\\
&+\kappa^{2} M_{0} \sum_{n=0}^{M_{0}}\frac{\left(CT\right)^{M+n+1}}{\left(M+n+1\right)!^{\frac{1}{2}}}\notag
\end{align}
Through Eqs. (\ref{uno})-(\ref{quatro}) we have 
\begin{align}
 \lim_{t\rightarrow \infty}\lim_{N\rightarrow \infty}\mathbb{E}\left[\langle \psi^{1}_{M(t),M_{0}(t),\kappa(t)}(t)|\psi^{1}_{M(t),M_{0}(t),\kappa(t)}(t)\rangle \right]=0
\end{align}
\end{proof}
%%%%%%%%%%%%%%%%%%%%%%%%%%%%%%%%%%%%%%%%%%%%%%%%%%%%%%%%%%%%%%%%%%%%%%%%%%%%
%%%%%%%%%%%%%%%%%%%%%%%%%%%%%%%%%%%%%%%%%%%%%%%%%%%%%%%%%%%%%%%%%%%%%%%%%%%%
%%%%%%%%%%%%%%%%%%%%%%%%%%%%%%%%%%%%%%%%%%%%%%%%%%%%%%%%%%%%%%%%%%%%%%%%
%%%%%%%%%%%%%%%%%%%%%%%%%%%%%%%%%%%%%%%%%%%%%%%%%%%%%%%%%%%%%%%%%%%%%%%%
\appendix
\section{Integrals}
In this section we will prove and state some useful bounds.
The following integral inequalities will be used:
\begin{align}
  \left|\int_{0}^{1} d\omega \left(\frac{-1}{\omega-\alpha -i\eta}\right)^{k}\right| & \leq  \left|\left(\frac{1}{1-\alpha -i\eta}\right)^{k-1}\right|+\left|\left(\frac{1}{-\alpha -i\eta}\right)^{k-1}\right| \label{ine:one}
\end{align}
\begin{align}
 \int_{0}^{1} d\omega \left|\frac{-1}{\omega-\alpha -i\eta}\right|^{k} 
\left|\frac{-1}{\omega-\beta +i\eta}\right|^{p} & \leq  
 \frac{1}{\eta^{k+p-2}}\int_{0}^{1} d\omega \left|\frac{-1}{\omega-\alpha -i\eta}\right| \left|\frac{-1}{\omega-\beta +i\eta}\right| \notag \\
& \leq  
\frac{1}{\eta^{k+p-2}}\left(\int_{0}^{1} d\omega \left|\frac{-1}{\omega-\alpha -i\eta}\right|^{2}\right)^{\frac{1}{2}} \left( \left|\frac{-1}{\omega-\beta +i\eta}\right|^{2}\right)^{\frac{1}{2}}\notag \\
 & \leq  \frac{1}{\eta^{k+p-1}}\label{ine:two}
\end{align}
\begin{align}
 \int_{-C}^{C}d \omega \left|\frac{1}{\omega-\alpha-i\eta}\right|^{\delta}\leq C_{\delta}\frac{1}{\eta^{\delta-1}}\label{eq:oneminusa}
\end{align}
With $\delta>1$.
We set $x_{0}=\frac{x-y}{2}$ and $z=x_{0}-i\eta$.
\begin{align}
 \int_{-C}^{C} d\alpha \frac{1}{\left|x-\alpha-i\eta\right|\left|y-\alpha-i\eta\right|}&=
 \int_{-C+\frac{x+y}{2}}^{C+\frac{x+y}{2}} d\alpha \frac{1}{\left|x_{0}-\alpha-i\eta\right|\left|-x_{0}-\alpha-i\eta\right|} \notag\\
&=
 \int_{-C+\frac{x+y}{2}=a}^{C+\frac{x+y}{2}=b} d\alpha \frac{1}{\left|z-\alpha\right|\left|\bar{z}+\alpha\right|} \notag\\
&\leq 
 \int_{a}^{0} d\alpha \frac{1}{\left|z\right|\left|\bar{z}+\alpha\right|} +
 \int_{0}^{b} d\alpha \frac{1}{\left|\bar{z}\right|\left|z-\alpha\right|} 
\notag\\
&\leq 
C_{1} \frac{1}{\left|x-y-i\eta\right|}\left|\log \eta\right|
\label{eq:ABlog}
\end{align}
In the same manner we can bound
\begin{align}
\int_{-C}^{C} d\alpha  \frac{1}{\left|x-\alpha-i\eta\right|^{\delta}\left|y-\alpha-i\eta\right|^{\delta}}&\leq C_{1}\frac{1}{\left|x-y-i\eta\right|^{\delta}}\label{eq:ABdelta}
\end{align}
where $\delta<1$.
We can apply this to the following integrals
\begin{align}
&\int_{-C}^{C} d\beta d \alpha\int_{0}^{1} d \omega_{1} d \omega_{2} 
\left|\frac{1}{\omega_{1}-\alpha -i\eta}\right| 
\left|\frac{1}{\omega_{1}-\beta +i\eta}\right|
\left|\frac{1}{\omega_{2}-\alpha -i\eta}\right| 
\left|\frac{1}{\omega_{2}-\beta +i\eta}\right|  \left|\frac{1}{x-\alpha -i\eta}\right|^{k}
\notag \\
\leq &\int_{-C}^{C}  d \alpha\int_{0}^{1} d \omega_{1} d \omega_{2} 
\left|\frac{1}{\omega_{1}-\alpha -i\eta}\right| 
\left|\frac{1}{\omega_{1}-\omega_{2} -i\eta}\right|
\left|\frac{1}{\omega_{2}-\alpha -i\eta}\right| 
 \left|\frac{1}{x-\alpha -i\eta}\right|^{k}C_{1}\left|\log \eta\right|
\notag \\
\leq &\int_{-C}^{C}  d \alpha\int_{0}^{1}  d \omega_{2} 
\left|\frac{1}{\omega_{2}-\alpha -i\eta}\right| 
\left|\frac{1}{\omega_{2}-\alpha -i\eta}\right| 
\left|\frac{1}{x-\alpha -i\eta}\right|^{k}C_{1}\left|\log \eta\right|^2
\notag \\
\leq &\int_{-C}^{C}  d \alpha
\left|\frac{1}{x-\alpha -i\eta}\right|^{k}\frac{C_{1}\left|\log \eta\right|^2}{\eta}
\notag \\
\leq &  
\frac{C_{1}\left|\log \eta\right|^2}{\eta^{k}}
\label{eq:4kintegrals}
\end{align}
We now analyze
\begin{align}
\left|\Theta(\alpha,\eta)-\Theta(\omega_{\bar{n}})\right|=&\left|\int_{0}^{1}d\omega \frac{1}{\omega-\alpha-i\eta} -\lim_{\eta'\rightarrow 0}\int_{0}^{1}d\omega \frac{1}{\omega-\omega_{\bar{n}}-i\eta'}\right| \notag \\
\leq & \lim_{\eta'\rightarrow 0}
\left(\left|\log \left(1-\alpha-i\eta\right) -\log \left(1-\omega_{\bar{n}}-i\eta'\right)  \right|+
\left|\log \left(-\alpha-i\eta\right) -\log \left(-\omega_{\bar{n}}-i\eta'\right)  \right|\right) \notag \\
\leq & \lim_{\eta'\rightarrow 0}
\left|\omega_{\bar{n}}-\alpha-i(\eta-\eta')\right| \left(\frac{1}{\left|1-\alpha-i\eta\right|}+\frac{1}{\left|1-\omega-i\eta'\right|}+
\frac{1}{\left|-\alpha-i\eta\right|}+\frac{1}{\left|-\omega-i\eta'\right|}\right) \notag \\
\leq &
\left|\omega_{\bar{n}}-\alpha-i\eta\right| \left(\frac{1}{\left|1-\alpha-i\eta\right|}+\frac{1}{\left|1-\omega\right|}+
\frac{1}{\left|\alpha+i\eta\right|}+\frac{1}{\left|\omega\right|}\right) \label{eq:thetabound} 
\end{align}
% BibTeX users please use one of
%\bibliographystyle{spbasic}      % basic style, author-year citations
%\bibliographystyle{spmpsci}      % mathematics and physical sciences
%\bibliographystyle{spphys}       % APS-like style for physics
\bibliographystyle{abbrv}
\bibliography{paper}   % name your BibTeX data base
% Non-BibTeX users please use
%\begin{thebibliography}{}
%
% and use \bibitem to create references. Consult the Instructions
% for authors for reference list style.
%
%\bibitem{RefJ}
% Format for Journal Reference
%Author, Article title, Journal, Volume, page numbers (year)
% Format for books
%\bibitem{RefB}
%Author, Book title, page numbers. Publisher, place (year)
% etc
%\end{thebibliography}
\end{document}